\documentclass[11pt,letterpaper]{article}

\usepackage[utf8]{inputenc}
\usepackage{latexsym}
\usepackage{amsmath}
\usepackage{amssymb}
\usepackage{amsthm}
\usepackage{hyperref}
\usepackage{cite}
\usepackage{graphicx}
\usepackage{color}
\usepackage{enumitem}
\usepackage{comment}
\usepackage{bm}
\usepackage[left=1in,right=1in,top=1in,bottom=1in]{geometry}
\usepackage{xspace}
\setlist[description]{font=\normalfont\itshape\textbullet\space}
\usepackage[all]{xy}
\usepackage{doi, url}
\usepackage{arydshln}
\usepackage{xcolor} 
\usepackage{framed}
\usepackage{physics}
\usepackage{tikz}
\usepackage{tikz-cd}
\usepackage{mathdots}
\usepackage{yhmath}
\usepackage{cancel}
\usepackage{siunitx}
\usepackage{array}
\usepackage{multirow}
\usepackage{gensymb}
\usepackage{tabularx}
\usepackage{extarrows}
\usepackage{booktabs}
\usetikzlibrary{fadings}
\usetikzlibrary{patterns}
\usetikzlibrary{shadows.blur}
\usetikzlibrary{shapes}

\renewcommand{\paragraph}[1]{\vspace{6pt} \noindent \textbf{#1}\xspace}

\theoremstyle{plain}
\newtheorem{theorem}{Theorem}[section]

\newtheorem{corollary}[theorem]{Corollary}
\newtheorem{lemma}[theorem]{Lemma}

\newtheorem{fact}[theorem]{Fact}

\theoremstyle{definition}
\newtheorem{problem}[theorem]{Problem}

\newtheorem{definition}[theorem]{Definition}
\newtheorem{example}[theorem]{Example}

\newcommand{\GL}{\mathrm{GL}}
\newcommand{\F}{\mathbb{F}}
\newcommand{\Z}{\mathbb{Z}}

\newcommand{\poly}{\mathrm{poly}}

\newcommand{\M}{\mathrm{M}}

\renewcommand{\S}{\mathrm{S}}

\newcommand{\aut}{\mathrm{Aut}}
\newcommand{\Hom}{\mathrm{Hom}}

\newcommand{\linspan}{\mathrm{span}}

\newcommand{\algprobm}[1]{\text{\sf #1}\xspace}
\newcommand{\GI}{\algprobm{GI}}

\newcommand{\GIlong}{\algprobm{Graph Isomorphism}}
\newcommand{\GpI}{\algprobm{GpI}}
\newcommand{\GpIlong}{\algprobm{Group Isomorphism}}

\newcommand{\TI}{\algprobm{TI}}
\newcommand{\TIlong}{\algprobm{Tensor Isomorphism}}

\newcommand{\cc}[1]{\mathsf{#1}}

\newcommand{\ynote}[1]{}
\newcommand{\gnote}[1]{}
\newcommand{\jnote}[1]{}
\newcommand{\xnote}[1]{}

\DeclareMathOperator{\Rad}{Rad}

\DeclareMathOperator{\id}{id}
\DeclareMathOperator{\diag}{diag}

\renewcommand{\le}{\leqslant}

\newcommand{\too}%
{\xrightarrow{\text{\raisebox{-3pt}{$\sim$}}\,}}

\mathchardef\mhyphen="2D

\usepackage[T1]{fontenc}
\def\DJ{{\hbox{D\kern-.8em\raise.15ex\hbox{--}\kern.35em}}}

\newcommand{\ignore}[1]{}

\title{Beyond odd characteristic: Faster isomorphism testing of 2-groups of Frattini class 2}

\author{
Joshua A. Grochow\thanks{Departments of Computer Science and Mathematics, University of Colorado Boulder, Boulder, USA. Email: {\tt jgrochow@colorado.edu}. Supported by NSF CAREER
grant CCF-2047756.}
\and
G\'abor Ivanyos\thanks{HUN-REN Institute for computer Science and Control, Budapest, Hungary.
E-mail: {\tt Gabor.Ivanyos@sztaki.hun-ren.hu}.
Supported by the project 2026-4.1.1-MISSZIÓK-2026-00011 AI Applications in Healthcare and Industry Mission-Driven National Laboratory.}
\and
Youming Qiao\thanks{Centre for Quantum Software and Information, University of Technology Sydney, Sydney, Australia. Email:
{\tt Youming.Qiao@uts.edu.au}. Partly supported by Australian Research Council LP220100332.}
\and
Xiaorui Sun\thanks{Computer Science Department, University of Illinois Chicago, Chicago, USA. Email:
{\tt xiaorui@uic.edu}. Supported by the National Science Foundation (NSF)
under Grant No. 2240024.
}
}
\date{}

\begin{document}
\maketitle
	\pagenumbering{gobble}  
\begin{abstract}
The finite group isomorphism problem asks whether two finite groups of order $N$ are isomorphic. The first algorithm, attributed to Tarjan (see Miller, \emph{STOC} '78), runs in time $N^{\log N + O(1)}$. Despite intensive study, the current best known algorithm has a running time of $N^{(1 / 4 + o(1))\log N}$ (Rosenbaum, '13).

$p$-groups of class $2$ have been recognized as the major bottleneck for faster group isomorphism. Recent progress has led to $N^{o(\log N)}$-time algorithms for $p$-groups of class $2$ where $p$ is odd (Sun, \emph{STOC} '23; Ivanyos--Mendoza--Qiao--Sun--Zhang, \emph{FOCS} '24; Grochow--Qiao--Stange--Sun, \emph{STOC} '25). However, the case of $p=2$, which represents the majority of $p$-groups of class 2 assuming a well-known conjecture in group enumeration, remained elusive, with essentially no progress until now.

In this paper, we present an algorithm for testing the isomorphism of two 2-groups of Frattini class 2 of order $N$ in time $N^{O((\log N)^{1/2})}$. To our knowledge, this is the first $N^{o(\log N)}$-time isomorphism algorithm for a class of $2$-groups that constitutes logarithmically almost all $2$-groups, in the sense that  $\lim_{N \to \infty} \frac{\log(\text{\# 2-groups of Frattini class 2 and order } \leq N)}{\log(\text{\# 2-groups of order} \leq N)} = 1$.

As our main tool, we present the first non-trivial algorithms for the quadratic form space/tuple isometry problems over $\mathbb{F}_2$. These algorithms rely on combinations of combinatorial and algebraic ideas, including finite matrix group algorithms developed by Luks (\emph{FOCS} '92). As far as we know, this is the first time that matrix group algorithms are used to make progress on the worst-case complexity of $p$-group isomorphism.

\end{abstract}
\newpage

\tableofcontents

\newpage

\pagenumbering{arabic}
\setcounter{page}{1}
\pagenumbering{arabic}
\section{Introduction}

Isomorphism problems---testing whether two mathematical objects are the same up to some naturally defined relabeling operations---have been a part of computational complexity since its inception \cite{Cook1971,Levin1972} (see \cite{AllenderDas} for more on the history). The \GIlong problem (\GI) has played a central role, for how natural it is, how basic the objects are (just graphs!), and how universal it is: the isomorphism problem for any class of explicitly given first-order structures reduces to GI \cite{ZKT}. 

The more-algebraic finite \GpIlong problem (\GpI) has been open for nearly as long, spanning across the communities of both computational complexity theory and computational group theory. In computer algebra systems, realistic input models for groups tend to be succinct---such as giving generating matrices, permutations, or a presentation by generators and relations. In the presentation model, \GpI becomes uncomputable \cite{Adian1955, Rabin1958}. In the former two (permutation and matrices) models, although computable, it is \GI-hard \cite{luks-dimacs,TI2}. Despite decades of research from both communities, the worst-case complexity of \GpI---whether given by generating matrices, permutations, or explicitly listing out its Cayley (multiplication) table---is not known to be better than $|G|^{O(\log |G|)}$. In the more realistic, succinct input models, this is exponential in the input size. But in the much more verbose Cayley table model, this runtime makes \GpI appear easier than \GI. 

Indeed, \GpI in the Cayley table model reduces to \GI \cite{ZKT}, and there is no $\cc{AC}^0$ reduction in the opposite direction \cite{CTW}. However, with the best-known time complexity being so close to one another---$n^{\Theta(\log n)}$ for \GpI\cite{FN,Mil78,WilsonProfile,Ros13,LeGallRosenbaum} 
\jnote{Maybe not BCGQ? By WilsonConferenceThing I actually meant p. 2 of this paper \url{http://dx.doi.org/10.4086/toc.2019.v015a019}, where James explicitly says that even all modern practical isomorphism algorithms are known to take $n^{\log n}$ time.} 
\ynote{Changed \cite{BCGQ11} to \cite{WilsonProfile}.}
versus $n^{\Theta(\log^2 n)}$ for \GI \cite{babai2016graph}\footnote{The exact exponent of the exponent was worked out by Helfgott \cite{helfgott}.}---\GpI stands as a significant bottleneck to further progress on \GI, in addition to being of interest in its own right.

In this paper, we focus on improving the worst-case runtime of \GpI for a class of groups that are some of the hardest cases: 2-groups of Frattini class 2. The 2-groups of Frattini class 2 are groups whose order is a power of $2$, and that have a normal subgroup $N$ contained in the center such that both $N$ and $G / N$ are elementary abelian. We briefly explain why this family of groups is significant for group isomorphism testing before stating our results. 


From the group-theoretic side, 2-groups (whose order is a power of 2) account for most isomorphism classes of finite groups.
By \cite{besche2002millennium}, there are 11,759,892 many groups of order $\leq 1023$, while there are 49,487,367,289 groups of order $1024$. Also, the smallest pair of non-isomorphic non-abelian groups of the same order consists of $D_8$, the dihedral group of order $8$, and $Q_8$, the quaternion group. 
Such an abundance of $2$-groups suggests that the family of $2$-groups is a culprit for the hardness of \GpI. 
Among $2$-groups, $2$-groups of Frattini class $2$ form a basis for many isomorphism classes; indeed, $D_8$ and $Q_8$ are both in this class. As shown by Higman \cite{HigmanEnum} and Sims \cite{sims1965enumerating}, 
$$\lim_{N \to \infty} \frac{\log(\text{\# 2-groups of Frattini class 2 and order } \leq N)}{\log(\text{\# 2-groups of order} \leq N)} = 1.$$
That is, the number of $2$-groups of Frattini class $2$ matches the number of general $2$-groups on the logarithmic scale. This is a strong indication that to make progress on the isomorphism testing of $2$-groups, the family of $2$-groups of Frattini class $2$ is the starting point. 

Despite the importance of 2-groups, 
there has not been much progress in their isomorphism testing.
In the odd $p$ case,
$|G|^{o(|G|)}$ time algorithms have been proposed for $p$-groups of class $2$ and exponent $p$ \cite{Sun23}, $p$-groups of Frattini class $2$ \cite{IMQSZ}, and $p$-groups of class 2 \cite{GQSS25}. None of these results apply to $p=2$. Although \cite{Sun23,IMQSZ,GQSS25} pave the way for our work here, significant bottlenecks must be overcome to achieve such a result for $2$-groups, as we will explain later. 

We now introduce our results. Because we focus on worst-case runtime, and our results will improve from $|G|^{\Theta(\log |G|)}$ to $|G|^{\Theta(\log|G|)^{1-\varepsilon}}$ for some fixed $\varepsilon > 0$, we state our results in the Cayley table model. We note, however, that this is not necessary---all our algorithms also work when the group is given as a matrix group or permutation group, but since the runtime depends on $|G|$ anyway, it is a formal convenience to state the results in the Cayley table model.

Our main theorem for \GpI is as follows:

\begin{theorem}\label{thm:group_iso}
There is an algorithm to determine the isomorphism of 2-groups of Frattini class 2 in time $|G|^{O((\log |G|)^{1/2})}$. 
\end{theorem}

The Frattini class $2$ condition allows for a translation from groups to polynomials. 
Recall that a quadratic form is a homogeneous degree-$2$ polynomial. If $f(x_1,\dotsc,x_n)$ is a quadratic form and $S$ is an invertible linear transformation of the variables $\vec{x}$, we use $f^S$ to denote the quadratic form $f^S(\vec{x}) := f(S \vec{x})$. 
 Our key algorithm 
is in our second main theorem.

\begin{theorem}\label{thm:quadratic_space_isometry}
Given two sequences of quadratic forms $f_1, \dots, f_m$ and $g_1, \dots, g_m$ in $n$ variables over $\F_2$,
    there is an algorithm to determine whether there is an $n\times n$ invertible matrix $S$ over $\F_2$ such that the linear space spanned by $ f_1^S, \dots, f_m^S$ equals the linear space spanned by $g_1, \dots, g_m$ with running time $2^{O((m + n)^{1.5})}$. 
\end{theorem}

We note, for comparison, that the runtime of the brute-force algorithm here is $2^{O(m+n)^2}$, which would correspond to a group isomorphism algorithm for the class of groups above in time $|G|^{O(\log |G|)}$, not better than the essentially trivial bound.

To solve the above problem, we reduce from a problem on \emph{spaces} of quadratic forms to the corresponding (but simpler) problem on \emph{tuples} of quadratic forms (where we are no longer allowed to take linear combinations of the forms themselves, only of the variables). 
Solving the latter problem in the stated time bound is then the key subroutine, which is essentially\footnote{Technically, the key subroutine is Thm.~\ref{thm:matrix_tuple_constructed_main}, which then also gets used in the proof of Thm.~\ref{thm:quadratictuple_isometry_intro}. We state Thm.~\ref{thm:quadratictuple_isometry_intro} here instead because it is a simpler statement that, at a high level, has the same flavor. Namely, we reduce from a problem on \emph{spaces} of quadratic forms to \emph{tuples} of quadratic forms, and then solve the latter problem in the stated time bound.} encapsulated in:


\begin{theorem}\label{thm:quadratictuple_isometry_intro}
Given two sequences of quadratic forms $f_1, \dots, f_m$ and $g_1, \dots, g_m$ in $n$ variables over $\F_2$,
    there is an algorithm with running time 
    $2^{O((m + n)^{1.5})}$
    to determine whether there is an invertible matrix $S$ such that $f_i^S = g_i$ for all $1 \leq i \leq m$.
\end{theorem}



In Section~\ref{sec:intro:Frattini} we discuss both the significance of Frattini class 2 and the difficulties that specifically make the $p=2$ case one of the hardest. In Section~\ref{sec:intro:techniques} we then discuss our methods.

\subsection{The significance of 2-groups of Frattini class 2} \label{sec:intro:Frattini}
\paragraph{From general \GpI to $2$-groups.} Nilpotent groups have long been regarded as the most difficult cases of \GpI; a finite group $G$ is nilpotent if and only if it is the direct product of groups of prime power order. As direct products of groups can be efficiently decomposed into their direct factors \cite{WilsonProductsI, WilsonProductsII, KN}, this reduces from nilpotent groups to groups of prime power order (called $p$-groups if the order is a power of $p$). We note that log-asymptotically, most groups are $p$-groups, in the sense that $$\lim_{n \to \infty} \frac{\log(\text{\# $p$-groups of order } \leq n)}{\log(\text{\# all groups of order} \leq n)} = 1$$
as proved by Pyber \cite{pyber1993enumerating}. 
It is a well-known, long open conjecture that this is true without the logs; see, e.g., \cite{mann1999some}. From this, because $p=2$ is the smallest prime, the $2$-groups of order $2^k$ occur at smaller sizes than the $p$-groups of order $p^k$ for any other $p$, giving a sense in which ``most groups are $2$-groups.''

\paragraph{Difficulties specific to $p=2$.} There are several obstacles that are particular to the $p=2$ case, some of which are perhaps why it has resisted progress despite significant recent progress in the case of odd $p$ \cite{Sun23, IMQSZ, GQSS25} (which we discuss more below). 

For groups of Frattini class 2, classical results from more than sixty years ago \cite{HigmanEnum} 
can be made explicit to show that the isomorphism problem reduces to multilinear-algebraic isomorphism/equivalence problems on spaces of alternating matrices. The first obstacle in the case of $p=2$ is that, in the case of odd $p$, the alternating matrix spaces have dimension bounded by $\log|G|$, but in the case of $p=2$ the alternating matrix space may have dimension as large as $O(\log |G|)^2$. The difficulty here is inherently about the more-complicated relationship between symmetric bilinear forms, anti-symmetric (or alternating) bilinear forms, symmetric matrices, and quadratic forms in characteristic 2, compared to odd or $0$ characteristic. As a simple example, when $p$ is odd (or $0$), symmetric bilinear forms $f(\vec{x},\vec{y})$ can always be represented by a symmetric matrix $A$ such that $f(\vec{x}, \vec{y}) = \vec{x}^T A \vec{y}$. 
However, such a representation involves dividing by $2$, and is not possible in characteristic 2. Rather, in characteristic 2, we find that symmetric matrices more closely correspond to \emph{linear functions on} the space of quadratic forms, and vice versa (they are dual vector spaces). Thus, instead of working with matrix spaces as in the case of odd $p$ \cite{Sun23, IMQSZ, GQSS25}, for $p=2$ we seem essentially forced to deal with spaces of quadratic forms. 

A second obstacle has to do with the difficulty of a problem at the technical core of prior algorithms: the isometry problem on tuples. 
The recent results on odd $p$ \cite{Sun23, IMQSZ, GQSS25} all reduce the group isomorphism problem for specific $p$-groups to the alternating matrix tuple isometry problem, or equivalently, given two sequences of skew-symmetric matrix tuples $\mathbf{A} = (A_1, \dots, A_m)$ and $\mathbf{B} = (B_1, \dots, B_m)$ of the same dimension $n$, determine whether the two tuples are \emph{congruent}, i.\,e., determine whether 
there is an invertible matrix $S$ such that $S^T A_i S = B_i$ for all $1 \leq i \leq m$.  
The alternating matrix tuple isometry problem in odd characteristic was solved in~\cite{IQ19}. However, the approach of~\cite{IQ19} does not work for characteristic 2. 
Specifically, there is a ``lifting'' procedure at the core of the algorithm of \cite{IQ19, GQSS25} that first solves the problem in a ``quotient space'' that satisfies a strong algebraic constraint, then lifts the solution back to the original spaces. But the lifting procedure inherently involves dividing by 2, and can fail badly in the case of $p=2$ (there can be solutions to the quotient problem that have no lifts to the original problem when $p=2$).
Previously, only a trivial algorithm was known for this problem in characteristic 2 in the general case 
\footnote{There is also a different algorithm for quadratic form tuple equivalence in \cite{berthomieu2015polynomial}, but that algorithm only works under a regular condition. For example, because of the regular condition, their algorithm does not work with even number of variables. Indeed, an open question in \cite{berthomieu2015polynomial} is about the general case of quadratic form tuple isometry.}, which is insufficient for improving group isomorphism algorithms. 
Overcoming this obstacle is one of our main technical contributions.

Finally, a third obstacle in the case of $p=2$ is that many properties of matrices and matrix tuples over odd or 0 characteristic, which are very useful for designing faster algorithms, do not hold in characteristic 2. 
One particular example is the block diagonalization, which decomposes matrices or matrix tuples into minimum blocks, under the congruence action. 
It is known that in odd and 0 characteristic, the minimum block diagonalization under the congruence action for symmetric/alternating matrices and matrix tuples is unique (up to ``obvious'' transformations like permuting the blocks, similar to uniqueness of prime factorization). 
But in characteristic 2, it is well-known that the minimum block diagonalization for symmetric matrices is not unique. 
For alternating matrix tuples, we also give an example of non-unique congruent block diagonalization in Appendix~\ref{sec:appendix}. 
In this paper, we develop new theory of congruent block diagonalization for matrix tuples (Theorem~\ref{thm:congruence_block_diagonalization}), which is a key part of our algorithms.

\paragraph{Definition and significance of Frattini class 2.} Among 2-groups, those of Frattini class 2 form a key, initial stepping stone. A 2-group $G$ has \emph{Frattini class 2} if it has  a normal subgroup $N$ contained in the center (all elements of $N$ commute with all elements of $G$) such that both $N$ and $G/N$ are elementary abelian: $N \cong \Z_2^k$ and $G/N \cong \Z_2^\ell$ for some $k,\ell$. 

One crude, initial measure of the structural complexity of a nilpotent group is its \emph{nilpotency class}, defined as follows. Recall the center of a group $G$ is $Z(G) := \{g \in G : (\forall h \in G) gh = hg\}$; we then define the upper central series as $Z^1(G) := Z(G)$ and $Z^{k+1}(G) := \{g \in G : (\forall h \in G) gh \equiv hg \pmod{Z^k(G)}\}$. The nilpotency class of $G$ is the least $k$ such that $Z^k(G) = G$. Nilpotency class 1 thus coincides with Abelian groups, and the largest nilpotency class of a $p$-group is $\log_p|G|$. Groups of maximal class $\log_p|G|$ exist, but groups whose nilpotency class is $\log_p|G|-O(1)$---which admit a beautiful theory, e.\,g., \cite{LGM}---are also generated by $O(1)$ elements,\footnote{Quick proof: If $G$ is a $p$-group, then the Frattini subgroup is $\Phi(G) = G^p [G,G]$, $G/\Phi(G) \cong \Z_p^d$, and $G$ is generated by $d$ elements. If the group has class $\log_p|G|-O(1)$, then each factor of the lower central series has order at least $p$, and since there are $\log_p|G|-O(1)$ such factors, each factor can have order at most $p^{O(1)}$. Thus $G/\Phi(G) \cong \Z_p^d$ where $d = O(1)$.} and isomorphism of $O(1)$-generated groups can be tested in polynomial time \cite{FN,Mil78}. Class 2 $p$-groups are expected to be some of the hardest cases for isomorphism testing, as they are somewhat at the ``opposite end of the spectrum'' of groups of near-maximal class, without being Abelian (testing Abelian groups for isomorphism can be done in linear time in the group order \cite{Iliopoulos1985AnalysisAbelian,Kavitha2007LinearAbelianIso}). Just as $p$-groups are log-asymptotically most groups, the same is true even for $p$-groups of nilpotency class 2: $\lim_{n \to \infty} \frac{\log(\text{\# $p$-groups of class 2 and order } \leq n)}{\log(\text{\# all groups of order} \leq n)} = 1$ \cite{HigmanEnum,sims1965enumerating,pyber1993enumerating}.

Furthermore, class 2 $p$-groups provide important structural information about general nilpotent groups. If one merely examines the layers of the upper central series $Z^{k+1}(G) / Z^k(G)$, one gets a series of Abelian groups, which does not tell us much about the structure of the group in general: there are many pairs of non-isomorphic groups where the layers of the upper central series are all isomorphic. For example, when $n$ is a power of $p$, there are $n^{\Theta(\log ^2 n)}$ $p$-groups of order $n$, class $2$, and exponent $p$ where $Z(G) = \Z_p^{(1/3)(\log n)}$ and $G/Z(G) = \Z_p^{(2/3)\log n}$. However, if one examines the quotients two layers apart, each two-layer quotient $Z^{k+2}(G) / Z^k(G)$ is nilpotent of class 2. If two groups have all their two-layer quotients of their upper central series isomorphic, it is intuitively very hard (but not impossible) for the two groups to be non-isomorphic. Such information was incorporated into a strategy for group isomorphism testing in, e.g., \cite{BGLQW}. 

The set of groups of Frattini class 2 is a large subset of those of nilpotency class 2. They furthermore have exponent at most $p^2$ (the \emph{exponent} of a group is the least integer $e$ such that $g^e=1$ for all $g \in G$); the latter is the first interesting case for 2-groups, since groups of exponent 2 are Abelian (easy exercise). Extending from Frattini class 2 to nilpotency class 2 will require dealing with quadratic forms over non-fields such as $\Z/2^k \Z$, and though we expect our methods to be useful there, we also expect that significant additional work is needed to handle such cases.

In addition to 2-groups of Frattini class 2 being  expected to be hard cases of \GpI, their isomorphism problem also enjoys a sort of ``universality'' among a large class of multilinear algebraic problems. 2-groups of Frattini class 2 are closely related to tensors over finite fields: groups in this class can be specified more succinctly by vector space bases and (multi)linear algebraic data, and in that setting the isomorphism problem for the groups we study is \TI-hard. The \TIlong problem over finite fields is connected to several other areas including cryptography \cite{RST, NQT, TI1} and number theory \cite{GQSS25} (tensors over other fields are connected to even more areas). As with \GI, the \TIlong problem also enjoys a universality among ``degree-3 linear-algebraic isomorphism problems'' \cite{FGS19} (see discussion in \cite[Sec.~7.1]{TI1}).\footnote{Although $d$-\TI reduces to 3-\TI \cite{TI1}, it is an open question whether the stronger universality results holds for linear-algebraic isomorphism problems of arbitrary degree, that is, it is open whether all linear-algebraic isomorphism problems of degree $d$ reduce to $d$-\TI.} 


To summarize, here is a diagram putting our work in the context of prior work and the more general problem. For anyone skimming, we highlight that the $p=2$ case is not merely ``a leftover add-on'' to the odd $p$ case, but presents significant additional obstacles that we discussed above.

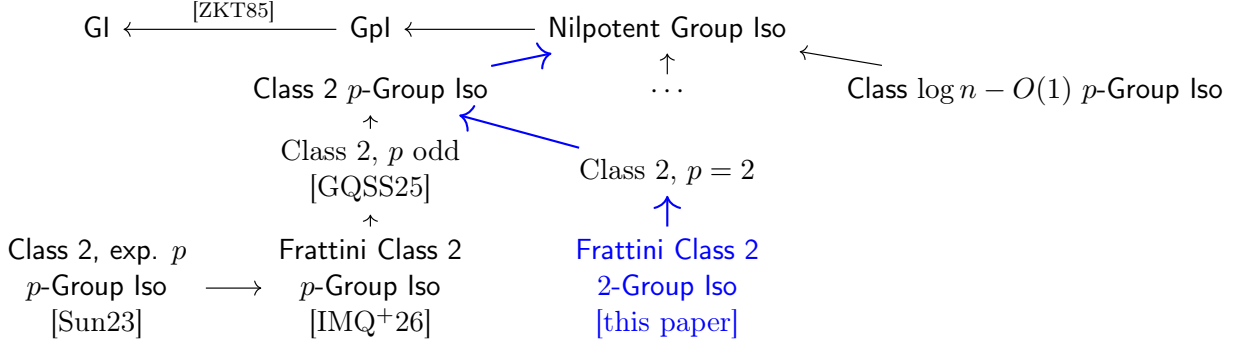
\begin{center}
\begin{tikzcd}[row sep=tiny, column sep = small]
\GI & \GpI \arrow{l}[above]{\text{\cite{ZKT}}} &  \algprobm{Nilpotent Group Iso} \arrow{l} \\
& \algprobm{Class 2 $p$-Group Iso} \arrow[blue, thick]{ru} & \dotsb \arrow{u}& \algprobm{Class $\log n - O(1)$ $p$-Group Iso} \arrow{lu}
\\
 & \parbox{1in}{\centering \text{Class 2, $p$ odd} \cite{GQSS25}} \arrow{u} & \text{Class 2, $p=2$} \arrow[blue,thick]{lu} \\
 \parbox{1in}{\centering \algprobm{Class 2, exp. $p$} \algprobm{$p$\text{-}Group Iso} \cite{Sun23}} \arrow{r} & \parbox{1in}{\centering \algprobm{Frattini Class 2} \algprobm{$p$-Group Iso} \cite{IMQSZ}} \arrow{u} & \parbox{1in}{\centering {\color{blue} \algprobm{Frattini Class 2} \algprobm{$2$-Group Iso} [this paper]}}\arrow[blue,thick]{u}
\end{tikzcd}
\end{center}

\subsection{Our techniques} \label{sec:intro:techniques}

We give an overview of our techniques for the main results in this section. 
\subsubsection{Reduce isomorphism testing of 2-groups of Frattini class 2 to quadratic form space isometry.} We first reduce the isomorphism testing of 2-groups of Frattini class 2 to the quadratic form space isometry problem, i.e., given two sequences of homogeneous quadratic forms $f_1, \dots, f_m$ and $g_1, \dots, g_m$, determine whether there exists an invertible matrix $S$ such that 
$\linspan\{f_1^S, \dots, f_1^S \} = \linspan\{g_1, \dots, g_m\}$. 

In \cite{HigmanEnum}, Higman developed a criterion for testing isomorphism of $p$-groups of Frattini class $2$. 
Let $G$ and $H$ be a $p$-group of Frattini class $2$. 
By the Frattini class $2$ condition, $\Phi(G)$ and $G/\Phi(G)$ are both elementary abelian $p$-groups, where $\Phi(G)$ is the intersection of maximal subgroups of $G$. Let $F_{\Phi\text{-}2, p, n}$ be the relatively free $p$-group of Frattini class $2$ with $n$ generators, that is the quotient of the free group with $n$ generators by the subgroup generated by relations of the form $a^{p^2}$, $[a,b]^p$, and $[a, [b,c]]$. Then $G$ can be identified as the quotient of $F_{\Phi\text{-}2, p, n}$ by a subgroup $S_G\leq \Phi(F_{\Phi\text{-}2, p, n})$. For the other group $H$, we obtain $S_H\leq \Phi(F_{\Phi\text{-}2, p, n})$. To test whether $G$ and $H$ are isomorphic then boils down to test whether $S_G$ and $S_H$ are in the same orbit (as subgroups in $\Phi(F_{\Phi\text{-}2, p, n})$) under the induced action of $\aut(F_{\Phi\text{-}2, p, n}/\Phi(F_{\Phi\text{-}2, p, n}))\cong \GL(n, p)$. 

The action of  $\aut(F_{\Phi\text{-}2, p, n}/\Phi(F_{\Phi\text{-}2, p, n}))$ on $\Phi(F_{\Phi\text{-}2, p, n})$ for odd $p$ was explicitly calculated by Higman \cite[Theorem 2.2]{HigmanEnum}. For $p=2$, this is more complicated. Recall that $\Phi(F_{\Phi\text{-}2, p, n})=F_{\Phi\text{-}2, p, n}^p[F_{\Phi\text{-}2, p, n},F_{\Phi\text{-}2, p, n}]$. In the case of odd $p$, the actions of $\aut(F_{\Phi\text{-}2, p, n}/\Phi(F_{\Phi\text{-}2, p, n}))$ on $F_{\Phi\text{-}2, p, n}^p$ and $[F_{\Phi\text{-}2, p, n},F_{\Phi\text{-}2, p, n}]$ are independent. For $p=2$, this is no longer the case, so some care is required. Fortunately, it turns out that this action is equivalent to the natural action of $\GL(n, 2)$ on $\S(n, 2)$, the linear space of $n\times n$ symmetric matrices over $\F_2$ (Lemma~\ref{lem:induced-action}). 

Given this, we still need one more step, as the dimension of $S_G$ may be large relative to the order of $G$. This can be resolved by going to the dual action of $\GL(n, 2)$ on $\S(n, 2)$. Unlike the odd $p$ case, the dual action of $\GL(n, 2)$ on $\S(n, 2)$ is no longer equivalent to itself, but equivalent to the action of $\GL(n, 2)$ on $Q(n, 2)$, the linear space of quadratic forms in $n$ variables over $\F_2$ (Fact~\ref{fact:dual-action}). After this step, we complete the reduction from isomorphism testing of $2$-groups of Frattini class $2$ to the quadratic form space isometry problem. 

\subsubsection{Reduce quadratic form space isometry to quadratic form tuple isometry.} \label{sec:into_overview_2}
We show that quadratic forms and alternating matrices are closely related over $\F_2$. 
For a quadratic form $f$ over variables $x_1, \dots, x_n$, let $D_f$ be its partial derivative matrix, in which the $(i, j)$ entry and the $(j, i)$ entry are the coefficient of $x_i x_j$ for all $i < j$, and the diagonal entries are zero.
For an arbitrary invertible matrix $S$, the partial derivative matrix of $f^S$ is $S^T D_f S$. 
Furthermore, this relation is also additive, i.e., for any two quadratic forms $f$ and $g$ defined on the same variables, we have $D_{(f + g)^S} = D_{f^S} + D_{g^S}$. 

In order to determine whether there is an invertible matrix $S$ such that $\linspan\{f_1^S, \dots, f_m^S\} = \linspan\{g_1, \dots, g_m\}$ for two sequences of quadratic forms $f_1, \dots, f_m$ and $g_1, \dots, g_m$, one necessary condition is that $\linspan\{D_{f_1^S}, \dots, D_{f_m^S}\} = \linspan\{D_{g_1}, \dots, D_{g_m}\}$, which is an instance of the alternating matrix tuple isometry problem. 
The alternating matrix space isometry problem over $\F_p$ for  odd prime $p$ was studied in~\cite{Sun23,IMQSZ}. Although the results of \cite{Sun23,IMQSZ} cannot be used to solve alternating matrix space isometry over $\F_2$, their techniques allow us to reduce the quadratic form space isometry problem to the quadratic form tuple isometry problem. 
More specifically, based on $f_1, \dots, f_m$ and $g_1, \dots, g_m$, we can construct two tuples of quadratic forms $\mathbf{f} = (f_1', \dots, f_{m'}')$ and $\mathbf{g} = (g_1', \dots, g_{m'}')$ in $n'$ variables over $\F_2$ with $n' = O(n + m)$ and $m'=\poly(n, m)$
so that the linear space spanned by $f_1, \dots, f_m$ is isometric to that of $g_1, \dots, g_m$ if and only if there is an invertible matrix $S'$ such that $\mathbf{f}^{S'} = \mathbf{g}$, i.e., 
$f_i'^{S'} = g_i'$ for all $1 \leq i \leq m'$. 


We show that the quadratic form tuple isometry problem for the tuples derived from the space isometry problem can be solved in time $2^{O(n'^{1.5})}\cdot \poly(m')$ time, which is $2^{O((n + m)^{1.5})}$ due to the relations between $n, m, n'$ and $m'$ (Theorem~\ref{thm:matrix_tuple_constructed_main}). Consequently, Theorem~\ref{thm:group_iso} and Theorem~\ref{thm:quadratic_space_isometry} are obtained. 
For the general quadratic form tuple isometry problem (for tuples not necessarily obtained from quadratic form space problem), Theorem~\ref{thm:quadratictuple_isometry_intro} is obtained as a special case of Theorem~\ref{thm:quadratic_space_isometry}.

\subsubsection{An algorithm for quadratic form tuple isometry.}
In the rest of this section, we give a high-level overview of our algorithm for the quadratic form tuple problem for the tuples constructed from the space isometry problem.

For simplicity of notation, we assume that the input matrix tuples are $\mathbf{f}$ and $\mathbf{g}$ of length $m$ in variables $n$ over $\F_2$, where $m$ and $n$ in this section correspond to $m'$ and $n'$ respectively in Section~\ref{sec:into_overview_2}. 

Before getting into the details, we make some remarks. Our algorithm starts with a new block diagonalization theory for alternating matrix tuples. The block diagonalization then leads to two special cases of the quadratic form tuple isometry problem, which we call the frequent block case and the infrequent block case, respectively. We develop a combinatorial approach for the frequent block case. For the infrequent block case, we show a bound on the index of a normal
$2$-subgroup 
of an associated group, and use matrix group algorithms \cite{Luks92, LuksMiyazaki02, LuksMiyazaki11}
to solve the problem. As far as we know, this is the first time that finite matrix group algorithms are utilized to make progress on the worst-case complexities of $p$-group isomorphism. 
Finally, the general problem is reduced to the two special cases. 
Note that our algorithm follows a 
different strategy from the matrix tuple isometry algorithm for odd $p$ in \cite{IQ19}. 

\paragraph{Congruent block diagonalization of alternating matrix tuples.}
Our algorithm is built on a block-diagonalization theory for alternating matrix tuples under the congruence action which we develop here. Recall that for an alternating matrix tuple $\mathbf{A}\in\M(n, \F)^m$, the congruence action of $S\in\GL(n, \F)$ sends $\mathbf{A}$ to $S^T\mathbf{A}S$.

More specifically, we show that every alternating matrix tuple can be block-diagonalized under the congruence action with each block corresponding to one or two minimum blocks under the left-right action ($\mathbf{A} \mapsto P\mathbf{A}Q^T$ for $P,Q \in \GL(n, \F)$). 
(See Definition~\ref{def:block} for precise definitions of 
blocks and minimum blocks.)
If the block for the congruence action corresponds to two minimum left-right blocks, then the two minimum left-right blocks are the transposes of one another (Theorem~\ref{thm:congruence_block_diagonalization}). 


For isometry testing, we categorize the minimum left-right blocks within a minimum left-right block diagonalization into two types: (for some $0 < \Delta < 1$ 
to be determined later) frequent blocks, which appear more than $n^{1 - \Delta}$  
times, and infrequent blocks, which do not.
Consequently, every congruent block is either a frequent congruent block---a congruent block in which the minimum left-right  blocks are all frequent---or an infrequent congruent block, in which the minimum left-right blocks are all infrequent.

Our algorithm for the quadratic form tuple isometry first considers two special cases:
\begin{enumerate}
    \item Frequent block quadratic form tuple isometry problem, in which every congruent block is frequent.
    \item Infrequent block quadratic form tuple isometry problem, in which every congruent block is infrequent.
\end{enumerate}

We first give new algorithms to solve these two special cases in time 
$2^{O(n^{1+\Delta})} \cdot \poly(m)$ 
and $2^{O(n^{2 - \Delta})}\cdot \poly(m)$ 
respectively, and then reduce the general problem to these two special cases using $\Delta = 1/ 2$. 

\paragraph{Frequent block quadratic form tuple isometry.}
We develop a combinatorial approach for the frequent block quadratic form tuple problem (Theorem~\ref{thm:small_block}). 
The intuition behind our algorithm is that the number of quadratic form tuples congruent to a given input tuple---and satisfying certain conditions---is significantly smaller than the trivial upper bound of $2^{O(n^2)}$. Moreover, all such tuples can be enumerated efficiently.

This assumed frequent block property implies that every minimum left-right block of the matrix tuple has at most $n^\Delta$ rows and  columns. 
Consequently, if we apply a linear transformation to the quadratic form tuple such that the corresponding partial derivative matrix tuple becomes congruently block diagonalized as discussed,
then the entries of the partial derivative matrix tuple that are far from the diagonal are zero. Building on this observation, we show that the number of possible congruently block diagonalized partial derivative matrix tuples is roughly bounded by $2^{O(n^{1 + \Delta})}$ (with some additional conditions). 

We then develop an algorithm to find all quadratic form tuples that are congruent to a given input quadratic form tuple and whose partial derivative matrix tuples are congruently block diagonalized. 
Our algorithm iteratively applies \emph{local operations} to an arbitrary such quadratic form tuple to discover all of them. These operations are local in the sense that each affects only one or two congruent blocks in the congruent block diagonalization, while preserving the overall block diagonal structure.
We show that for any two quadratic form tuples congruent to an input quadratic form tuple satisfying our requirements, there exists a sequence of such local operations that transforms one into another.

\paragraph{Infrequent block quadratic form tuple isometry.} Unlike the combinatorial approach for the frequent block quadratic form tuple isometry problem, 
our algorithm for the infrequent block quadratic form tuple isometry heavily relies on algebraic insights (Theorem~\ref{thm:large_block}).

Our algorithm starts with constructively solving the left-right equivalence 
problem for the partial derivative matrix tuples of the two forms. That is, we attempt to find a pair $(S_0,S_0')$ of invertible matrices such that
the second partial derivative tuple can be obtained by multiplying the first one 
by $S_0^T$ on the left and by $S_0'$ on the right. This problem can be solved in polynomial time. 
If no such pair is found then the two quadratic form tuples are not isometric. 

Otherwise any matrix $S$ that transforms the first quadratic form tuple to the second satisfies that for $S_1=S_0^{-1}S$ and $S_1'=S^{-1}S_0'$, multiplication by $S_1^T$ from the left and by $S_1'$ from the right leaves the second partial derivative matrix tuple invariant. It turns out that the set of invertible matrices $S_1$ such that there exists an invertible matrix $S_1'$ with the property above is a group $G$ having a normal $2$-subgroup $R$ of index at most $2^{n^{2-\Delta}}$ 
which can be computed in polynomial time. We take a system of representatives of the cosets of $R$ in $G$ of each element $Y$ of the system we attempt to find $X\in R$ such that $S_0YX$ transforms the first quadratic form tuple to the second one. A wide generalization of this task can be solved in polynomial time by \cite{LuksMiyazaki02,LuksMiyazaki11}, because $R$ is a $2$-group. Of course we can stop as soon as such an $X$ is found for a particular $Y$. If we exhaust the coset representatives without success then the two form tuples are non-isometric. 

\paragraph{Quadratic form tuple isometry.}
Finally, we solve the quadratic form tuple isometry problem by reducing to roughly $2^{O(n^{1+\Delta})}$
instances of the aforementioned two special cases (Theorem~\ref{thm:quadratic_form_tuple_iso_algo}).

To start, for the input quadratic form $\mathbf{f}$, we apply an invertible linear transform so that the partial derivative matrix tuple of  $\mathbf{f}$  is congruently block diagonalized as discussed above. 
Naturally, we can divide  $\mathbf{f}$ into two induced quadratic form tuples  $\mathbf{f}'$ and  $\mathbf{f}''$ such that the partial derivative matrix tuple of 
$\mathbf{f}'$ contains all frequent congruent blocks from $\mathbf{f}$, and $\mathbf{f}''$ contains all the infrequent congruent blocks. 
In the same way, we can divide $\mathbf{g}$ into $\mathbf{g}'$ and $\mathbf{g}''$. 

Our goal is to run the frequent quadratic form isometry algorithm for $\mathbf{f}'$ and $\mathbf{g}'$ and the infrequent quadratic form isometry algorithm for $\mathbf{f}''$ and $\mathbf{g}''$. 
However, it is not hard to construct an example such that  although $\mathbf{f}$ and $\mathbf{g}$ are congruent, any linear transform $S$ such that $\mathbf{f}^S = \mathbf{g}$ does not map $\mathbf{f}'$ to $\mathbf{g}'$ and $\mathbf{f}''$ to $\mathbf{g}''$. 

In our algorithm, we produce a set $\{\mathbf{f}_1, \dots, \mathbf{f}_k\}$ of quadratic forms 
that are all congruent to $\mathbf{f}$ and have their partial derivative matrix tuples congruently block diagonalized. We divide each $\mathbf{f}_i$ into $\mathbf{f}_i'$ and $\mathbf{f}_i''$ as the way to divide $\mathbf{f}$ into $\mathbf{f}'$ and $\mathbf{f}''$.
We further require that if $\mathbf{f}$ and $\mathbf{g}$ are congruent, then there exists a produced $\mathbf{f}_i$ such that $\mathbf{f}_i'$ and $\mathbf{g}'$ are congruent, and $\mathbf{f}_i''$ and $\mathbf{g}''$ are congruent. 
We give an algorithm to produce a set of roughly
$2^{O(n^{1 + \Delta})}$ 
different quadratic form tuples satisfying the above requirements, and then the quadratic form tuple isometry problem is reduced to 
$2^{O(n^{1 + \Delta})}$ 
instances of the frequent quadratic form tuple isometry problem and the infrequent quadratic form tuple isometry problem. 




\section{Preliminaries}
In this paper, we assume that the groups are given by their Cayley tables, and matrices and quadratic forms are over $\F_2$ unless otherwise specified. 

We use $\M(n \times r, \F)$ to denote the linear space of all the $n\times r$ matrices over $\F$, and $\GL(n, \F)$ to denote the general linear group of degree $n$ over $\F$.

For a matrix $U \in \M(n \times r, \F_2)$, we use $\mathrm{columnspan}(U)$ to denote the linear space spanned by the columns of the matrix $U$. The notation $U^T$ denotes the transpose of $U$. 
For a vector space $V \leq \F_2^r$, we use $U V$ to denote the linear space $\langle U v : v \in V\rangle$, where $\langle \cdot \rangle$ denotes the linear span. 

Let $V$ be a vector space such that $V = V_1 \oplus \dots \oplus V_k$. Any vector $v \in V$ has a unique $v_i \in V_i$ for all $1 \leq i \leq k$ such that $v = v_1 + \dots + v_k$. We say that $v_i$ is the projection of $v$ on $V_i$ with respect to $V_1\oplus \dots\oplus V_k$.
Furthermore, let $U$ be a matrix with each column in $V$. 
The projection of $U$ on $V_i$ with respect to $V_1 \oplus \dots \oplus V_k$ is a matrix $U'$ of the same dimensions  
as $U$ such that the $j$-th column of $U'$ is the projection of the $j$-th column of $U$ on $V_i$ with respect to $V_1 \oplus \dots \oplus V_k$.

Let $f$ be a quadratic form (=homogeneous polynomial of degree 2) in $n$ variables over $\F_2$. 
Throughout the paper, we suppose that the $n$ variables are $x_1, \dots, x_n$ unless specified. 
$f$ is defined by the coefficients $f[i, j] \in \F_2$ for all $1 \leq i \leq j \leq n$ such that  $f = \sum_{1 \leq i \leq j \leq n} f[i, j] x_i x_j$.

Let $S$ be an invertible matrix over $\F_2$. Define $f^S$ on $n$ variables $y_1, \dots, y_n$ over $\F_2$ as the quadratic form 
derived from $f$ on $x_1, \dots, x_n$ by applying the substitution of variables $x_i = \sum_{j = 1}^n S[i, j] y_j$, viz. 
$f^S(\vec{y}) := f(S \vec{y})$. 

\subsection{Left-right block diagonalization of matrix tuples}

An $n \times r$ matrix tuple $\mathbf{A}$ of length $m$ over $\F_p$ for some prime $p$ is an ordered sequence $(A_1, \dots, A_m) \in \M(n \times r, \F_p)^m$. 
A square matrix tuple $\mathbf{A} = (A_1, \dots, A_m) \in \M(n \times n, \F_p)$ is an \emph{alternating} matrix tuple if $A_i[j, j]$ = 0 for any $1 \leq i \leq m, 1 \leq j \leq n$, and $A_i^T = - A_i$ for any $1 \leq i \leq m$.
If $p = 2$, this implies that every $A_i$ is a symmetric matrix (with zero diagonal). 

For an $n \times r$ matrix tuple $\mathbf{A}$ of length $m$ and a column vector $v \in \F_p^r$, we use $\mathbf{A}v$ to denote the $n \times 1$ column vector tuple $(A_1 v, \dots, A_m v)$.

Throughout this paper, unless specified, we assume that all matrix tuples are non-degenerate, i.e., for $\mathbf{A} \in \M(n \times r, \F_p)^m$, there is no nonzero column vector $v \in \F_p^r$ such that $\mathbf{A} v = 0$, and no nonzero column vector $u \in \F_p^n$ such that $u^T \mathbf{A} = 0$.

We consider two different actions on matrix tuples. For a tuple $\mathbf{A} = (A_1, \dots, A_m) \in \M(n \times r, \F_p)^m$, the left-right action by matrices $U\in \GL(n,  \F_p)$ and $W \in \GL(r, \F_p)$ is given by $U^T \mathbf{A} W := (U^T A_1 W, \dots, U^T A_m W)$. In the congruence action, we require that the matrices in $\mathbf{A}$ are square; under this action, a matrix $U \in \GL(n, \F_p)$ acts on $\mathbf{A}$ by conjugation via $U^T \mathbf{A} U$. 

\begin{definition}
    Two matrix tuples $\mathbf{A}, \mathbf{A}' \in \M(n \times r, \F_p)^m$ are \emph{left-right equivalent} if there are invertible matrices $U \in \GL(n, \F_p)$ and $W \in \GL(r, \F_p)$ such that $U^T \mathbf{A} W = \mathbf{A}'$. 

    Two alternating matrix tuples $\mathbf{A}, \mathbf{A}' \in \M(n \times n, \F_p)^m$ are \emph{congruent} 
    if there is an invertible matrix $U \in \GL(n, \F_p)$ such that $U^T \mathbf{A} U = \mathbf{A}'$. 
\end{definition}

\begin{definition}[Left-right block and minimum left-right block]\label{def:block}
For any prime $p$, a matrix tuple $\mathbf{A} \in \M(n \times r, \F_p)^m$,  and two matrices $X \in \M(n \times n', \F_p)$, $Y \in \M(r \times r' , \F_p)$ for two integers $n'$ and $r'$,
$X^T \mathbf{A} Y$ is a \emph{left-right block} of $\mathbf{A}$ if there exist matrices $X' \in \M(n \times (n - n'), \F_p )$ and $Y' \in \M(r\times (r - r'), \F_p)$ such that invertible matrices 
\[
\begin{pmatrix}
    X & X'
\end{pmatrix} \in \GL(n, \F_p) \text{ and }\begin{pmatrix}
    Y & Y'
\end{pmatrix} \in \GL(r, \F_p)
\] satisfy
\[
\begin{pmatrix}
    X & X'
\end{pmatrix}^T \mathbf{A} \begin{pmatrix}
    Y & Y'
\end{pmatrix}
=\begin{pmatrix}
    X^T \mathbf{A} Y & \\ & X'^T \mathbf{A} Y'  
\end{pmatrix}. 
\]
$X^T \mathbf{A} Y$ is a \emph{minimum left-right block} of $\mathbf{A}$ if there does not exist $U \in \M(n' \times n'', \F_p)$ and $W \in \M(r' \times r'', \F_p)$ for some $1 \leq n'' < n'$ and $1 \leq r'' < r'$ such that $U^T X^T \mathbf{A} Y W$ is a left-right block of $\mathbf{A}$.

Moreover, for a minimum left-right block $\mathbf{B} = X^T \mathbf{A} Y$ of $\mathbf{A}$ consisting of $n'$ rows and $r'$ columns, let $\Gamma_\mathbf{B}$ be the set 
\[\begin{split}\Gamma_\mathbf{B} := \{X' \in \M(n \times n', \F_p): & \exists Y' \in \M(r \times r', \F_p)\text{ s.t. }  \\ & X'^T \mathbf{A} Y' \text{ is a left-right block and } X'^T \mathbf{A} Y' =\mathbf{B}\}, \end{split}\]   
$V_{\mathbf{B}} := \langle \Gamma_\mathbf{B} \rangle$, and $K_\mathbf{B}$ be the set difference $V_\mathbf{B} \setminus \Gamma_\mathbf{B}$. 

\end{definition}

For invertible matrices $U$ and $W$, we say that $U^T \mathbf{A} W$ is a left-right block diagonalization of $\mathbf{A}$ if $U^T \mathbf{A} W = \diag(\mathbf{A}_1,\dots, \mathbf{A}_k)$ such that each of the $\mathbf{A}_i$ is a left-right block of $\mathbf{A}$, and the left-right block diagonalization is minimum if each of $\mathbf{A}_i$ is a minimum left-right block of $\mathbf{A}$.



In this paper, we use the following results for matrix tuples under the left-right action from \cite{qiao2024canonical}. We note that these results hold even when $p=2$. 

\begin{theorem}[Rephrased, Theorem 1.3 and Lemma 7.2 of \cite{qiao2024canonical}]\label{thm:left_right_decomposition}
Let $\mathbf{A} \in \M(n \times r, \F_p)^m$ be a matrix tuple over $\F_p$ for some prime $p$. 
Let $\mathbf{B} = X^T \mathbf{A} Y$ be a minimum left-right block of $\mathbf{A}$ with $s$ rows and $t$ columns. 
Then we have
\begin{enumerate}
\item
 $K_\mathbf{B}$ is a linear subspace of $V_\mathbf{B}$, i.e., 
 $\langle K_\mathbf{B} \rangle = K_\mathbf{B}$. 

\item 
Let $U = \begin{pmatrix}U_1 \cdots U_k \end{pmatrix}$ and $W =\begin{pmatrix}W_1  \cdots  W_k \end{pmatrix}$ be two invertible matrices such that 
$U^T \mathbf{A} W = \diag(U_1^T \mathbf{A} W_1, \dots, U_k ^T\mathbf{A} W_k)$ is a left-right block diagonalization of $\mathbf{A}$. 
The projection of $X$ on any $U_i$ with respect to $U_1 \oplus \dotsb \oplus U_k$ is in $V_\mathbf{B}$. 
\end{enumerate}
\end{theorem}

\begin{theorem}[Theorem 1.1 of {\cite{qiao2024canonical}}]\label{thm:left_right_decomposition_algorithm}
Given a matrix tuple $\mathbf{A} \in \M(n \times r, \F)^n$, 
there is a Left-Right Canonical Form Algorithm to compute $\diag(\mathbf{B}_1, \dots, \mathbf{B}_k)$ and two invertible matrices $U \in \GL(n, \F), W \in \GL(r, \F)$ such that 
$U^T \diag(\mathbf{B}_1, \dots, \mathbf{B}_k) W = \mathbf{A}$, and every $\mathbf{B}_i$ is a minimum left-right block of $\mathbf{A}$. 
The algorithm also outputs $V_{\mathbf{B}_i}$ and $K_{\mathbf{B}_i}$ for every $1 \leq i \leq k$.
\end{theorem}

\subsection{Tensor isomorphism}

A (3-)tensor $\mathfrak{T}$ of dimension $n \times m \times r$ over a finite field $\F$ is an element in $\F^{m \times n \times r}$. 
For convenience, we use $\mathfrak{T}[i, j, k]$ to denote the $(i, j, k)$-th entry of $\mathfrak{T}$. 
We also use $\mathcal{X}_{\mathfrak{T}, i}$ to denote the $n \times r$ matrix such that $\mathcal{X}_{\mathfrak{T}, i}[j, k] = \mathfrak{T}[i, j, k]$, 
$\mathcal{Y}_{\mathfrak{T}, j}$ to denote the $m \times r$ matrix such that $\mathcal{Y}_{\mathfrak{T}, j}[i, k] = \mathfrak{T}[i, j, k]$, 
and $\mathcal{Z}_{\mathfrak{T}, k}$ to denote the $m \times n$ matrix such that $\mathcal{Z}_{\mathfrak{T}, k}[i, j] = \mathfrak{T}[i, j, k]$. 

We say that a tensor $\mathfrak{T}$ is an alternating tensor if the tensor is of dimension $m \times n \times n$ such that $\mathcal{X}_{\mathfrak{T}, i}^T = -\mathcal{X}_{\mathfrak{T}, i}$ for any $1 \leq i \leq m$.

For an alternating tensor $\mathfrak{T}$ and two invertible matrices $M \in \mathrm{GL}(m, \F), N \in \mathrm{GL}(n, \F)$, we use 
$\mathfrak{T}_{\mathrm{Trans}(M, N)}$ to denote the tensor $\mathfrak{S}$ such that $\mathcal{X}_{\mathfrak{S}, i} = \sum_{j = 1}^m M[j, i] N^T \mathcal{X}_{\mathfrak{T}, j} N$. 
\begin{problem}[Tensor isomorphism problem]\label{def:tensor_iso}
Given two alternating 3-tensors $\mathfrak{T}$ and $\mathfrak{T}'$ of dimension $m \times n \times n$ over some  field $\F$, determine if there exist two matrices $M \in \mathrm{GL}(m, \F)$ and $N \in \mathrm{GL}(n, \F)$ such that $\mathfrak{T}_{\mathrm{Trans}(M, N)} = \mathfrak{T}'$. 
\end{problem}

In \cite{Sun23}, a $p^{O((n + m)^{1.8}\log p)}$ time algorithm for the tensor isomorphism problem with tensors over $\F_p$, where $p$ is odd, was given. This result was improved to $p^{O((n + m)^{1.5}\log p)}$ in~\cite{IMQSZ}. However, neither result works for tensors over $\F_2$. 

We say that an $m \times n \times n$ alternating tensor $\mathfrak{T}$ is with respect to the parameters $\alpha_0, \alpha_1, \beta_0, \beta_1$ if the following two conditions hold: 
\begin{enumerate}
\item For any $1 \leq i \leq \alpha_0$, $\mathcal{X}_{\mathfrak{T}, i}[j, k] = 0$ for any $1 \leq j, k \leq \beta_0 + \beta_1$. 
\item For any $\alpha_0 + 1 \leq i \leq \alpha_0 + \alpha_1$, $\mathcal{X}_{\mathfrak{T}, i}[j, k] = 0$ for any $1 \leq j \leq \beta_0 + \beta_1, 1 \leq k \leq \beta_0$ and any $1 \leq j \leq \beta_0,  1 \leq k \leq \beta_0 + \beta_1$. 
\end{enumerate}
\begin{problem}[Restricted tensor isomorphism problem]\label{def:restricted_tensor_iso}
Given two alternating 3-tensors $\mathfrak{T}$ and $\mathfrak{T}'$ of dimension $m \times n \times n$ over some finite field $\F$ with respect to parameters $\alpha_0, \alpha_1, \beta_0, \beta_1$, determine if there are matrices $M \in \mathrm{GL}(m, \F)$ and $N \in \mathrm{GL}(n, \F)$ such that $\mathfrak{T}_{\mathrm{Trans}(M, N)} = \mathfrak{T}'$ with 
\begin{equation}\label{equ:restricted_tensor_iso}
M = \begin{pmatrix}
M_1\\
& I_{m - \alpha_0 - \alpha_1}
\end{pmatrix}
\text{ and }
N = \begin{pmatrix}
N_1 \\
& I_{n - \beta_0 - \beta_1}
\end{pmatrix},
\end{equation}
where 
\begin{equation}\label{equ:restricted_tensor_iso_1}
M_1 = \begin{pmatrix}
M' & M''  \\
& I_{\alpha_1} \\
\end{pmatrix}
\end{equation}
and 
\begin{equation}\label{equ:restricted_tensor_iso_2}
N_1 = \begin{pmatrix}
N' & N''\\
 & I_{\beta_1} \\
\end{pmatrix}
\end{equation}
for some $M' \in \mathrm{GL}(\alpha_0, \F)$, $M'' \in \mathrm{M}(\alpha_0 \times \alpha_1, \F)$, 
$N' \in \mathrm{GL}(\beta_0, \F)$, and $N'' \in \mathrm{M}(\beta_0 \times \beta_1, \F)$.
\end{problem}

Although the tensor isomorphism algorithms in~\cite{Sun23,IMQSZ} do not work for tensors over $\F_2$, the following reduction from the tensor isomorphism problem to the restricted tensor isomorphism problem continues to work over $\F_2$.

\begin{theorem}[Rephrased, Lemma 5.14 of \cite{Sun23}, Lemma 3.4 and Lemma 3.6 of \cite{IMQSZ}]\label{thm:tensor_iso_to_restricted_tensor_iso}
For two input alternating tensors $\mathfrak{T}$ and $\mathfrak{T}'$ of dimension $m \times n \times n$ over $\F_p$ for any prime $p$, 
there is an algorithm, called Tensor Isomorphism Reduction Algorithm,  with running time $p^{\tilde O((m + n)n^{0.5})}$
to reduce the tensor isomorphism problem 
$p^{\tilde O((m + n)n^{0.5})}$
instances of the restricted tensor isomorphism problem.
For any $1 \leq i \leq r$, the algorithm produces invertible matrices $A_i, A_i', B_i, B_i'$
and parameters $\alpha_{i, 0}, \alpha_{i, 1}, \beta_{i, 0}, \beta_{i, 1}$ satisfying $0 \leq m - \alpha_{i, 0} - \alpha_{i, 1}, n - \beta_{i, 0} - \beta_{i, 1} \leq n^{0.5}$ such that 
the $i$-th input instance of the restricted tensor isomorphism problem is $\mathfrak{T}_{\mathrm{Trans}(A_i, B_i)}, \mathfrak{T}'_{\mathrm{Trans}(A_i', B_i')}$ with respect to $\alpha_{i, 0}, \alpha_{i, 1}, \beta_{i, 0}, \beta_{i, 1}$.


Moreover, 
for any $P$ and $Q$ satisfying $\mathfrak{T}' = \mathfrak{T}_{\mathrm{Trans}(P, Q)}$, there exist $1 \leq i \leq r$ and $M, N$ satisfying Equations~(\ref{equ:restricted_tensor_iso}), (\ref{equ:restricted_tensor_iso_1}), and (\ref{equ:restricted_tensor_iso_2}) 
as a solution for the $i$-th instance of the restricted tensor isomorphism problem such that $P = A_i M A_i'^{-1}$ and $Q = B_i N B_i'^{-1}$. 
\end{theorem}


\section{Reduce Frattini class 2 group isomorphism to quadratic form space isometry}


In this section, we reduce the problem of testing isomorphism of $2$-groups of Frattini class 2 to the following quadratic form space isometry problem.

\begin{problem}[Quadratic form space isometry problem]\label{problem:quadratic_form_space_iso}
Let $Q(n, \F_2)$ be the linear space of quadratic forms in $n$ variables over $\F_2$.
Given two sequences of quadratic forms $f_1, \dots f_m$ and $g_1, \dots, g_m$ defined in $n$ variables, determine whether there exists an invertible matrix $S \in \mathrm{GL}(n, \F_2)$ such that $\langle f_1^S, \dots, f_m^S \rangle = \langle g_1, \dots, g_m\rangle$, where $\langle \cdot \rangle$ denotes the linear span.
\end{problem}

This reduction is based on Higman's results on enumerating $p$-groups of Frattini class $2$ \cite{HigmanEnum}. We recall some results from there, and we will refer to \cite{BNV07} which contains a more detailed exposition of the results in \cite{HigmanEnum}.

\begin{definition}\label{def:frattini}
Let $G$ be a finite group. The \emph{Frattini subgroup} of $G$, denoted as $\Phi(G)$, is the intersection of all maximal subgroups of $G$. 

A group $G$ is a  \emph{$p$-group of Frattini class $2$} for some prime $p$, if the order of every element of $G$ is  a power of $p$, and 
its Frattini subgroup is central and elementary abelian, or equivalently, $G$ satisfies that for any $a, b, c\in G$, $a^{p^2}=\id$, $[a, b]^p=\id$, and $[a, [b, c]]=\id$, where $[a, b] = a^{-1} b^{-1} a b$, and $\id$ is the identity of the group. 

The \emph{relatively free $p$-group of Frattini class $2$ with $n$ generators}, denoted as $F_{\Phi\text{-}2, p, n}$, is the quotient of the free group with $n$ generators $\{a_1, \dots, a_n\}$ by the subgroup generated by all words of the form $a^{p^2}$, $[a, b]^p$, and $[a, [b, c]]$.
\end{definition}

\paragraph{Relatively free $p$-groups of Frattini class $2$.} 
Let $F_{\Phi\text{-}2, p, n}$ be the relatively free $p$-group of Frattini class $2$ with $n$ generators.
Let $a_1, \dots, a_n$ be a set of generators of $F_{\Phi\text{-}2, p, n}$. Let $P$ be the group generated by $a_i^2$, $i\in[n]$. Let $D$ be the commutator subgroup generated by $[a_i, a_j]$ for $1\leq i<j\leq n$. By \cite[pp. 24]{BNV07}, $a_i^p$ and $[a_i,a_j]$ form a minimal generating set of $\Phi(F_{\Phi\text{-}2, p, n})$; that is, there is no non-trivial relation between them. Therefore, \[\Phi(F_{\Phi\text{-}2, p, n})\cong\Z_p^{n+\binom{n}{2}}.\] 


\paragraph{Automorphisms of $F_{\Phi\text{-}2, p, n}$.} Let $\alpha$ be an automorphism of $F_{\Phi\text{-}2, p, n}$. By \cite[Lemma 4.2]{BNV07}, if $\alpha$ induces the identity mapping on $F_{\Phi\text{-}2, p, n}/\Phi(F_{\Phi\text{-}2, p, n})$, then $\alpha$ fixes $\Phi(F_{\Phi\text{-}2, p, n})$ pointwise. Or in other words, the natural homomorphism from $\aut(F_{\Phi\text{-}2, p, n})$ to $\aut(F_{\Phi\text{-}2, p, n}/\Phi(F_{\Phi\text{-}2, p, n}))$ is injective. It is also clear that any $\alpha'\in \aut(F_{\Phi\text{-}2, p, n}/\Phi(F_{\Phi\text{-}2, p, n})$ yields $\alpha\in\aut(F_{\Phi\text{-}2, p, n})$. So the natural homomorphism 
from $\aut(F_{\Phi\text{-}2, p, n})$ to $\aut(F_{\Phi\text{-}2, p, n}/\Phi(F_{\Phi\text{-}2, p, n}))$ is an isomorphism.

\paragraph{Higman's isomorphism criterion.} Let $G$ be a $p$-group of Frattini class $2$, such that the Frattini quotient $G/\Phi(G)$ is isomorphic to $\Z_p^n$. Then there exists $S_G\leq \Phi(F_{\Phi\text{-}2, p, n})$, such that $G$ is isomorphic to $F_{\Phi\text{-}2, p, n}/S_G$. Let $H$ be another $p$-group of Frattini class $2$ with $H/\Phi(H)\cong \Z_p^n$, so that $H\cong F_{\Phi\text{-}2, p, n}/S_H$, where $S_H\leq \Phi(F_{\Phi\text{-}2, p, n})$. By \cite[Lemma 4.3]{BNV07}, $G$ and $H$ are isomorphic if and only if the induced action of $\aut(F_{\Phi\text{-}2, p, n}/\Phi(F_{\Phi\text{-}2, p, n}))\cong \GL(n, \F_p)$ on $\Phi(F_{\Phi\text{-}2, p, n})$ sends $S_G$ to $S_H$ as subgroups of $\Phi(F_{\Phi\text{-}2, p, n})$.

\paragraph{The induced action of $\aut(F_{\Phi\text{-}2, p, n}/\Phi(F_{\Phi\text{-}2, p, n}))$ on  $\Phi(F_{\Phi\text{-}2, p, n})$.} Higman noted that for $p>2$, the induced action of $\aut(F_{\Phi\text{-}2, p, n}/\Phi(F_{\Phi\text{-}2, p, n}))$ on  $\Phi(F_{\Phi\text{-}2, p, n})$ is the natural action of $\GL(n, \F_p)$ on $(\F_p^n\wedge \F_p^n)\oplus \F_p^n$, where $\wedge$ denotes the exterior product \cite[Theorem 2.2]{HigmanEnum}. The case of $p=2$ seems not explicitly written in the literature, so we include an exposition here. 

\begin{lemma}\label{lem:induced-action}
    Let $p=2$. The induced action of $\aut(F_{\Phi\text{-}2, 2, n}/\Phi(F_{\Phi\text{-}2, 2, n}))$ on  $\Phi(F_{\Phi\text{-}2, 2, n})$ is equivalent to the congruence action of $\GL(n, \F_2)$ on $\S(n, 2)$, the linear space of $n\times n$ symmetric matrices over $\F_2$.
\end{lemma}
\begin{proof}
In the following we shall make use of some basic properties of commutator brackets, as summarized in \cite[Section 3.2]{BNV07}.

Let $F_{\Phi\text{-}2, p, n}$ be the relatively free $p$-group of Frattini class $2$ with $n$ generators.
Let $a_1, \dots, a_n$ be a set of generators of $F_{\Phi\text{-}2, p, n}$. Recall that $(a_ia_j)^p=a_i^pa_j^p[a_j, a_i]^{\binom{p}{2}}$, $P$ is the group generated by $a_i^p$, $i\in[n]$, and $D$ is the commutator subgroup generated by $[a_i, a_j]$, for $1\leq i<j\leq n$. 


When $p=2$, we have $(a_ia_j)^2=a_i^2a_j^2[a_j,a_i]$. As $[a_j,a_i]=[a_i,a_j]^{-1}=[a_i,a_j]$ (using the fact that in $F_{\Phi\text{-}2, 2, n}$, $[a_i,a_j]^2=\id$), this leads to $(a_ia_j)^2=a_i^2a_j^2[a_i,a_j]$. Therefore, $P$ is not invariant under $\aut(F_{\Phi\text{-}2, 2, n}/\Phi(F_{\Phi\text{-}2, 2, n}))$, unlike the odd $p$ case. By induction, it can be shown that $(a_{i_1}a_{i_2}\dots a_{i_s})^2=a_{i_1}^2a_{i_2}^2\dots a_{i_s}^2 \prod_{1\leq j<k\leq s}[a_{i_j}, a_{i_k}]$.

We identify $\aut(F_{\Phi\text{-}2, 2, n}/\Phi(F_{\Phi\text{-}2, 2, n}))\cong \GL(n, 2)$ and write the group operations in $F_{\Phi\text{-}2, 2, n}$ by additive notation. 

It can be checked easily that $D$ is invariant under $\aut(F_{\Phi\text{-}2, 2, n}/\Phi(F_{\Phi\text{-}2, 2, n}))$. This is because $[ab,c]=[a,c][[a,c],b][b,c]$, so by $[[a,c],b]=\id$ in $F_{\Phi\text{-}2, 2, n}$, we have $[ab,c]=[a,c][b,c]$. It can be verified that the action of $\aut(F_{\Phi\text{-}2, 2, n}/\Phi(F_{\Phi\text{-}2, 2, n}))$ on $D$ is equivalent to the action of $\GL(n, 2)$ on the linear space of alternating matrices over $\F_2$, similar to the odd $p$ case.

It remains to understand the induced action of $\aut(F_{\Phi\text{-}2, 2, n}/\Phi(F_{\Phi\text{-}2, 2, n}))$ on $P$. 
As mentioned above, this action will take elements of $P$ outside of $P$ itself (into the group $PD$), so we are really understanding the part of the action on $PD$ where we start with an element of $P$ and an automorphism, and see where that element is mapped by the automorphism.

Note that $P\cong \Z_2^n$, which can be understood as the vector space of length-$n$ \emph{row} vectors over $\F_2$.

Let $U=(u_{i,j})\in \GL(n, 2)$. 
Let ${v}=(v_1, \dots, v_n)\in \F_2^n$, and let $V=\diag({v})\in\M(n, 2)$ be the diagonal matrix with $V(i,i)=v_i$. We claim that the coefficients in $U^T V U$ 
correspond to the action of $U\in\aut(F_{\Phi\text{-}2, 2, n}/\Phi(F_{\Phi\text{-}2, 2, n}))$ on $P$; here we identify the diagonal elements of matrices with elements of $P$ and the off-diagonal part (which will always be alternating) with $D$.

To see this, we first work on the group side. First, $a_i^2$ is sent to $(\prod_{j\in[n], u_{i, j}\neq 0}a_j)^2$, which is equal to \[\prod_{j\in[n], u_{i, j}\neq 0}a_j^2\prod_{1\leq j<k\leq n, u_{i, j}\neq 0, u_{i, k}\neq 0}[a_j, a_k].\] Second, note that ${v}$ corresponds to $\prod_{i\in[n], v_i\neq 0}a_i^2$. Collecting over $i\in[n], v_i\neq 0$, we see that whether $a_j^2$ appears in the final result depends on the parity of $|\{i\in [n]\mid v_i\neq 0, u_{i, j}\neq 0\}|$. Since $U$ is a matrix over $\F_2$, the parity of $|\{i\in [n]\mid v_i\neq 0, u_{i, j}\neq 0\}|$ is equal to $\sum_{i\in[n], v_i\neq 0}u_{i, j}$. Similarly, for $1\leq j<k\leq n$, whether $[a_j, a_k]$ appears in the final result depends on the parity of $|\{i\in [n]\mid v_i\neq 0, u_{i, j}u_{i, k}\neq 0\}|$, which is equal to $\sum_{i\in[n], v_i\neq 0} u_{i, j}u_{i, k}$.

We then examine $U^T VU$. 
We first examine its $j$th diagonal entry, which is $\sum_{i\in[n], v_i\neq 0}u_{i, j}^2$. Over $\F_2$, we have $\sum_{i\in[n], v_i\neq 0}u_{i, j}^2=\sum_{i\in[n], v_i\neq 0}u_{i, j}$. We then examine its $(j, k)$-th entry, which is $\sum_{i\in[n], v_i\neq 0}u_{i, j}u_{i, k}$. Therefore, the coefficients of $U^TVU$ correspond to the action of $U$ on $P$. 

Combining our analyses for the action of $\aut(F_{\Phi\text{-}2, 2, n}/\Phi(F_{\Phi\text{-}2, 2, n}))$ on $D$ and $P$, and noting that $\S(n, 2)$ is a direct sum of the linear space of diagonal matrices and the linear space of alternating matrices, we can conclude the proof.
\end{proof}
\begin{example}
	Consider $F_{\Phi\text{-}2, 2, 3}$ with generators $a_1, a_2, a_3$. An automorphism of $F_{\Phi\text{-}2, 2, 3}/\Phi(F_{\Phi\text{-}2, 2, 3})$ can be $a_1\to a_1a_3$, $a_2\to a_2a_3$, $a_3\to a_1a_2a_3$. This corresponds to the matrix \[U=\begin{bmatrix}
		1 & 0 & 1 \\
		0 & 1 & 1 \\
		1 & 1 & 1
	\end{bmatrix}\in\GL(3, 2).\] Then $a_1^2a_3^2$ is transformed to \[(a_1a_3)^2(a_1a_2a_3)^2=a_1^2a_3^2a_1^2a_2^2a_3^2[a_1,a_3][a_1,a_2][a_1,a_3][a_2,a_3]=a_2^2[a_1,a_2][a_2,a_3].\] 
    Let \[V=\begin{bmatrix}
		1 & 0 & 0 \\
		0 & 0 & 0 \\
		0 & 0 & 1 
	\end{bmatrix}.\] Then we compute \[U^TVU=\begin{bmatrix}
		0 & 1 & 0 \\
		1 & 1 & 1 \\
		0 & 1 & 0 
	\end{bmatrix},\] whose coefficients in the upper triangular part correspond to $a_2^2[a_1,a_2][a_2,a_3]$.
\end{example}

We then use the following fact from linear algebra. 
\begin{fact}\label{fact:dual-action}
Let $\S(n, 2)$ be the linear space of $n\times n$ symmetric matrices over $\F_2$. Let $Q(n, 2)$ be the linear space of quadratic forms in $n$ variables over $\F_2$. 
The dual of the congruence representation of $\GL(n,2)$ on
$\S(n,2)$ is canonically isomorphic to the
change-of-variables representation of $\GL(n, 2)$ on $Q(n,2)$.
\end{fact}
\begin{proof}
Define a pairing between $\S(n,2)$ and $Q(n,2)$ by
$\left\langle C,q\right\rangle
=
\sum_{i\in[n]} c_{ii}d_i+\sum_{1\le i<j\le n} c_{ij}e_{ij}$,
where $C=(c_{ij})\in \S(n,2)$ and
$q(x_1,\dots,x_n)=\sum_{i\in[n]} d_i x_i^2+\sum_{1\le i<j\le n} e_{ij}x_ix_j
\in Q(n,2)$, $d_i, e_{i,j}\in\F_2$.
This is a perfect bilinear pairing, since both spaces have dimension
$\binom{n+1}{2}$ and the pairing matches the standard coordinates.

Let $U\in \GL(n,2)$. The congruence action on $\S(n,2)$ is
$C \mapsto U^TCU$,
and the change-of-variables action on $Q(n,2)$ is
$(U\cdot q)(x)=q(U^{-1}x)$.
By expanding $q(U^{-1}x)$ and comparing coefficients, one checks that
$
\langle U^TCU,\, q\rangle
=
\langle C,\, U^{-1}\cdot q\rangle
$
for all $C\in \S(n,2)$ and $q\in Q(n,2)$. 
Hence the representation on $Q(n,2)$ is the contragredient of the
congruence representation on $\S(n,2)$.
\end{proof}

\begin{theorem}
    There exists a polynomial-time algorithm that takes $G$ and $H$, which are $2$-groups of Frattini class-$2$ given by their Cayley tables, and produces two sequences of quadratic forms $g_1, \dots, g_m$ and $h_1, \dots, h_m$ where $g_i, h_j\in\F_2[x_1, \dots, x_n]$, such that $G$ and $H$ are isomorphic if and only if the quadratic form spaces spanned by $g_1, \dots, g_m$ and $h_1, \dots, h_m$ are isometric. Furthermore, $m=\dim(\Phi(G))$ and $n=\dim(G/\Phi(G))$.
\end{theorem}
\begin{proof}
By the discussions above, we need to compute the identification of $G$ as $F_{\Phi\text{-}2, 2, n}/S_G$ for $S_G\leq \Phi(F_{\Phi\text{-}2, 2, n})$, and then take the dual space of $S_G$ and interpret it as a subspace of $Q(n, 2)$. 

First, we compute $S_G$. Find a linear basis $\{a_1, \dotsc, a_n\}$ of $G / \Phi(G)$; then $G$ is a quotient of $F_{\Phi\text{-}2,2,n}$, since the latter is the relatively free $n$-generated group of Frattini class 2. Then $S_G$ is the kernel of the quotient map $F_{\Phi\text{-}2,2,n} \to G$. To compute $S_G$ explicitly, write out the elements $a_1^2, a_2^2, \dotsc, a_n^2, [a_1, a_2], [a_1,a_3], \dotsb, [a_{n-1},a_n]$ of $\Phi(G)$ (in some arbitrary but fixed order). For each such element, if it is in the subgroup generated by the prior elements in the list (in $G$), then writing it as such gives an $\F_2$ linear relation among these elements; viz. if we have $[a_i,a_j] = \sum_{k \in S} a_k^2 + \sum_{(k,\ell) \in T} [a_k,a_\ell]$ in $\Phi(G)$, then we include the element $[a_i,a_j] + \sum_{k \in S} a_k^2 + \sum_{(k,\ell) \in T} [a_k,a_\ell]$ in $S_G$, where here the $a_i$ are the generators of $F_{\Phi\text{-}2,2,n}$. The set of all such elements generates $S_G$.

Next, by Lem.~\ref{lem:induced-action}, the action of $\aut(F_{\Phi\text{-}2,2,n})$ on $\Phi(F_{\Phi\text{-}2,2,n})$ is equivalent to that of $\GL(n,2)$ on $S(n,2)$. We may thus view $\Phi(F_{\Phi\text{-}2,2,n})$ as $S(n,2)$, and view $S_G,S_H$ as linear spaces of symmetric matrices. Higman's criterion then tells us that $G \cong H$ iff the symmetric matrix spaces $S_G,S_H$ are isometric.

Now, let $S$ be a vector space and $A$ a group acting on $S$ by linear transformations. If $S_i \leq S$ is a subspace, we define the dual space $S_i^\bot$ to be the subspace of $S^* := \Hom(S,\F)$ consisting of those linear functions on $S$ that vanish on $S_i$, viz. $S_i^\bot := \{\ell \colon S \to \F | \ell(S_i) = 0\}$. It is a standard fact that two subspaces $S_1,S_2 \leq S$ are in the same $A$-orbit iff their duals $S_1^\bot$ and $S_2^\bot$ are in the same $A$-orbit under the contragradient action on $S^*$. 

Now apply the preceding paragraph to $A = \aut(F_{\Phi\text{-}2,2,n})$ acting on $S(n,2)$, with $S_1 = S_G$ and $S_2 = S_H$. We have $S_G^\bot,S_H^\bot \leq S(n,2)^*$. By Fact~\ref{fact:dual-action}, we may view the latter space as the space of quadratic forms $Q(n,2)$, and thus $G \cong H$ iff the quadratic form spaces $S_G^\bot$ and $S_H^\bot$ are isometric. 

Finally, we have $\dim_{\F_2} S_G^\bot = \dim \Phi(F_{\Phi\text{-}2,2,n}) - \dim S_G$. But by our description of $S_G$ above we have $\dim S_G = \dim \Phi(F_{\Phi\text{-}2,2,n}) - \dim \Phi(G)$, so we have $\dim S_G^\bot = \dim \Phi(G)$, as claimed. Finally, we may take $g_1,\dotsc,g_m$ to be any basis of $S_G^\bot$ and $h_1,\dotsc,h_m$ to be any basis of $S_H^\bot$. This completes the proof. \qedhere

\end{proof}

We note that, had we been willing to look at the spaces $S_G,S_H \leq S(n,2)$ instead, then we could have solved the group isomorphism problem in this class of groups by reduction to the symmetric matrix space isometry problem, much more similar to the alternating matrix space isometry problem used in the case of odd $p$ \cite{Sun23, IMQSZ, GQSS25}. However, the issue with that approach is that the latter spaces can have dimension quadratically larger than $\dim \Phi(G)$ (for example, if $\dim \Phi(G) = O(\dim G/\Phi(G))$), thus essentially negating any gains from a faster isometry algorithm. By taking the dual space instead, we guarantee that our spaces of quadratic forms (rather than matrices) have the same dimension as $\Phi(G)$. 

\section{Reduce quadratic form space isometry to quadratic form tuple isometry}

In this section, we reduce the quadratic form space isometry problem to the quadratic form tuple isometry problem. 

In Section~\ref{sec:basic_property_quad_form}, we prove some useful properties of quadratic forms and their partial derivative matrices. 
In Section~\ref{sec:reduction_restricted_quad_space}, we reduce the quadratic form space isometry problem to the restricted quadratic form space isometry problem for the convenience of the final reduction. 
In Section~\ref{sec:quadratic_form_tuple_construction}, we reduce the problem to the quadratic form tuple isometry problem.

\subsection{The partial derivative matrix of a quadratic form}\label{sec:basic_property_quad_form}

We first establish the relations between the coefficients of quadratic forms $f$ and $f^S$ for some invertible $S$. 


\begin{fact}\label{fact:quadratic_form_transform}
For an invertible matrix $S$ over $\F_2$, the value of $f^S[k, k]$ for any $1 \leq k \leq n$ is
\[\sum_{i = 1}^n f[i, i] S[i, k]^2 + \sum_{1 \leq i < j \leq n} f[i, j] S[i, k] S[j, k]\] 
and the value of $f^S[k, \ell]$ for any $1 \leq k < \ell \leq n$ is
 \[
 \sum_{1 \leq i < j \leq n} f[i, j] \left(S[i, k] S[j, \ell] + S[i, \ell] S[j, k]\right)\]
\end{fact}
\begin{proof}
By the definition of the quadratic form, we have
\[
\begin{split}
f =  & \sum_{1 \leq i \leq j \leq n} f[i, j] x_i x_j = \sum_{1 \leq i \leq j \leq n} f[i, j] \left(\sum_{k = 1}^n S[i, k] y_k\right) \left(\sum_{\ell = 1}^n S[j, \ell] y_\ell\right) \\
= & \sum_{k = 1}^n \left(\sum_{1 \leq i \leq j \leq n} f[i, j] S[i, k] S[j, k]\right) y_k^2 \\ & + \sum_{1 \leq k < \ell \leq n} \left(\sum_{1 \leq i \leq j \leq n} f[i, j](S[i, k] S[j, \ell] + S[i, \ell] S[j, k])\right) y_k y_\ell.
\end{split}
\]
Hence, we have 
\[\begin{split}
& \sum_{k = 1}^n \left(\sum_{1 \leq i \leq j \leq n} f[i, j] S[i, k] S[j, k]\right) y_k^2 \\
= & \sum_{k = 1}^n \left(\sum_{i = 1}^n f[i, i] S[i, k]^2 + \sum_{1 \leq i < j \leq n} f[i, j] S[i, k] S[j, k]\right) y_k^2 \\
\end{split}
\]
and 
\[
\begin{split}
&  \sum_{1 \leq k < \ell \leq n} \left(\sum_{1 \leq i \leq j \leq n} f[i, j](S[i, k] S[j, \ell] + S[i, \ell] S[j, k])\right) y_k y_\ell\\
= & \sum_{1 \leq k < \ell \leq n} \left(\sum_{i = 1}^n 2 f[i, i] S[i, k] S[i, l] + \sum_{1 \leq i < j \leq n} f[i, j] (S[i, k] S[j, \ell] + S[i, \ell] S[j, k])\right) y_k y_\ell. \\ 
= &  \sum_{1 \leq k < \ell \leq n} \left(\sum_{1 \leq i < j \leq n} f[i, j] (S[i, k] S[j, \ell] + S[i, \ell] S[j, k])\right) y_k y_\ell. \\ \qedhere
\end{split} 
\]
\end{proof}


Now we consider the partial derivative matrix for a quadratic form. 

\begin{definition}
    The \emph{partial derivative matrix} $D_f$ over $ \F_2$ for a quadratic form $f$ over $ \F_2$ is defined as follows: 
$D_f[i, i] = 0$ for all $i \in [n]$, $D_f[i, j] = D_f[j, i] = f[i, j]$ for all $1 \leq i < j \leq n$. 
\end{definition}

\begin{lemma}\label{lem:partial_derivative_matrix}
Let $f$ be a homogeneous quadratic form in $n$ variables over $ \F_2$, and $D_f$ be the partial derivative matrix of $f$. 
For an arbitrary invertible $n \times n$ matrix $S$ over $\F_2$, $S^T D_f S$ is the partial derivative matrix of $f^S$. 

In addition, for two quadratic forms $f_1$ and $f_2$ in $n$ variables over $ \F_2$. 
$D_{f_1} + D_{f_2}$ is the partial derivative matrix of $f_1 + f_2$. 
\end{lemma}
\begin{proof}
For any $1 \leq k, \ell \leq n$, $(S^T D_f S) [k, l]$ is 
\begin{align*}
\sum_{j = 1}^n \left(\sum_{i = 1}^n S[i, k] D_f[i, j]\right) S[j, \ell] = & \sum_{i = 1}^n D_f[i, i]S[i, k] S[i, \ell] +  \sum_{i \neq j} D_f[i, j]S[i, k] S[j, \ell] \\
=  & \sum_{1 \leq i < j \leq n} D_f[i, j] (S[i, k] S[j,\ell] + S[i, \ell] S[j, k]).
\end{align*}
If $k = \ell$, then $(S^T D_f S) [k, k] = 0$. Otherwise $(S^T D_f S) [k, \ell] = f^S[k, \ell]$ if $k < \ell$ or $(S^T D_f S) [k, \ell] = f^S[\ell, k]$ if $k > \ell$ by Fact~\ref{fact:quadratic_form_transform}. 

The second property is obtained by the definition of the partial derivative matrix for a quadratic form. 
\end{proof}

We use the following lemma for the square terms of quadratic forms under some special linear transforms.

\begin{lemma}\label{lem:lemma_square_term}
Let $f$ be a homogeneous quadratic form in $n$ variables over $\F_2$. 
Let $S$ be an invertible matrix equal to $I + S'$ such that for any $1 \leq i, j \leq n$,  $(D_f)_i \cdot S'[i, j] = 0$, where $(D_f)_i$
is the $i$-th column of $D_f$. 
Then $f^S$ satisfies the following conditions:
\begin{enumerate}
\item $D_{f^S} = D_{f}$.
\item
Let $a$ and $b$ be the column vectors 
\[a = \begin{pmatrix}f[1, 1] \\ \cdots \\ f[n, n]\end{pmatrix} \text{ and } b = \begin{pmatrix}f^S[1, 1] \\ \cdots \\ f^S[n, n]\end{pmatrix}.\]
We have $b = S^T a$.
\end{enumerate}
\end{lemma}
\begin{proof}
By Lemma~\ref{lem:partial_derivative_matrix}, 
\[D_{f^S} = S^T D_f S = D_f + S'^T D_f + D_f S' + S'^T D_f S' = D_f.\]
By Fact~\ref{fact:quadratic_form_transform}, \[f^S[k, k] = \sum_{i = 1}^n f[i, i] S[i, k]^2 + \sum_{1 \leq i < j \leq n} f[i, j] S[i, k] S[j, k].\]

We first show that $f[i, j] S[i, k] S[j, k] = 0$ for any $1 \leq i < j \leq n$. 
By the definition of $D_f$, we have \[
f[j, i] S'[i, k] = D_f[j, i] S'[i, k] = 0\]
and \[f[i, j] S'[j, k] = D_f[i, j] S'[j, k] = 0.\] Since $i$ is not equal to $j$, either $S[i, k]$ is equal to $S'[i,k]$, or $S[j, k]$ is equal to $S'[j, k]$. Hence, $f[i, j] S[i, k] S[j, k] = 0$ for any $1 \leq i < j \leq n$. 

Consequently, by Fact~\ref{fact:quadratic_form_transform},
we have 
\[f^S[k, k] = \sum_{i = 1}^n f[i, i]S[i, k]^2 =  \sum_{i = 1}^n f[i, i]S[i, k]\]
using the fact that $S[i, k]$ is in $\F_2$. 
\end{proof}

\subsection{Reduce quadratic form space isometry to non-degenerate restricted quadratic form space isometry}\label{sec:reduction_restricted_quad_space}

Using the reduction from the tensor isomorphism problem (Definition~\ref{def:tensor_iso}) to the restricted tensor isomorphism problem (Definition~\ref{def:restricted_tensor_iso}), we first reduce the quadratic form space isometry problem to the restricted quadratic form space isometry problem defined as follows.

\begin{problem}[Restricted quadratic form space isometry problem]\label{problem:restricted_quadratic_form_space_iso}
Given two sequences of quadratic forms $f_1, \dots f_m$ and $g_1, \dots, g_m$ both in $n$ variables such that $\dim(\langle f_1, \dots f_m\rangle) = \dim(\langle g_1, \dots g_m\rangle) = m$ and parameters $\alpha_0, \alpha_1, \beta_0, \beta_1$, 
determine if there are two matrices $M \in \text{GL}(m, \F_2)$ and $N \in \text{GL}(n, \F_2)$ satisfying Equation~(\ref{equ:restricted_tensor_iso}) such that $g_i = \sum_{j = 1}^m M[j, i] f_j^{N}$. 
\end{problem}

We prove the following theorem for the reduction. 


\begin{theorem}\label{thm:reduction_space_iso_to_restricted_space_iso}
For an input of the quadratic form space isometry problem with two quadratic form spaces of dimension $m$ on $n$ variables over $\F_p$ for any prime $p$, 
there is an algorithm with running time $p^{O((m + n)n^{0.5}\log p)}$ to reduce the quadratic form space isometry problem 
to $r = p^{O((m + n)n^{0.5}\log p)}$ instances of the restricted quadratic form space isometry problem such that for each instance, the input quadratic form spaces are of dimension $m$ on $n$ variables over $\F_p$, and the parameters $\alpha_0, \alpha_1, \beta_0, \beta_1$ satisfy
$\max\{m-\alpha_0-\alpha_1,n-\beta_0-\beta_1\} \leq n^{0.5}$. 
\end{theorem}

Let $f_1, \dots, f_m$ be a sequence of quadratic forms in $n$ variables. The partial derivative tensor of $f_1, \dots, f_m$ is an $m \times n \times n$ tensor $\mathfrak{T}$ such that $\mathfrak{T}[i, j, k] = D_{f_i}[j, k]$.

To prove Theorem~\ref{thm:reduction_space_iso_to_restricted_space_iso}, we consider the following algorithm.

\begin{framed}\begin{flushleft}
\vspace{-.3cm}\textbf{Quadratic Form Space Isometry Algorithm} 

\textbf{Input:} 
Quadratic forms $f_1, \dots, f_m$ and $g_1, \dots, g_m$ such that $f_1, \dots, f_m$ are linearly independent, and $g_1, \dots, g_m$ are linearly independent.

\textbf{Output:} Yes or no. 
\begin{enumerate}
\item Construct the partial derivative tensors $\mathfrak{T}$ from $f_1, \dots, f_m$ and $\mathfrak{T}'$ from $g_1, \dots, g_m$. Run the Tensor Isomorphism Reduction Algorithm on $\mathfrak{T}$ and $\mathfrak{T}'$, and let $A_i, B_i, A_i', B_i'$, $\alpha_{i, 0}$, $\alpha_{i, 1}$, $\beta_{i, 0}$, and $\beta_{i, 1}$ for $1 \leq i \leq \ell$ be the inputs for the restricted tensor isomorphism problem for some $\ell$. 
\item For each $1 \leq i \leq \ell$, 
run the Restricted Quadratic Form Space Isometry Algorithm on $f_1', \dots, f_m'$ and $g_1', \dots, g_m'$ with $f_\ell' = \sum_{j = 1}^m A_i[j, \ell] f_j^{B_{i}}$ and  $g_\ell' = \sum_{j = 1}^m A_i'[j, \ell] g_j^{B_{i}'}$ for any $1 \leq \ell \leq m$, and parameters $\alpha_{i, 0}, \alpha_{i, 1}, \beta_{i, 0}, \beta_{i, 1}$. Output yes if any output of the Restricted Quadratic Form Space Isometry Algorithm is yes. 
\item Output no. 
\end{enumerate}
\end{flushleft}
\vspace{-.3cm}\end{framed}

\begin{proof}[Proof of Theorem~\ref{thm:reduction_space_iso_to_restricted_space_iso}]
For any input $f_1', \dots, f_m'$ and $g_1', \dots, g_m'$ of the restricted quadratic form space isometry problem, there exist two invertible matrices $B_i$ and $B_i'$ such that 
\[\left\langle f_1', \dots, f_m'\right\rangle = \left\langle f_1^{B_i}, \dots, f_m^{B_i} \right\rangle \text{ and } \left\langle g_1', \dots, g_m'\right\rangle = \left\langle g_1^{B_i'}, \dots, g_m^{B_i'} \right\rangle.\]
Therefore, if there is no invertible matrix $S$ such  that $\langle f_1^S, \dots, f_m^S \rangle = \langle g_1, \dots, g_m \rangle$, then the restricted quadratic form space isometry problem for every input instance outputs no, and thus the algorithm outputs no. 

If the input of the Quadratic Form Space Isometry Algorithm is a yes instance, then there exist invertible matrices $U$ and $W$ such that $g_\ell = \sum_{j = 1}^m W[j, \ell] f_j^U$, and consequently $\mathfrak{T}_{\mathrm{Trans}(W, U)} = \mathfrak{T'}$. 
By Theorem~\ref{thm:tensor_iso_to_restricted_tensor_iso}, there exist $A_i, B_i, A_i', B_i', \alpha_{i, 0}, \alpha_{i, 1}, \beta_{i, 0}, \beta_{i, 1}$ as inputs for the restricted tensor isomorphism problem such that there exist $N$ and $M$ satisfying Equation~(\ref{equ:restricted_tensor_iso}) such that $U = B_i N B_i'^{-1}$ and $W = A_i M A_i'^{-1}$. Hence, the restricted quadratic form space isometry problem instance for $A_i, B_i, A_i', B_i', \alpha_{i, 0}, \alpha_{i, 1}, \beta_{i, 0}, \beta_{i, 1}$ outputs yes, and thus the Quadratic Form Space Isometry Algorithm outputs yes. 

The bound on the number of instances of the restricted quadratic form space isometry problem and the running time is obtained by Theorem~\ref{thm:tensor_iso_to_restricted_tensor_iso}. 
\end{proof}

For the convenience of the reduction to the quadratic form tuple isometry problem, we further reduce the restricted quadratic form space isometry problem to the \emph{non-degenerate restricted quadratic form space isometry problem} as Definition~\ref{def:non-degenerate_quad_form_space}.

For a quadratic form $f$ over $\F_2$ on $n$ variables and a column vector $v \in \F_2^n$, we say that $v$ is \emph{irrelevant} to $f$ if there exists an invertible matrix $S$ with $v$ as the first column such that $f^S[1, i] = 0$ for any $1 \leq i \leq n$. 
\begin{lemma}
    The following properties hold for any quadratic form $f$.
    \begin{enumerate}
        \item If $v$ is a column vector irrelevant to $f$, then any invertible matrix $U$ with $v$ as the first column satisfies $f^U[1, i] = 0$ for any $1 \leq i \leq n$. 
        \item For any two non-zero column vectors $v_1$ and $v_2$ both irrelevant to $f$, $v_1 + v_2$ is irrelevant to $f$.
    \end{enumerate}
\end{lemma}
\begin{proof}
    For the first property, by definition, there exists an invertible matrix $S$ with $v$ as the first column such that $f^S[1, i] = 0$ for any $1 \leq i \leq n$. 
    
    Then for any invertible matrix $U$ with $v$ as the first column, there exists another matrix 
    \[W = \begin{pmatrix}1 & W'' \\ & W'\end{pmatrix}\]
    for some $W'' \in M(1 \times (n-1), \F_2)$ and  $W' \in \text{GL}(n-1, \F_2)$ such that 
    \[S W = U.\] 
    Since $f^S[1, i] = 0$ for any $1 \leq i\leq n$ and \[f^U = f^{SW} = (f^S)^W,\] $f^U[1, i] = 0$ for any $1 \leq i \leq n$ by Fact~\ref{fact:quadratic_form_transform}. 

    For the second property, let $S$ be an arbitrary invertible matrix with $v_1$ as the first column and $v_2$ as the second column. By the first property, $f^S[1, i] = 0$ for any $1\leq i \leq n$. Since $f^{SW}[1, i] = 0$ for any $1 \leq i \leq n$ with \[W = \begin{pmatrix}& 1 \\ 1 \\ & & I\end{pmatrix},\]
    we have $f^S[2, i] = 0$ for any $2 \leq i \leq n$ by Fact~\ref{fact:quadratic_form_transform}. 
    Also by Fact~\ref{fact:quadratic_form_transform}, $f^{SY}[1, i] = 0$ for any $1 \leq i \leq n$ with 
    \[Y = \begin{pmatrix}1   \\ 1 & 1\\ & & I\end{pmatrix}.\]
    Hence, $v_1 + v_2$ is irrelevant to $f$. 
\end{proof}

We say that a column vector $v \in F_2^n$ is \emph{degenerate} for a quadratic form sequence $f_1, \dots, f_m$ over $n$ variables with respect to parameters $\alpha_0, \alpha_1, \beta_0, \beta_1$ if the following conditions hold:
\begin{enumerate}
\item Denote $v$ as $(v_1, \dots, v_n)^T$. $v_i$ is zero for any $i > \beta_0$. 
\item $v$ is irrelevant to $f_i$ for any $1 \leq i \leq \alpha_0 + \alpha_1$. 
\item $D_{f_i} v$ = 0 for any $\alpha_0 + \alpha_1 + 1 \leq i \leq m$. 
\end{enumerate}


We say that a sequence of quadratic forms $f_1, \dots, f_m$ with respect to parameters $\alpha_0, \alpha_1, \beta_0, \beta_1$ is \emph{degenerate} if there is a degenerate vector for the quadratic form sequence, otherwise, the quadratic form sequence is non-degenerate.

\begin{problem}[Non-degenerate restricted quadratic form space isometry problem]\label{def:non-degenerate_quad_form_space}
An instance of the restricted quadratic form space isometry problem is an instance of the non-degenerate restricted quadratic form space isometry problem if the two input quadratic form sequences are non-degenerate (with respect to the parameters of the input).
\end{problem}

We use the following algorithm to apply linear transforms for a sequence of quadratic forms to identify the degenerate vectors. 

\begin{framed}\begin{flushleft}
\vspace{-.3cm}\textbf{Degeneracy Identification Algorithm} 

\textbf{Input:} 
Quadratic forms $f_1, \dots, f_m$ over $n$ variables such that $f_1, \dots, f_m$ are linearly independent, and parameters $\alpha_0, \alpha_1, \beta_0, \beta_1$.

\textbf{Output:} Quadratic forms $f_1', \dots, f_m'$ over $n$ variables and an integer $d$.

\begin{enumerate}
    \item Compute the linear space spanned by all column vectors degenerate for $f_1, \dots, f_m$ with respect to parameters $\alpha_0, \alpha_1, \beta_0, \beta_1$, denoted as $V$. 
    \item Let $d = \dim(V)$ and compute a linear basis $v_1, \dots, v_d$ of $V$ such that \[a_{1, 1}, \dots, a_{1, m -\alpha_0 - \alpha_1}, a_{2,1}, \dots,a_{2, m -\alpha_0 - \alpha_1}, \dots, a_{d, 1}, \dots,  a_{d, m -\alpha_0 - \alpha_1}\] is lexically minimized for all the possible linear basis of $V$, where $a_{i, j}$ is $f_{\alpha_0 + \alpha_1 + j}^{N}[i, i]$ for any $N$ that is an arbitrary matrix satisfying Equation~(\ref{equ:restricted_tensor_iso}) with $v_i$ as $i$-th column for any $1 \leq i \leq d$.  
    \item Compute an arbitrary matrix $N$ satisfying Equation~(\ref{equ:restricted_tensor_iso}) with $v_i$ as $i$-th column for any $1 \leq i \leq d$, and return $f_1^N, \dots, f_m^N$ and the integer $d$.
\end{enumerate}
\end{flushleft}
\vspace{-.3cm}\end{framed}



We reduce the restricted
quadratic form space isometry problem to the non-degenerate restricted
quadratic form space isometry problem by the following algorithm. 
\begin{framed}\begin{flushleft}
\vspace{-.3cm}\textbf{Restricted Quadratic Form Space Isometry Algorithm} 

\textbf{Input:} 
Quadratic forms $f_1, \dots, f_m$ and $g_1, \dots, g_m$ such that $f_1, \dots, f_m$ are linearly independent and $g_1, \dots, g_m$ are linearly independent, and parameters $\alpha_0, \alpha_1, \beta_0, \beta_1$.

\textbf{Output:}
Yes or no.

\begin{enumerate}
    \item Run the Degeneracy Identification Algorithm for both $f_1, \dots, f_m$ and $g_1, \dots, g_m$ with parameters $\alpha_0, \alpha_1, \beta_0, \beta_1$, and denote the outputs as $f_1', \dots, f_m'$, $d_f$ and $g_1', \dots, g_m'$, $d_g$. Return no if $d_f \neq d_g$ or $f_{i}'[j, j] \neq g_{i}'[j, j]$ for some $\alpha_0 + \alpha_1 + 1 \leq i \leq m, 1 \leq j \leq d_f$.
    \item Let $g''_i$ be the induced quadratic form of $g_i'$ on  the last $n-d_f$ variables for any $1 \leq i \leq m$. 
    \item For any \[S = \begin{pmatrix} I_{d_f} & S'\\
    & I_{\beta_0 + \beta_1 - d_f} \\ & &  I_{n - \beta_0 - \beta_1}\end{pmatrix}\] with $S' \in \M(d_f \times (\beta_0 + \beta_1 - d_f), \F_2)$, let $f_i''$ be the induced quadratic form of $f_i'^{S}$ on the last $n - d_f$ variables for any $1 \leq i \leq m$, run the Non-Degenerate Restricted Quadratic Form Space Isometry Algorithm with $f_1'', \dots, f_m''$, $g_1'', \dots, g_m''$, and parameters $\alpha_0, \alpha_1, \beta_0- d_f, \beta_1$, and return yes if any output is yes. 
    \item Return no.
\end{enumerate}
\end{flushleft}
\vspace{-.3cm}\end{framed}

\begin{theorem}\label{thm:restricted_space_iso_algo}
For an input of the restricted quadratic form space isometry problem consisting of two quadratic form sequences $f_1, \dots, f_m$, $g_1, \dots, g_m$ in $n$ variables over $\F_2$ and parameters $\alpha_0, \alpha_1, \beta_0, \beta_1$, 
there is an algorithm with running time $2^{O((m - \alpha_0 - \alpha_1)\cdot n)}$ to reduce the restricted quadratic form space isometry problem 
to $2^{O((m - \alpha_0 - \alpha_1)\cdot n)}$ instances of the non-degenerate restricted quadratic form space isometry problem with input quadratic form sequences of length $m$ in $n'$ variables with parameters $\alpha_0, \alpha_1, \beta_0', \beta_1$ satisfying the following conditions:
\begin{enumerate}
    \item $0 \leq \beta_0' \leq \beta_0$. 
    \item $0 \leq n' \leq n$. 
    \item $n - n' = \beta_0 - \beta_0'$. 
\end{enumerate}
\end{theorem}
\begin{proof}
    We prove the correctness of the Restricted Quadratic Form Space Isometry Algorithm.
    For $f_1, \dots, f_m$ and $g_1, \dots, g_m$ as input of the Degeneracy Identification Algorithm, we have the following properties:
    \begin{enumerate}
        \item $f_i'[j, k] = 0$ for any $1\leq i\leq m$,
        $1 \leq j \leq d_f$, and $j < k \leq n$.
        \item $g_i'[j, k] = 0$ for any $1 \leq i \leq m$, $1 \leq j \leq d_g$, and $j < k \leq n$. 
        \item $f_i'[j, j] = 0$ for any $1\leq i\leq m - \alpha_0 - \alpha_1$,
        $1 \leq j \leq d_f$.
        \item $g_i'[j, j] = 0$ for any $1\leq i\leq m - \alpha_0 - \alpha_1$,
        $1 \leq j \leq d_g$.
        \item 
        $\langle e_1, \dots, e_{d_f}\rangle$ are degenerate columns for $f_1', \dots, f_m'$. 
        \item 
        $\langle e_1, \dots, e_{d_g}\rangle$ are degenerate columns for $g_1', \dots, g_m'$. 
    \end{enumerate}
    Let $V_f$ be the linear space spanned by degenerate columns for $f_1, \dots, f_m$, and $V_g$ be the linear space spanned by degenerate columns for $g_1, \dots, g_m$.
    If there exist matrices $M$ and $N$ that satisfy Equation~(\ref{equ:restricted_tensor_iso}) such that $g_i = \sum_{j = 1}^m M[i, j] f_j^{N}$, then $V_g = N^{-1}V_f$. 
    Then $d_f = d_g$ and $f'_{\alpha_0 + \alpha_1 + i}[j, j] = g'_{\alpha_0 + \alpha_1 + i}[j, j]$ for any $1 \leq i \leq m - \alpha_0 - \alpha_1$ and $1 \leq j \leq d_f$. 

    Hence, there must exist matrices $M$ and $N$  such that $g_i' = \sum_{j = 1}^m M[i, j] f_j'^{N}$, $M$ satisfies Equation~(\ref{equ:restricted_tensor_iso}), and 
    \[N = \begin{pmatrix}I_{d_f} & N_3 & N_4 \\ & N_1 & N_5 \\ & & N_2 \\ & & & I_{n - \beta_0 - \beta_1}\end{pmatrix}\]
    for some $N_1 \in \text{GL}(\beta_0 - d_f, \F_2)$ and $N_2 \in \text{GL}(\beta_1, \F_2)$,  $N_3 \in M(d_f \times (\beta_0 - d_f), \F_2)$, $N_4 \in M(d_f \times \beta_1, \F_2)$, and $N_5 \in M((\beta_0 - d_f) \times \beta_1, \F_2)$. 
    For an arbitrary such $N$, using matrix $S$ with \[S' = \begin{pmatrix}N_3 & N_4\end{pmatrix}\] for Step 3 of the Restricted Quadratic Form Space Isometry Algorithm, the Non-Degenerate Restricted Quadratic Form Space Isometry Algorithm for $f_1'', \dots, f_m''$ and $g_1'', \dots, g_m''$ outputs yes, and thus the overall algorithm outputs yes. 

    Now we show that the algorithm outputs yes only when the input instance is a yes instance. 
    Notice that the algorithm outputs yes only when the non-degenerate restricted quadratic form space isometry algorithm outputs yes for some $f_1'', \dots, f_m''$ obtained from some $S$.  
    If so, there exist $M$ and $N'$ such that $g_i'' = \sum_{j = 1}^m M[j, i] f_j''^{N'}$. 
    Therefore, $g_i' = \sum_{j = 1}^m M[j, i] f_j'^{N}$ with 
    \[N = S \begin{pmatrix}I_{d_f} & \\ & N'\end{pmatrix},\] and consequently, the input instance of the restricted quadratic form space isometry problem is a yes instance. 
\end{proof}

\subsection{Reduce non-degenerate restricted quadratic form space isometry to quadratic form tuple isometry}\label{sec:quadratic_form_tuple_construction}
We reduce the non-degenerate restricted quadratic form space isometry problem to the quadratic form tuple isometry problem by encoding the information of a non-degenerate restricted quadratic form space as a tuple of quadratic forms. 

\begin{definition}
    A quadratic form tuple $\mathbf{f}$ of length $m$ is an ordered sequence of $m$ quadratic forms defined on the same variables. 

    The partial derivative matrix tuple of $\mathbf{f}$, denoted as $\mathbf{D}_\mathbf{f}$, 
    is the matrix tuple $(D_1, \dots, D_m)$, where $D_i$ is the partial derivative matrix of $i$-th quadratic form of $\mathbf{f}$. 
    
\end{definition}

For a sequence of quadratic forms $f_1, \dots, f_m$ in $n$ variables $x_1, \dots, x_n$ with respect to parameters $\alpha_0, \alpha_1, \beta_0, \beta_1$, we use a quadratic form tuple to encode the non-degenerate restricted quadratic form space spanned by $f_1, \dots, f_m$. The construction is given in Figure~\ref{fig:constrution}. 
In the construction, coefficients are zero unless specified for each constructed quadratic form. We say that a quadratic form is of type $i(j)$ if it is constructed for the ($j$) copy of Step i in the construction. For example, if a quadratic form of the (b) copy is constructed in Step 1, then it is a type 1(b) quadratic form.
\begin{figure}
\begin{framed}\begin{flushleft}
\vspace{-.1cm}\textbf{Quadratic Form Tuple Construction} 

\textbf{Input:}
Quadratic forms $f_1, \dots, f_m$ on $n$ variables, and parameters $\alpha_0, \alpha_1, \beta_0, \beta_1$.

\textbf{Output:}
A quadratic form tuple on $2(\alpha_0 + \alpha_1 + \beta_0 + \beta_1) +3$ variables. 

\begin{enumerate}
\item Construct two copies of $m - \alpha_0 - \alpha_1$ quadratic forms such that for any $1 \leq i \leq j \leq \beta_1 + \beta_0$ and $1 \leq \ell \leq m - \alpha_0 - \alpha_1$, the $\ell$-th quadratic form has the coefficient of $x_i x_j$ for the $(a)$ copy, and $x_{\beta_0 + \beta_1 + i} x_{\beta_0 + \beta_1 + j}$ for the $(b)$ copy, equal to $f_{\ell + \alpha_0 + \alpha_1}[i, j]$. 




\item Construct four copies of $n - \beta_0$ quadratic forms such that for any $1 \leq i\leq \beta_1 + \beta_0$, $1 \leq j \leq \alpha_0 + \alpha_1$, and $1 \leq \ell \leq n - \beta_0$, the $\ell$-th quadratic form has the coefficient of 
$x_i x_{2\beta_0 +2 \beta_1 + j}$ for the (a) copy,
$x_{\beta_0 + \beta_1 + i} x_{2\beta_0 + 2\beta_1 + j}$ for the (b) copy,
$x_{i} x_{2\beta_0 + 2\beta_1 + \alpha_0 + \alpha_1 + j}$  for the (c) copy,
and $x_{\beta_0 + \beta_1 + i} x_{2\beta_0 + 2\beta_1 + \alpha_0 + \alpha_1 + j}$ for the (d) copy,
equal to $D_{f_{j}}[i, \beta_0 + \ell]$. 

\item Construct four copies of one quadratic form such that for any $1 \leq i\leq \beta_1 + \beta_0$ and $1 \leq j \leq \alpha_0 + \alpha_1$, 
the quadratic form has the coefficient of 
$x_i x_{2\beta_0 +2 \beta_1 + j}$ for the (a) copy,
$x_{\beta_0 + \beta_1 + i} x_{2\beta_0 + 2\beta_1 + j}$ for the (b) copy,
$x_{i} x_{2\beta_0 + 2\beta_1 + \alpha_0 + \alpha_1 + j}$  for the (c) copy,
and $x_{\beta_0 + \beta_1 + i} x_{2\beta_0 + 2\beta_1 + \alpha_0 + \alpha_1 + j}$ for the (d) copy, 
equal to $f_{j}[i, i]$. 

\item (Let $\gamma = 2\beta_0 + 2\beta_1 + 2 \alpha_0 + 2\alpha_1$.) Construct (a) a quadratic form such that only the coefficient of $x_{\gamma+ 1}x_{\gamma + 2}$ is $1$,  (b) a quadratic form such that only the coefficient of $x_{\gamma+ 1}x_{\gamma + 3}$ is $1$,  and (c) a quadratic form such that only the coefficient of $x_{\gamma+ 2}x_{\gamma + 3}$ is $1$. 

\item Construct $(m - \alpha_0 - \alpha_1) \cdot (n - \beta_0 - \beta_1 + 1)\cdot (n - \beta_0 - \beta_1) / 2$ quadratic forms such that 
for any $1 \leq k \leq m - \alpha_0 - \alpha_1, 1 \leq i \leq j \leq n - \beta_0 - \beta_1$, use one quadratic form to encode $f_{\alpha_0 + \alpha_1 + k}[\beta_0 + \beta_1 + i, \beta_0 + \beta_1 + j]$ as the coefficient of $x_{\gamma + 1} x_{\gamma + 2}$. 

\item Construct four copies of $(m - \alpha_0 - \alpha_1) \cdot (n - \beta_0 - \beta_1)$ quadratic forms such that the $\ell$-th quadratic form with $\ell = (\ell_1-1) \cdot (n - \beta_0 - \beta_1) + \ell_2 $ for any $1 \leq \ell_1 \leq m - \alpha_0 - \alpha_1$, $1 \leq \ell_2 \leq n - \beta_0 - \beta_1$ has the coefficient of $x_{i} x_{\gamma + 1}$ for the (a) copy, $x_{\beta_0 + \beta_1 + i}x_{\gamma+1}$ for the (b) copy, $x_{i} x_{\gamma + 2}$ for the (c) copy, and $x_{\beta_0 + \beta_1 + i}x_{\gamma+2}$ for the (d) copy, equal to $f_{\alpha_0 + \alpha_1 + \ell_1}[i, \beta_0 + \beta_1 + \ell_2]$, for any $1 \leq i \leq \beta_0 + \beta_1$. 

\item Construct four copies of $(n - \beta_0 - \beta_1 + 1)\cdot (n - \beta_0 - \beta_1) / 2$ quadratic forms such that 
for any $1 \leq \ell_1 \leq \ell_2 \leq n - \beta_0 - \beta_1$, use one quadratic form to encode $f_{\alpha_0 + \alpha_1 + i}[\beta_0 + \beta_1 + \ell_1, \beta_0 + \beta_1 + \ell_2]$ as the coefficient of $x_{2\beta_0 + 2\beta_1 + i} x_{\gamma + 1}$ for the (a) copy, 
the coefficient of $x_{2\beta_0 + 2\beta_1 + \alpha_0 + \alpha_1 + i} x_{\gamma + 1}$ for the (b) copy,
the coefficient of $x_{2\beta_0 + 2\beta_1 + \alpha_0 + \alpha_1 + i} x_{\gamma + 2}$ for the (c) copy,
and the coefficient of $x_{2\beta_0 + 2\beta_1 + \alpha_0 + \alpha_1 + i} x_{\gamma + 2}$ for the (d) copy
for any $1 \leq i \leq \alpha_0 + \alpha_1$. 

\item Construct four copies of $\beta_1$ quadratic forms such that for any $1 \leq \ell \leq \beta_1$,  the $\ell$-th quadratic form has the coefficient of $x_{\beta_0 + \ell} x_{\gamma + 1}$ for the (a) copy, 
$x_{2\beta_0 +\beta_1 + \ell} x_{\gamma + 1}$ for the (b) copy, 
$x_{\beta_0 + \ell} x_{\gamma + 2}$ for the (c) copy, and $x_{2\beta_0 +\beta_1 + \ell} x_{\gamma + 2}$ for the (d) copy, equal to $1$.


\item Construct four copies of $\alpha_1$ quadratic forms such that for any $1 \leq \ell \leq \alpha_1$, the $\ell$-th quadratic form has the coefficient of
$x_{2\beta_0 + 2\beta_1 + \alpha_0 + \ell} x_{\gamma + 1}$ for the (a) copy, $x_{2 \beta_0  + 2 \beta_1 + 2\alpha_0 + \alpha_1 + \ell} x_{\gamma + 1}$ for the (b) copy,  
$x_{2\beta_0 + 2\beta_1 + \alpha_0 + \ell} x_{\gamma + 2}$ for the (c) copy, and $x_{2 \beta_0  + 2 \beta_1 + 2\alpha_0 + \alpha_1 + \ell} x_{\gamma + 2}$ for the (d) copy, equal to 1. 
\end{enumerate}
\end{flushleft}
\vspace{-.3cm}\end{framed}
\vspace{-.3cm}\caption{Quadratic Form Tuple Construction}\label{fig:constrution}
\end{figure}

We prove some properties of our matrix tuple construction. 

\begin{lemma}\label{lem:space_to_tuple_basic_1}
For a quadratic form tuple $\mathbf{f}$ over $n$ variables and a linear space of column vectors $V$, if any of the following conditions is satisfied, then for any invertible matrices \[U = \begin{pmatrix}U_1 & \cdots & U_k\end{pmatrix} \text{ and } W = \begin{pmatrix}W_1 & \cdots & W_k\end{pmatrix}\] such that $U^T \mathbf{D}_\mathbf{f} W = \diag(U_1^T \mathbf{D}_\mathbf{f}W_1, U_k^T \mathbf{D}_\mathbf{f}W_k)$, the projection of any $v \in V$ on any $U_i$ (or $W_i$) with respect to $U_1 \oplus \dots \oplus U_k$ (or $W_1 \oplus \dots \oplus W_k$) is in $V$. 
\begin{enumerate}
\item $V = \langle v : D_{f} v = 0\rangle$ for some $f$ in $\mathbf{f}$. 
\item $V = V_1 \cap V_2$ for some linear spaces $V_1$ and $V_2$ such that for any invertible matrices \[U' = \begin{pmatrix}U_1' & \cdots & U_{k'}'\end{pmatrix} \text{ and } W' = \begin{pmatrix}W_1' & \cdots & W_{k'}'\end{pmatrix}\] satisfying $U'^T \mathbf{D}_\mathbf{f} W' = \diag(U_1'^T \mathbf{D}_\mathbf{f}W_1', U_{k'}'^T \mathbf{D}_\mathbf{f}W_{k'}')$,  for any $1 \leq i \leq 2$ and $1 \leq j \leq k'$, the projection of any $v \in V_i$ on any $U_j'$ (or $W_j'$) with respect to $U_1' \oplus \dots \oplus U_{k'}'$ (or $W_1' \oplus \dots \oplus W_{k'}'$) is in $V_i$. 
\end{enumerate}
\end{lemma}
\begin{proof}
    For the first property, let $v$ be an arbitrary vector such that $D_{f}v = 0$, and $v_j$ be the projection of $v$ on $W_j$ with respect to $W_1\oplus \dots\oplus W_k$ for any $1 \leq j \leq k$. 
    We have \[U^T D_{f} v = \begin{pmatrix}U_1^T D_{f} v_1 \\ \cdots \\ U_k^T D_{f} v_k\end{pmatrix} = 0. \]
    Hence, $U_j^T D_{f} v_j = 0$ for any $1 \leq j \leq k$. On the other hand, $U_j^T  D_{f} v_\ell = 0$ for any $1 \leq j, \ell \leq k$ such that $j \neq \ell$. 
    So we have $U^T D_f v_j = 0$ for any $1 \leq j \leq k$, and thus  $D_f v_j = 0$ for any $1 \leq j \leq k$. 
    Hence, $v_j$ is in $V$ for any $1 \leq j \leq k$. Similarly, we can show that the projection of any $v \in V$ on any $U_j$ with respect to $U_1 \oplus \dots \oplus U_k$ is in $V$.

    For the second property, let $U' = U$ and $W' = W$. 
    Then for any $v \in V_1 \cap V_2$, 
    the projection of $v$ on any $U_j'$ (or $W_j'$) with respect to $U_1' \oplus \dots \oplus U_{k}'$ (or $W_1' \oplus \dots \oplus W_{k}'$) is in $V_1 \cap V_2$. 
\end{proof}

\begin{lemma}\label{lem:space_to_tuple_basic_0}
For two quadratic form tuples $\mathbf{f}$ and $\mathbf{g}$ on $n$ variables and two linear spaces of column vectors $V$ and $V'$, if any of the following conditions is satisfied, then for any invertible matrix $U$ such that $\mathbf{f}^{U} = \mathbf{g}$, $U^{-1} V = V'$. 
\begin{enumerate}
\item There exists an integer $i$ such that $V = \langle v : D_{\mathbf{f}, i} v = 0\rangle$ and $V' = \langle v : D_{\mathbf{g}, i} v = 0\rangle$, where $D_{\mathbf{f}, i}$ and $D_{\mathbf{g}, i}$ are the partial derivative matrices of the $i$-th quadratic forms of $\mathbf{f}$ and $\mathbf{g}$ respectively. 
\item Let $V_1, V_2$ and $V_1', V_2'$ be linear spaces such that for any invertible matrix $U$ satisfying $\mathbf{f}^{U} = \mathbf{g}$,  $U^{-1} V_i = V'_i$ for any $1 \leq i \leq 2$. 
Let $V = V_1 \cap V_2$ and $V' = V_1' \cap V_2'$. 
\end{enumerate}
\end{lemma}
\begin{proof}

For the first case,
since $\mathbf{f}^U = \mathbf{g}$, by Lemma~\ref{lem:partial_derivative_matrix}, for any vector $v \in V$, 
\[\mathbf{D}_{\mathbf{f}, i} v = \mathbf{D}_{\mathbf{g}^{U^{-1}}, i} v = U^{-T}\mathbf{D}_{\mathbf{g}, i} U^{-1} v = 0.\] 
Since $U$ is invertible, $U^{-1} v$ is in $V'$. 
Conversely, for any $v' \in V'$, \[\mathbf{D}_{\mathbf{g}, i} v' = \mathbf{D}_{\mathbf{f}^{U}, i} v' = U^{T}\mathbf{D}_{\mathbf{f}, i} U v' = 0.\]
Therefore, $U v'$ is in $V$. So, $U^{-1} V = V'$.

For the second case, for any $v \in V_1 \cap V_2$, we have that $U^{-1} v$ is in both $V_1' $ and $V_2'$, and thus in $V_1' \cap V_2'$. Conversely, for any $v' \in V_1' \cap V_2'$, we have that $U v'$ is in both $V_1$ and $V_2$, and thus in $V_1 \cap V_2$. Hence, $U^{-1} (V_1 \cap V_2) = V_1' \cap V_2$.
\end{proof}

Let $E_1 = \langle e_1, \dots, e_{\beta_0 + \beta_1} \rangle$, $E_2 = \langle e_{\beta_0 + \beta_1 + 1}, \dots, e_{2\beta_0 +2 \beta_1} \rangle$, $E_3 = \langle e_{2\beta_0 + 2\beta_1 + 1}, \dots, e_{2\beta_0 +2 \beta_1 + \alpha_0 + \alpha_1} \rangle$, $E_4 = \langle e_{2\beta_0 + 2\beta_1 + \alpha_0 + \alpha_1 + 1}, \dots, e_{2\beta_0 +2 \beta_1 + 2\alpha_0 + 2\alpha_1} \rangle$, $E_5 = \langle e_{\gamma + 1}\rangle$, $E_6 = \langle e_{\gamma + 2}\rangle$, and $E_7 = \langle e_{\gamma + 3}\rangle$.

\begin{lemma}\label{lem:space_to_tuple_basic_2}
    Let $\mathbf{f}$ be a quadratic form tuple constructed by the Quadratic Form Tuple Construction process.
    Then for any non-zero column vector $v$, there exists a quadratic form $f \in \mathbf{f}$ such that $D_{f} v$ is a non-zero vector.  
\end{lemma}
\begin{proof}
    Since the input to the Quadratic Form Tuple Construction process is non-degenerate with respect to the parameters, for any $1 \leq i \leq 7$ and any non-zero vector $v \in E_i$, there exists a quadratic form $f$ in the tuple $\mathbf{f}$ such that $D_f v \neq 0$. 

Now we show that for any column vector $v$, there exists some quadratic form $f$ in $\mathbf{f}$ such that $D_f v \neq 0$. 
Let $v_1, \dots, v_7$ be the projection of $v$ on $E_1, \dots, E_7$ with respect to $E_1\oplus \dots \oplus E_7$. 
By the construction of the quadratic form tuple, for each quadratic form $f$ in $\mathbf{f}$, if $D_f v_i \neq 0$, then $D_f v \neq 0$. Therefore, the lemma holds. 
\end{proof}

We also define linear spaces $I_{\mathbf{f}, 1}, \dots, I_{\mathbf{f}, 7}$ for a quadratic form tuple $\mathbf{f}$ obtained by the Quadratic Form Tuple Construction process as follows:
\begin{enumerate}
    \item $I_{\mathbf{f}, 1}$ is the linear space spanned by all the column vectors $v$ such that $D_fv = 0$ for any $f$ in $\mathbf{f}$ of type 1(b), 2(b), 2(d), 3(b), 3(d), 4, 5, 6(b), 6(d), 7, 8(b), 8(d), or 9.
    \item  $I_{\mathbf{f}, 2}$ is the linear space spanned by all the column vectors $v$ such that $D_fv = 0$ for any $f$ in  $\mathbf{f}$ of type 1(a), 2(a), 2(c), 3(a), 3(c), 4, 5, 6(a), 6(c), 7, 8(a), 8(c), or 9.
    \item $I_{\mathbf{f}, 3}$ is the linear space spanned by all the column vectors $v$ such that $D_fv = 0$ for any $f$ in  $\mathbf{f}$ of type 1, 2(c), 2(d), 3(c), 3(d), 4, 5, 6, 7(b), 7(d), 8, 9(b), or 9(d).
    \item  $I_{\mathbf{f}, 4}$ is the linear space spanned by all the column vectors $v$ such that $D_fv = 0$ for any $f$ in  $\mathbf{f}$ of type 1, 2(a), 2(b), 3(a), 3(b), 4, 5, 6, 7(a), 7(c), 8, 9(a), or 9(c). 
    \item $I_{\mathbf{f}, 5}$ is the linear space spanned by all the column vectors $v$ such that $D_fv = 0$ for any $f$ in  $\mathbf{f}$ of type 1, 2, 3, 4(c), 6(c), 6(d), 7(c), 7(d), 8(c), 8(d), 9(c), or 9(d).
    \item $I_{\mathbf{f}, 6}$ is the linear space spanned by all the column vectors $v$ such that $D_fv = 0$ for any $f$ in  $\mathbf{f}$ of type 1, 2, 3, 4(b), 6(a), 6(b), 7(a), 7(b), 8(a), 8(b), 9(a), or 9(b).
    \item $I_{\mathbf{f}, 7}$ is the linear space spanned by all the column vectors $v$ such that $D_fv = 0$ for any $f$ in  $\mathbf{f}$ of type 1, 2, 3, 4(a), 5, 6, 7, 8, or 9. 
\end{enumerate}

\begin{lemma}\label{lem:space_equivalence}
For a sequence of quadratic forms $f_1, \dots, f_m$ and parameters $\alpha_0, \alpha_1, \beta_0, \beta_1$, 
let $\mathbf{f}$ be a quadratic form tuple obtained by the Quadratic Form Tuple Construction process for  $f_1, \dots, f_m$. 
Then $E_i = I_{\mathbf{f}, i}$ for any $1 \leq i \leq 7$. 
\end{lemma}

\begin{proof}
We first prove some properties for the quadratic form tuple obtained by the Quadratic Form Tuple. 
We have the following observations:

\begin{enumerate}
\item For any non-zero $v \in E_1$, there is a quadratic form $f$ of type 1(a), 2(a), 2(c), 3(a), 3(c), 6(a), 6(c), 8(a) or 8(c) such that $D_f v$ is non-zero.
\item For any non-zero $v \in E_2$, there is a quadratic form $f$ of type 1(b), 2(b), 2(d), 3(b), 3(d), 6(b), 6(d), 8(b), or 8(d) such that $D_f v$ is non-zero.
\item For any non-zero $v \in E_3$, there is a quadratic form $f$ of type 2(a), 2(b), 3(a), 3(b), 7(a), 7(c), 9(a), or 9(c) such that $D_f v$ is non-zero.
\item For any non-zero $v \in E_4$, there is a quadratic form $f$ of type 2(c), 2(d), 3(c), 3(d), 7(b), 7(d), 9(b), or 9(d) such that $D_f v$ is non-zero.
\end{enumerate}

For any quadratic form $f$ of type $\dagger$  and linear space $V$ specified below, $D_f v$ is zero. 

\begin{enumerate}
\item $V = \langle E_2 \cup E_3 \cup E_4 \cup E_5 \cup E_6 \cup E_7\rangle$, and $\dagger = $ 1(a).
\item $V = \langle E_1 \cup E_3 \cup E_4 \cup E_5 \cup E_6 \cup E_7\rangle$, and $\dagger = $ 1(b).  
\item $V = \langle E_2 \cup E_4 \cup E_5 \cup E_6 \cup E_7\rangle$, and $\dagger = $ 2(a) or 3(a).
\item $V = \langle E_2 \cup E_3 \cup E_5 \cup E_6 \cup E_7\rangle$, and $\dagger = $ 2(b) or 3(b).
\item $V = \langle E_1 \cup E_4 \cup E_5 \cup E_6 \cup E_7\rangle$, and $\dagger = $ 2(c) or 3(c). 
\item $V = \langle E_1 \cup E_3 \cup E_5 \cup E_6 \cup E_7\rangle$, and $\dagger = $ 2(d) or 3(d). 
\item $V = \langle E_1 \cup E_2 \cup E_3 \cup E_4 \cup E_7\rangle$, and $\dagger = $ 4(a). 
\item 
$V = \langle E_1 \cup E_2 \cup E_3 \cup E_4 \cup E_6\rangle$, and $\dagger = $ 4(b).
\item 
$V = \langle E_1 \cup E_2 \cup E_3 \cup E_4 \cup E_5\rangle$, and $\dagger= $ 4(c). 
\item $V = \langle E_1 \cup E_2 \cup E_3 \cup E_4 \cup E_7\rangle$, and $\dagger = $ type 5.
\item $V = \langle E_2 \cup E_3 \cup E_4 \cup E_6 \cup E_7\rangle$, and $\dagger = $ 6(a) or 8(a).
\item $V = \langle E_1 \cup E_3 \cup E_4 \cup E_6 \cup E_7\rangle$, and $\dagger = $ 6(b) or 8(b).
\item $V = \langle E_2 \cup E_3 \cup E_4 \cup E_5 \cup E_7\rangle$, and $\dagger = $ 6(c) or 8(c). 
\item $V = \langle E_1 \cup E_3 \cup E_4 \cup E_5 \cup E_7\rangle$, and $\dagger = $ 6(d) or 8(d). 
\item $V = \langle E_1 \cup E_2 \cup E_4 \cup E_6 \cup E_7\rangle$, and $\dagger = $ 7(a) or 9(a).
\item $V = \langle E_1 \cup E_2 \cup E_3 \cup E_6 \cup E_7\rangle$, and $\dagger = $ 7(b) or 9(b).
\item $V = \langle E_1 \cup E_2 \cup E_4 \cup E_5 \cup E_7\rangle$, and $\dagger = $ 7(c) or 9(c). 
\item 
$V = \langle E_1 \cup E_2 \cup E_3 \cup E_5 \cup E_7\rangle$, and $\dagger = $ 7(d) or 9(d). 
\end{enumerate}
Then the lemma follows by Lemma~\ref{lem:space_to_tuple_basic_0}. 
\end{proof}

\begin{lemma}\label{lem:reduction_space_to_tuple_decomposition}
The following properties hold for the quadratic form tuple constructed by the Quadratic Form Tuple Construction process. 
\begin{enumerate}
\item 
For
any 
quadratic form tuple $\mathbf{f}$ constructed by the  Quadratic Form Tuple Construction process, 
$1 \leq i \leq 7$, $v \in E_i$, and any invertible matrices \[U = \begin{pmatrix}U_1 & \cdots & U_k\end{pmatrix} \text{ and } W = \begin{pmatrix}W_1 & \cdots & W_k\end{pmatrix}\] such that $U^T \mathbf{D}_\mathbf{f} W = \diag(U_1^T \mathbf{D}_\mathbf{f}W_1, U_k^T \mathbf{D}_\mathbf{f}W_k)$,
the projection of $v$ on $U_j$ (or $W_j$) with respect to $U_1 \oplus \dots \oplus U_k$ (or $W_1 \oplus \dots \oplus W_k$) is in $E_i$. 
\item 
For two sequences of quadratic forms $f_1, \dots, f_m$ and $g_1, \dots, g_m$ on $n$ variables with parameters $\alpha_0, \alpha_1, \beta_0, \beta_1$, 
let $\mathbf{f}$ and $\mathbf{g}$ be the quadratic form tuples constructed by the  Quadratic Form Tuple Construction process for  $f_1, \dots, f_m$ and $g_1, \dots, g_m$ respectively. For any invertible matrix $U$ such that $\mathbf{f}^{U} = \mathbf{g}$, 
$U^{-1} E_i = E_i$ for any $1 \leq i \leq 7$. 
\end{enumerate}
\end{lemma}
\begin{proof}
For the first property, for any $1 \leq i \leq 7$, 
let $v$ be an arbitrary vector in $I_{\mathbf{f}, i}$.
By the definition of $I_{\mathbf{f}, i}$ and Lemma~\ref{lem:space_to_tuple_basic_1}, the projection of $v$ on any $U_i$ (or $W_i$) with respect to $U_1 \oplus \dots \oplus U_k$ (or $W_1 \oplus \dots \oplus W_k$) is in $I_{\mathbf{f}, i}$.  
On the other hand, by Lemma~\ref{lem:space_equivalence}, $I_{\mathbf{f}, i} = E_i$. 
Then the first case holds. 

For the second property, since $\mathbf{f}$ and $\mathbf{g}$ are constructed with respect to the same parameters, $\mathbf{f}$ and $\mathbf{g}$ are of the same length, and for each $\ell$, $\mathbf{f}_\ell$ is of the same type as $\mathbf{g}_\ell$. 
Hence, 
for any $1 \leq i \leq 7$, 
by the definition of $I_{\mathbf{f}, i}$ and Lemma~\ref{lem:space_to_tuple_basic_0}, 
$U^{-1} I_{\mathbf{f}, i} = I_{\mathbf{g}, i}$. 
By Lemma~\ref{lem:space_equivalence}, $U^{-1}E_i = E_i$. 
\end{proof}

We show that the non-degenerate restricted quadratic form space isometry problem is reducible to the quadratic form tuple isometry problem if the input quadratic form tuples of the quadratic form tuple isometry problem are constructed by the Quadratic Form Tuple Construction process. 
\begin{lemma}
Let $f_1, \dots, f_m$ and $g_1, \dots, g_m$ be two sequences of quadratic forms with the same parameters $\alpha_0, \alpha_1, \beta_0, \beta_1$ as the input for the non-degenerate restricted quadratic form space isometry problem. 
Let $\mathbf{f}$ and $\mathbf{g}$ be the quadratic form tuples constructed by the Quadratic Form Tuple Construction process for $f_1, \dots, f_m$ and $g_1, \dots, g_m$ respectively. 
The following properties hold. 
\begin{enumerate}
\item For any  $M$ and $N$ satisfying Equation~(\ref{equ:restricted_tensor_iso}) 
with $M_1$ and $N_1$ satisfying Equation~(\ref{equ:restricted_tensor_iso_1}) and Equation~(\ref{equ:restricted_tensor_iso_2}) respectively
such that $g_i = \sum_{j = 1}^m M[j, i] f_j^{N}$ for any $1 \leq i \leq m$, then $\mathbf{f}^S = \mathbf{g}$ for 
\begin{equation}\label{equ:space_to_tuple}S = \begin{pmatrix} M_1 \\ & M_1 \\ & & N_1 \\ & & & N_1 \\ & & & & I_3\end{pmatrix}.\end{equation}

\item For any $S$ such that $\mathbf{f}^S = \mathbf{g}$, 
$S$ satisfies Equation~(\ref{equ:space_to_tuple}) with $M_1$ and $N_1$ satisfying Equation~(\ref{equ:restricted_tensor_iso_1}) and Equation~(\ref{equ:restricted_tensor_iso_2}) respectively and
$g_i = \sum_{j = 1}^m M[j, i] f_j^{N}$ for any $1 \leq i \leq m$ with $M$ and $N$ defined by Equation~(\ref{equ:restricted_tensor_iso}).
\end{enumerate}
\end{lemma}
\begin{proof}
By the construction of $\mathbf{f}$ and $\mathbf{g}$, the first property holds. 

For the second property, if $\mathbf{f}^S = \mathbf{g}$, then by Lemma~\ref{lem:reduction_space_to_tuple_decomposition}, we have 
\[
S = \begin{pmatrix}M_1 \\ & M_2 \\ & & N_1 \\ & & & N_2 \\ & & & & I_3 \end{pmatrix}
\] 
for invertible $M_1, M_2, N_1$ and $N_2$. 
By type 8 and 9 quadratic forms for $\mathbf{f}$ and $\mathbf{g}$, both $M_1$ and $M_2$ are of the form Equation~(\ref{equ:restricted_tensor_iso_1}), and 
both $N_1$ and $N_2$ are of the form Equation~(\ref{equ:restricted_tensor_iso_2}). 
By the construction of $\mathbf{f}$ and $\mathbf{g}$, since $f_1, \dots, f_m$ and $g_1, \dots, g_m$ are non-degenerate, $M_1 = M_2$ and $N_1 = N_2$. Then the lemma follows.
\end{proof}

Finally, we show that the quadratic form tuple constructed by the Quadratic Form Tuple Construction process has non-zero coefficients for square terms only for type 1 quadratic forms, even under certain linear transforms of the quadratic form tuple.

\begin{definition}\label{def:non-zero-square-term-condition}
A quadratic form tuple $\mathbf{f} = (f_1, \dots, f_m)$ on $n$ variables satisfies the \emph{$\rho$-non-zero-square-term condition} with respect to a matrix tuple $\mathbf{A}$ if for any invertible matrix $S$ such that $\mathbf{D}_{\mathbf{f}^S} = \mathbf{A} W$ for some invertible matrix $W$, $f^{S}_i [j, j] = 0$ for any $i > \rho$ and $1 \leq j \leq n$. 
\end{definition}

\begin{lemma}\label{lem:nonzero_suare_term_condition}
    Let $\mathbf{f}$ and $\mathbf{g}$ be two quadratic form tuples on $n_0$ variables and with the same length constructed by the Quadratic Form Tuple Construction process for $f_1, \dots, f_m$ and $g_1, \dots, g_m$ on $n$ variables using the same parameters $\alpha_0, \alpha_1, \beta_0, \beta_1$, and let $U$ be an arbitrary invertible matrix such that every column of $U$ is in one of $E_1$, $E_2$, $E_3$, $E_4$, $E_5$, $E_6$, and $E_7$. 
    Let $\mathbf{A} = U^T \mathbf{D}_\mathbf{f} U$. 
    Both $\mathbf{f}$ and $\mathbf{g}$ satisfy the $(m - \alpha_0 - \alpha_1)$-non-zero-square-term condition with respect to $\mathbf{A}$.
\end{lemma}
\begin{proof}
    Let $\mathbf{f}_i$ and $\mathbf{g}_i$ denote the $i$-th quadratic form of $\mathbf{f}$ and $\mathbf{g}$ respectively. 
    By the construction of $\mathbf{f}$ and $\mathbf{g}$, for any invertible $S$ such that every column of $S$ is from one of $E_1, E_2, E_3, E_4, E_5, E_6$, or $E_7$, $\mathbf{f}^S$ and $\mathbf{g}^S$ satisfy $\mathbf{f}^S_i[j, j] = \mathbf{g}^S_i[j, j] = 0$ for any $i > m - \alpha_0 - \alpha_1$ and $1 \leq j \leq n_0$.
    Therefore, it is sufficient to show that for any $\mathbf{h}\in \{\mathbf{f}, \mathbf{g}\}$, and any invertible $S$ such that there exists an invertible matrix $W$ satisfying $\mathbf{D}_\mathbf{h}^S = \mathbf{A} W$, every column of $S$ is from one of $E_i$ for some $1 \leq i \leq 7$. 

    For any $1 \leq j \leq n_0$, by the definition of $U$, the $j$-th column of $U$, denoted as $u_j$, is in $E_\ell$ for some $1 \leq \ell \leq 7$. 
    By Lemma~\ref{lem:space_equivalence}, $u_j$ is also in $I_{\mathbf{f}, \ell}$. 

    Since $\mathbf{f}$ and $\mathbf{g}$ are constructed with respect to the same parameters, $\mathbf{f}$ and $\mathbf{g}$ are of the same length, and for each $i$, $\mathbf{f}_i$ is of the same type as $\mathbf{g}_i$. 

    For any $\mathbf{h} \in \{\mathbf{f}, \mathbf{g}\}$, 
    since \[\mathbf{D}_\mathbf{h}^S = S^T \mathbf{D}_\mathbf{h} S = \mathbf{A} W = U^T \mathbf{D}_\mathbf{f} U W,\] 
    $s_j^T \mathbf{D}_{\mathbf{h}} = u_j^T \mathbf{D}_\mathbf{f} U W S^{-1}$, where $s_j$ denotes the $j$-th column of $S$. 
    Hence, $s_j$  must be in $I_{\mathbf{h}, \ell}$ by the definition of $I_{\mathbf{h}, \ell}$. 
    By Lemma~\ref{lem:space_equivalence}, the $j$-th column of $S$ is from $E_\ell$ for any $1 \leq j \leq n_0$. 
\end{proof}

In the following sections, we prove the following theorem. 

\begin{theorem}\label{thm:matrix_tuple_constructed_main}
Let $f_1, \dots, f_m$ and $g_1, \dots, g_m$ be two sequences of quadratic forms over $n$ variables with the same parameters $\alpha_0, \alpha_1, \beta_0, \beta_1$ as the input for the non-degenerate restricted quadratic form space isometry problem. 
Let $\mathbf{f}$ and $\mathbf{g}$ be the quadratic form tuples obtained by the Quadratic Form Tuple Construction Algorithm for $f_1, \dots, f_m$ and $g_1, \dots, g_m$ respectively. 
There is an algorithm to determine whether there is an invertible matrix $S$ such that $\mathbf{f}^S = \mathbf{g}$ in time \[2^{O((n+m)^{1.5} + (n+m)\cdot (m - \alpha_0 -\alpha_1))}.\]
\end{theorem}

With Theorem~\ref{thm:matrix_tuple_constructed_main}, we obtain Theorem~\ref{thm:quadratic_space_isometry}. 
\begin{proof}[Proof of Theorem~\ref{thm:quadratic_space_isometry}]
    The correctness of the theorem is obtained by Theorem~\ref{thm:reduction_space_iso_to_restricted_space_iso}, Theorem~\ref{thm:restricted_space_iso_algo}, and Theorem~\ref{thm:matrix_tuple_constructed_main}. 

    We bound the running time of the algorithm. 
    By Theorem~\ref{thm:reduction_space_iso_to_restricted_space_iso} and Theorem~\ref{thm:restricted_space_iso_algo}, the problem reduces to 
    \[2^{O((m + n)n^{0.5})}\cdot 2^{O(n^{1.5})} = 2^{O((m + n)n^{0.5})}\]
    instances of the following problem: determine the isometry of two quadratic form tuples $\mathbf{f}$ and $\mathbf{g}$ constructed by the Quadratic Form Tuple Construction Algorithm for $f_1, \dots, f_m$ and $g_1, \dots, g_m$ with parameters $\alpha_0, \alpha_1, \beta_0, \beta_1$ satisfying $m - \alpha_0 - \alpha_1 \leq n^{0.5}$.
    By Theorem~\ref{thm:matrix_tuple_constructed_main}, the overall running time is upper bounded by 
    \[2^{O((m+n)n^{0.5})}\cdot 2^{O((n+m)^{1.5} + (n+m)\cdot (m - \alpha_0 -\alpha_1))}\cdot \mathrm{poly}(n, m) = 2^{O((m+n)^{1.5})}.\]
\end{proof}

To prove Theorem~\ref{thm:quadratictuple_isometry_intro}, we consider the following algorithm.

\begin{framed}\begin{flushleft}
\vspace{-.3cm}\textbf{Quadratic Form Tuple Isometry Algorithm} 

\textbf{Input:} 
Quadratic forms $f_1, \dots, f_m$ and $g_1, \dots, g_m$ such that $f_1, \dots, f_m$ are linearly independent, and $g_1, \dots, g_m$ are linearly independent.

\textbf{Output:} Yes or no, according to whether there exists $S \in \GL_n(\F_2)$ such that $f_i^S = g_i$ for all $i$. 
\begin{enumerate}
\item Construct the partial derivative tensors $\mathfrak{T}$ from $f_1, \dots, f_m$ and $\mathfrak{T}'$ from $g_1, \dots, g_m$. Run the Tensor Isomorphism Reduction Algorithm on $\mathfrak{T}$ and $\mathfrak{T}'$, and let $A_i, B_i, A_i', B_i'$, $\alpha_{i, 0}$, $\alpha_{i, 1}$, $\beta_{i, 0}$, and $\beta_{i, 1}$ for $1 \leq i \leq \ell$ be the inputs for the restricted tensor isomorphism problem for some $\ell$. 
\item For each $1 \leq i \leq \ell$, if there is a matrix $M$ satisfying Equation~(\ref{equ:restricted_tensor_iso}) such that $A_i = M A_i'$, then 
run the Restricted Quadratic Form Space Isometry Algorithm on $f_1', \dots, f_m'$ and $g_1', \dots, g_m'$ with $f_j' = \sum_{j = 1}^m A_i[j, i] f_j^{B_{i}}$ and  $g_j' = \sum_{j = 1}^m A_i[j, i] g_j^{B_{i}'}$ for any $1 \leq j \leq m$, and parameters $\alpha_{i, 0}', \alpha_{i, 1}', \beta_{i, 0}, \beta_{i, 1}$ with $\alpha_{i, 0}' = 0$ and $\alpha_{i, 1}' = \alpha_{i, 0} + \alpha_{i, 1}$. Output yes if any output of the Restricted Quadratic Form Space Isometry Algorithm is yes. 
\item Output no. 
\end{enumerate}
\end{flushleft}
\vspace{-.3cm}\end{framed}

\begin{proof}[Proof of Theorem~\ref{thm:quadratictuple_isometry_intro}]
For any input $f_1', \dots, f_m'$ and $g_1', \dots, g_m'$ of the restricted quadratic form space isometry problem, there exist three invertible matrices $A_i$, $B_i$ and $B_i'$ such that 
\[f_\ell' =  \sum_{j = 1}^m A_i[j, \ell] f_j^{B_{i}} \text{ and } g_\ell' =  \sum_{j = 1}^m A_i[j, \ell] g_j^{B_{i}'}, \]
for any $1 \leq \ell \leq m$. 
Therefore, if there is no invertible matrix $S$ such  that $(f_1^S, \dots, f_m^S ) = ( g_1, \dots, g_m )$, then the restricted quadratic form space isometry problem for every input instance outputs no, and thus the algorithm outputs no. 

If the input of the Quadratic Form Tuple Isometry Algorithm is a yes instance, then there exists an invertible matrix $U$ such that $g_\ell =  f_\ell^U$ for any $1 \leq \ell \leq m$, and consequently $\mathfrak{T}_{\mathrm{Trans}(I, U)} = \mathfrak{T'}$. 
By Theorem~\ref{thm:tensor_iso_to_restricted_tensor_iso}, there exist $A_i, B_i, A_i', B_i', \alpha_{i, 0}, \alpha_{i, 1}, \beta_{i, 0}, \beta_{i, 1}$ as inputs for the restricted tensor isomorphism problem such that there exist $N$ and $M$ satisfying Equation~(\ref{equ:restricted_tensor_iso}) such that $U = B_i N B_i'^{-1}$ and $I = A_i M A_i'^{-1}$. Hence, the restricted quadratic form space isometry problem instance for $A_i, B_i, A_i', B_i', \alpha_{i, 0}, \alpha_{i, 1}, \beta_{i, 0}, \beta_{i, 1}$ outputs yes, and thus the Quadratic Form Tuple Isometry Algorithm outputs yes.

    Then the correctness of the theorem is obtained by  Theorem~\ref{thm:restricted_space_iso_algo} and Theorem~\ref{thm:matrix_tuple_constructed_main}. 
    The running time is obtained by Lemma~\ref{thm:tensor_iso_to_restricted_tensor_iso}, Theorem~\ref{thm:restricted_space_iso_algo}, and Theorem~\ref{thm:matrix_tuple_constructed_main}. 
\end{proof}



\section{Congruent block diagonalization of alternating matrix tuple}

In this section, we give a new block diagonalization for alternating matrix tuples over $\F_2$ under the congruence action.

For an alternating matrix tuple $\mathbf{A} \in \M(n, \F_2)^m$ and a matrix $X \in \M(n \times n', \F_2)$ for some $0 < n' < n$, we say $X^T \mathbf{A} X$ is a congruent block of $\mathbf{A}$ if there exist another matrix $X' \in \M((n - n')\times n, \F_2 )$ such that the invertible matrix
\[
\begin{pmatrix}
    X & X'
\end{pmatrix} \in \GL(n, \F_2)
\] satisfies
\[
\begin{pmatrix}
    X & X'
\end{pmatrix}^T \mathbf{A} \begin{pmatrix}
    X & X'
\end{pmatrix}
=\begin{pmatrix}
    X^T \mathbf{A} X & \\ & X'^T \mathbf{A} X'
\end{pmatrix}. 
\]

We show that congruent blocks can be obtained from left-right blocks. 

\begin{theorem}\label{thm:congruence_block_diagonalization}
Let $X^T \mathbf{A}Y$ be a minimum left-right block for an alternating matrix tuple $\mathbf{A}$ over $\F_2$ such that the number of columns of $X$ is less than or equal to the number of columns of $Y$. Then one of the following two conditions holds:
\begin{enumerate}
\item $X^T \mathbf{A}X$ is a congruent block of $\mathbf{A}$, and $X$ has the same number of columns as $Y$. 
\item $\begin{pmatrix}X & Y \end{pmatrix}^T  \mathbf{A}\begin{pmatrix}X & Y \end{pmatrix}$ is a congruent block of $\mathbf{A}$. 
\end{enumerate}
\end{theorem}

We first prove the following lemma. 

\begin{lemma}\label{lem:block_diagonal_basic}
Suppose $X^T \mathbf{A} Y$ is a minimum left-right block of an alternating matrix tuple $\mathbf{A} \in \M(n, \F_2)^m$, and the number of columns of $X$ is less than or equal to that of $Y$.
If there exists an invertible matrix Z such that $X^T \mathbf{A} X Z = X^T \mathbf{A} Y$, then $X^T \mathbf{A} X$ is both a left-right block and a congruent block of $\mathbf{A}$. 
\end{lemma}
\begin{proof}
It is impossible for $X$ to have fewer columns than $Y$ because $X^T \mathbf{A} X Z = X^T \mathbf{A} Y$ for some invertible $Z$. Hence, $X$ has the same number of columns as $Y$. 

Let $U = \langle u \in \F_2^n :X^T \mathbf{A} u = 0 \rangle$. $U$ is of dimension $n - d$, where $d$ is the number of columns of $X$. Take an arbitrary matrix $W \in \M(n \times (n - d), \F_2)$ such that the columns of $W$ form a linear basis of $U$. Then the matrix \[\begin{pmatrix}X &  W\end{pmatrix}\] is invertible and 
\[\begin{pmatrix}X  &  W\end{pmatrix}^T \mathbf{A} \begin{pmatrix}X  & W\end{pmatrix}  = \diag(X^T \mathbf{A} X, W^T \mathbf{A} W).\] Therefore, $X^T \mathbf{A} X$ is a block of $\mathbf{A}$ in both left-right and congruence actions. 
 \end{proof}
 
 \begin{proof}[Proof of Theorem~\ref{thm:congruence_block_diagonalization}]
 We suppose that the first case does not hold, and show that the second case holds. 
 Let $d_x$ denote the number of columns of $X$ and $d_y$ denote the number of columns of $Y$. 
 
 Since $X^T \mathbf{A} Y$ is a minimum left-right block of $\mathbf{A}$, there are invertible matrices \[U = \begin{pmatrix}X &  U' \end{pmatrix} \text{ and } W = \begin{pmatrix}Y & W' \end{pmatrix} \] such that \[U^T \mathbf{A} W = \diag(X^T \mathbf{A} Y, U'^T \mathbf{A} W').\] 
 Let $X_0$ and $X_1$ be the projection of $X$ on $\mathrm{columnspan}(Y)$ and $\mathrm{columnspan}(W')$ with respect to $\mathrm{columnspan}(Y) \oplus \mathrm{columnspan}(W')$. Then we have $X = X_0 + X_1$. 
 By Theorem~\ref{thm:left_right_decomposition}, both $X_0$ and $X_1$ are in $V_{X^T \mathbf{A} Y}$.
 If $X_0$ is in $V_{X^T \mathbf{A} Y}\setminus K_{X^T \mathbf{A} Y}$, then $d_x = d_y$ because otherwise $X^T \mathbf{A} Y$ is not a minimum left-right block by Theorem~\ref{thm:left_right_decomposition}.
 By Lemma~\ref{lem:block_diagonal_basic}, $X^T \mathbf{A} X$ is a block, contradicting the assumption that the first case does not hold.
 Hence, $X_0$ is in $K_{X^T \mathbf{A} Y}$. 
By Theorem~\ref{thm:left_right_decomposition}, $X_1$ must be an element in $V_{X^T \mathbf{A} Y} \setminus K_{X^T \mathbf{A} Y}$. Hence, the rank of $X_1$ is the same as the rank of $X$.

Without loss of generality, we can assume 
\[U' = \begin{pmatrix}Y_1 & U''\end{pmatrix} \text{ and } W' = \begin{pmatrix}X_1 & W''\end{pmatrix}\] so that $X_1^T \mathbf{A} Y_1$ is a minimum left-right block of $\mathbf{A}$ equal to $X^T \mathbf{A} Y$ and \[U'^T \mathbf{A} W' = \diag(Y_1^T \mathbf{A} X_1, \mathbf{C}) = \diag(Y_1^T \mathbf{A}^T X_1, U''^T \mathbf{A} W'').\]
Hence, 
\[U^T \mathbf{A} W = \diag(X^T \mathbf{A} Y, Y_1^T \mathbf{A} X_1, U''^T \mathbf{A} W'').\]
There is an invertible matrix 
\[S = \begin{pmatrix}I_{d_y} & S' & \\ & I_{d_x} & \\ & & I_{n - x - y}\end{pmatrix}\] such that 
\[W S= \begin{pmatrix} Y &  X & W''\end{pmatrix}.\]
Consequently, we have 
\[U^T \mathbf{A} W S = \begin{pmatrix} X^T A Y & \mathbf{C}  & \\ & Y_1^T A X_1 & \\ & & U''^T \mathbf{A} W''\end{pmatrix}\]
for some $d_y \times d_x$ matrix tuple $\mathbf{C}$. 

Since $X^T\mathbf{A}Y$ is a minimum left-right block, for every column vector $v$, $X^T \mathbf{A} v$ is a linear combination of the columns of $X^T\mathbf{A}Y$. 
On the other hand, since \[Y_1^T \mathbf{A} X_1 = Y_1^T \mathbf{A}^T X_1 = (X_1^T \mathbf{A} Y_1)^T = (X^T \mathbf{A} Y)^T 
= Y^T \mathbf{A} X,\] every row of $\mathbf{C}$ is a linear combination of the rows of $Y_1^T \mathbf{A} X_1$. Consequently, $X^T \mathbf{A} Y_1$ is also a minimum left-right block of $\mathbf{A}$. 

Hence, the projection of any non-zero linear combination of the columns of $Y$ on $\mathrm{columnspan}(Y_1)$ with respect to $\mathrm{columnspan}(X)\oplus \mathrm{columnspan}(Y_1)\oplus \mathrm{columnspan}(U'')$ must be non-zero. 
Thus, the matrix \[\begin{pmatrix}X & Y & U''\end{pmatrix}\] is an invertible matrix such that 
\[ \begin{pmatrix}X & Y & U''\end{pmatrix}^T \mathbf{A} \begin{pmatrix}X & Y & U''\end{pmatrix} = \begin{pmatrix}X^T \mathbf{A} X & X^T \mathbf{A} Y & \\
Y^T \mathbf{A} X & Y^T \mathbf{A} Y & \\
& & U''^T \mathbf{A} U'' \end{pmatrix}.\]
So the second case of the theorem holds. 
\end{proof}

 Theorem~\ref{thm:congruence_block_diagonalization} implies the following corollary. 
 \begin{corollary}\label{cor:standard_block_diagonalization}
Any alternating matrix tuple $\mathbf{A}$ has a block diagonalization under the congruence action such that every block satisfies one of the following two conditions:
\begin{enumerate}
\item The block is $X^T \mathbf{A}X$ for some $X$ such that there exists another matrix $Y$ with a number of columns greater than or equal to $X$ and 
$X^T \mathbf{A} Y$ a minimum left-right block of $\mathbf{A}$.  
\item The block is $\begin{pmatrix} X &  Y\end{pmatrix}^T \mathbf{A} \begin{pmatrix} X & Y\end{pmatrix} $, where $X ^T \mathbf{A}Y$ is a minimum left-right block of $\mathbf{A}$. 
\end{enumerate}
Furthermore, one can construct a block diagonalization of $\mathbf{A}$ under the congruence action as follows: 
Arbitrarily select a minimum left-right block  $X^T \mathbf{A} Y$ of $\mathbf{A}$ such that the number of columns of $X$ is less than or equal to the number of columns of $Y$,
if $X$ satisfies the first condition of Theorem~\ref{thm:congruence_block_diagonalization}, then let $U_0 = X$; otherwise, let $U_0 = \begin{pmatrix}X & Y\end{pmatrix}$. 
Let $U = \begin{pmatrix} U_0 & U_1\end{pmatrix}$ be an arbitrary invertible matrix such that $U^T \mathbf{A} U = \diag(U_0^T \mathbf{A} U_0, U_1^T \mathbf{A} U_1)$.
Let $U_0^T \mathbf{A} U_0$ be a congruent block and recurse on $U_1^T \mathbf{A} U_1$.
\end{corollary}

\begin{lemma}\label{lem:two_blocks_right_multi}
Let $\mathbf{A}$ be an alternating matrix tuple and $X^T \mathbf{A} Y$ be a minimum left-right block such that $\begin{pmatrix}X & Y \end{pmatrix}^T \mathbf{A} \begin{pmatrix}X & Y \end{pmatrix}$ is a congruence block of $\mathbf{A}$. Then there is an invertible matrix $U$ such that 
\[\begin{pmatrix}X & Y \end{pmatrix}^T \mathbf{A} \begin{pmatrix}X & Y \end{pmatrix} U = \begin{pmatrix} X^T \mathbf{A} Y & \\ & Y^T \mathbf{A} X\end{pmatrix}.\]
\end{lemma}
\begin{proof}

Since $X^T\mathbf{A} Y$ is a minimum left-right block of $\mathbf{A}$, there exists an invertible matrix $U_0$ such that the following two conditions hold:
\begin{enumerate}
\item $ \begin{pmatrix}X & Y \end{pmatrix} U_0 = \begin{pmatrix} Y & U_0'\end{pmatrix}$ for some $U_0'$. 
\item \[\begin{pmatrix}X & Y \end{pmatrix}^T \mathbf{A} \begin{pmatrix}X & Y \end{pmatrix} U_0 = \begin{pmatrix} X^T \mathbf{A} Y & \\ \mathbf{B} & \mathbf{C}\end{pmatrix}\]
for some matrix tuples $\mathbf{B}$ and $\mathbf{C}$. 
\end{enumerate}


On the other hand, 
since $\mathbf{A}^T = \mathbf{A}$, $Y^T \mathbf{A} X$ is also a minimum block of $\mathbf{A}$.
Therefore, there must be an invertible matrix $S$ such that $\mathbf{C} S = Y^T \mathbf{A}X$. Because otherwise there is a column vector $v$ as a non-zero linear combination of the columns of $U_0'$ such that $\mathbf{A} v$ is a zero column tuple, contradicting the assumption that $\mathbf{A}$ is non-degenerate.
Consequently, every column of $\mathbf{B}$ is a linear combination of columns of $\mathbf{C}$. 
\end{proof}

We define the standard congruent block diagonalization of an alternating matrix tuple as follows. 

\begin{definition}\label{def:standard_congruent_block}
A congruent block $U^T\mathbf{A}U$ of an alternating matrix tuple $\mathbf{A}$ for some matrix $U$ is a \emph{standard congruent block} if one of the following two conditions holds:
\begin{enumerate}
    \item $U^T \mathbf{A} U$ is a minimum left-right block of $\mathbf{A}$.
    \item $U = \begin{pmatrix}X & Y \end{pmatrix}$ such that $X^T \mathbf{A} X$ is not a minimum left-right block of $\mathbf{A}$, and $X^T \mathbf{A} Y$ is a minimum left-right block of $\mathbf{A}$.
\end{enumerate}

We say that a standard congruent block $U^T\mathbf{A}U$ is obtained from a minimum left-right block $X^T \mathbf{A} Y$ if one of the following two cases is satisfied:
\begin{enumerate}
    \item $X = U$ and $Y$ is a matrix such that $X^T \mathbf{A} Y$ is a left-from minimum block of $\mathbf{A}$. 
    \item $U = \begin{pmatrix}X & Y \end{pmatrix}$. 
\end{enumerate}
\end{definition}

\begin{definition}\label{def:standard_diagonalization}
Matrix tuples $\mathbf{B}_1, \dots, \mathbf{B}_\ell$ are \emph{representative blocks} of an alternating matrix tuple $\mathbf{A}$ if the following conditions hold:
\begin{enumerate}
    \item Each of $\mathbf{B}_1, \dots, \mathbf{B}_\ell$ is a minimum left-right block of $\mathbf{A}$.
    \item For every $1 \leq i < j \leq \ell$, $\mathbf{B}_i$ is not left-right equivalent to either $\mathbf{B}_j$ or $\mathbf{B}_j^T$.
    \item If $\mathbf{B}_i^T$ is left-right equivalent to $\mathbf{B}_i$ for some $1 \leq i \leq \ell$, then there is an invertible matrix $W$ such that $\mathbf{B}_i^T W = \mathbf{B}_i$. 
    \item Let $n_i$ be the number of blocks in a minimum left-right block diagonalization of $\mathbf{A}$ that are left-right equivalent to $\mathbf{B}_i$. Then $n_1 \leq n_2 \leq \dots \leq n_\ell$. 
\end{enumerate}

An alternating matrix tuple $\mathbf{C} = \diag(\mathbf{C}_1, \dots, \mathbf{C}_k)$ is a \emph{standard congruent block diagonalization} of an alternating matrix tuple $\mathbf{A}$ with respect to representative blocks $\mathbf{B}_1, \dots, \mathbf{B}_\ell$ if the following conditions hold:
\begin{enumerate}
    \item Every $\mathbf{C}_i$ is a standard congruent block of $\mathbf{A}$. 
    \item There are $0 = t_0 < t_1 < \dots < t_\ell = k$ such that
    for each $1 \leq i \leq \ell$ and $t_{i-1} + 1 \leq j \leq t_i$, $\mathbf{C}_j$ is obtained from $X^T \mathbf{A} Y$ for some $X$ and $Y$ such that $X^T \mathbf{A} Y = \mathbf{B}_i$.
\end{enumerate}
\end{definition}

The following lemma suggests that for any two different standard congruent block diagonalizaitons $\mathbf{A}'$ and $\mathbf{A}''$ of an alternating matrix tuple $\mathbf{A}$ with respect to the same representative blocks, there exists an invertible matrix $W$ such that $\mathbf{A}' = \mathbf{A}'' W$. 

\begin{lemma}\label{lem:standard_diagonalization_1}
    Let $\mathbf{C} = \diag(\mathbf{C}_1, \dots, \mathbf{C}_k)$ be a standard congruent block diagonalization of an alternating matrix tuple $\mathbf{A}$ with respect to representative blocks $\mathbf{B}_1, \dots, \mathbf{B}_\ell$. 
    Let $n_i$ be the number of blocks in a minimum left-right block diagonalization of $\mathbf{A}$ left-right equivalent to $\mathbf{B}_i$ for any $1 \leq i \leq \ell$.
    There exist invertible matrices $W_1, \dots, W_k$ such that $\diag(\mathbf{C}_1 W_1, \dots, \mathbf{C}_k W_k) = \diag(\mathbf{F}_1, \dots, \mathbf{F}_\ell)$, where for any $1 \leq i \leq \ell$, $\mathbf{F}_i$ is the diagonal matrix tuple of $n_i$ copies of $\mathbf{B}_i$ if $\mathbf{B}_i^T$ is left-right equivalent to $\mathbf{B}_i$, or $n_i$ copies of $\diag(\mathbf{B}_i, \mathbf{B}_i^T)$ if $\mathbf{B}_i^T$ is not left-right equivalent to $\mathbf{B}_i$.
\end{lemma}
\begin{proof}

    If $\mathbf{C}_i$ itself is a minimum left-right block of $\mathbf{A}$, then by Definition~\ref{def:standard_congruent_block} and Definition~\ref{def:standard_diagonalization}, there is an invertible matrix $W_i$ such that $\mathbf{C}_i W_i = \mathbf{B}_j$ for some $1 \leq j \leq \ell$. 

    Otherwise, by Definition~\ref{def:standard_congruent_block}, there exist $X$ and $Y$ such that $X^T \mathbf{A} Y$ is a minimum left-right block of $\mathbf{A}$ and $\mathbf{C} = \begin{pmatrix}X & Y \end{pmatrix}^T \mathbf{A} \begin{pmatrix}X & Y \end{pmatrix}$. 
     By Lemma~\ref{lem:two_blocks_right_multi}, there is an invertible matrix $W_i'$ such that $\mathbf{C}_i W_i' = \diag(\mathbf{B}_j, \mathbf{B}_j^T)$ for some $1 \leq j \leq \ell$. 
     Furthermore, if $\mathbf{B}_j^T$ is left-right equivalent to $\mathbf{B}_j$, then there exists another invertible matrix $W_i''$ such that $\mathbf{C}_i W_i' W_i'' = \diag(\mathbf{B}_j, \mathbf{B}_j)$. 
     Then the lemma follows by Definition~\ref{def:standard_diagonalization}. 
\end{proof}

Now we describe an algorithm to construct a standard congruent block diagonalization of an alternating matrix tuple, given representative blocks.

\begin{framed}\begin{flushleft}
\vspace{-.3cm}
\textbf{Standard Congruent Block Diagonalization Algorithm}

\textbf{Input:} An alternating matrix tuple $\mathbf{A}$, and representative blocks $\mathbf{B}_1, \dots, \mathbf{B}_\ell$ for $\mathbf{A}$.

\textbf{Output:}
An invertible matrix $S$ such that $S^T \mathbf{A}S$ is a standard congruent block diagonalization of $\mathbf{A}$ with respect to $\mathbf{B}_1, \dots, \mathbf{B}_\ell$. 
\begin{enumerate}

\item Let $\mathbf{D} = \mathbf{A}$, and $U$ be the identity matrix initially. 
\item Repeat the following process until $\mathbf{D}$has no rows or columns remaining: 
\begin{enumerate}
    \item Let $i$ be the smallest integer such that there is a minimum left-right block $X^T \mathbf{D} Y$ of $\mathbf{D}$ equal to $\mathbf{B}_i$. 
    \item Let $Z_0 = X$ if $X^T \mathbf{D} X$ is a  minimum left-right block of $\mathbf{D}$, otherwise let $Z_0 = \begin{pmatrix}X & Y \end{pmatrix}$. 
    \item Compute an arbitrary invertible matrix $Z = \begin{pmatrix}Z_0 & Z_1\end{pmatrix}$ such that $Z^T \mathbf{D} Z = \diag(Z_0^T \mathbf{D} Z_0, Z_1^T \mathbf{D} Z_1)$. 
    \item Let $S=UZ_0$ if $S$ is not yet defined; otherwise let $S = \begin{pmatrix} S & UZ_0 \end{pmatrix}$. 
    Let $\mathbf{D} = Z_1^T \mathbf{D}Z_1$, and $U = U Z_1$. 
\end{enumerate}
\item Return $S$.
\end{enumerate}
\end{flushleft}
\vspace{-.3cm}\end{framed}

By Theorem~\ref{thm:left_right_decomposition_algorithm}, Corollary~\ref{cor:standard_block_diagonalization}, and Definition~\ref{def:standard_diagonalization}, we have the following corollary. 

\begin{corollary}\label{cor:standard_cong_block_diag_algo}
    Given an alternating matrix tuple $\mathbf{A} \in \M(n \times n, \F_2)^m$ and representative blocks $\mathbf{B}_1, \dots, \mathbf{B}_\ell$, the Standard Congruent Block Diagonalization Algorithm outputs a matrix $S$ such that $S^T \mathbf{A} S$ is a standard congruent block diagonalization of $\mathbf{A}$, with running time polynomial of $n$ and $m$. 
\end{corollary}

\section{Frequent block quadratic form tuple isometry}

In this section, we consider the following computational problem. 

\begin{problem}\label{problem:small_block}
Given two quadratic form tuples $\mathbf{f}_1$ and $\mathbf{f}_2$ over $n$ variables of length $m$, the representative blocks $\mathbf{B}_1, \dots, \mathbf{B}_\ell$ for $\mathbf{D}_{\mathbf{f}_1}$, and two parameters $1 \leq r \leq n$, $0 \leq \rho \leq n$, satisfying the following conditions:
\begin{enumerate}
\item For any $i \in \{1, 2\}$ and any minimum left-right block $\mathbf{B}$ of the partial derivative matrix tuple $\mathbf{D}_{\mathbf{f}_i}$  of $\mathbf{f}_i$, $\mathbf{B}$ has at most $r$ rows and at most $r$ columns.
\item Both $\mathbf{f_1}$ and $\mathbf{f}_2$ satisfy the $\rho$-non-zero-square-term condition with respect to any standard congruent block diagonalization of $\mathbf{D}_{\mathbf{f}_1}$ with respect to representative blocks $\mathbf{B}_1, \dots, \mathbf{B}_\ell$.
\end{enumerate}
decide whether there exists an invertible matrix $S$ such that $\mathbf{f}_1^S = \mathbf{f}_2$. 
\end{problem}

We remark that the parameter $r$ in this section will correspond to $n^\Delta$ as discussed in Section~\ref{sec:intro:techniques}. 
In this section, we prove the following theorem. 

\begin{theorem}\label{thm:small_block}
There is a $2^{O(n\cdot (\rho + r))}\cdot \mathrm{poly}(m)$ time algorithm to solve Problem~\ref{problem:small_block}. 
\end{theorem}

We say that a quadratic form tuple is \emph{standard congruent block diagonalized} with respect to representative blocks $\mathbf{B}_1, \dots, \mathbf{B}_\ell$ if its partial derivative matrix is a standard block diagonalization with respect to $\mathbf{B}_1, \dots, \mathbf{B}_\ell$. 

We first bound the number of different standard congruent block diagonalized quadratic form tuples congruent to a given quadratic form tuple. 


\begin{lemma}\label{lem:small_block_diag_number}
Let $\mathbf{f}$ be a matrix tuple over $n$ variables such that every minimum left-right block has at most $r$ rows and $r$ columns. 
Let $\mathbf{B}_1, \dots, \mathbf{B}_\ell$ be representative blocks of $\mathbf{D}_\mathbf{f}$ such that $\mathbf{f}$ satisfies the $\rho$-non-zero-square-term condition with respect to any standard congruent block diagonalization of $\mathbf{D}_\mathbf{f}$ with respect to $\mathbf{B}_1, \dots, \mathbf{B}_\ell$.

There are at most $2^{O(n\cdot (\rho + r))}$ different possible standard congruent block diagonalized quadratic form tuples congruent to $\mathbf{f}$.

\end{lemma}

\begin{proof}
For a standard congruent block diagonalization of $\mathbf{D}_\mathbf{f}$, let $k$ denote the number of standard congruent blocks, and $k$ numbers $\alpha_1, \dots, \alpha_k$ being either $1$ or $2$ indicating the number of minimum left-right blocks in the $i$-th standard congruent block of the standard congruent block diagonalization.

We bound the number of possible different standard congruent block diagonalized quadratic form tuples for fixed $k$ and $\alpha_1, \dots, \alpha_k$. 
Let $\mathbf{f}'$ be an arbitrary such quadratic form tuple. 
For any $1 \leq i \leq k$, let $m_i$ denote the number of rows and columns of the $i$-th standard congruent block and $g_i$ be the integer such that the $i$-th standard congruent block is obtained from a minimum left-right block equal to $\mathbf{B}_{g_i}$. 
By Lemma~\ref{lem:standard_diagonalization_1}, $m_1, \dots, m_k$ and $g_1, \dots, g_k$ are determined once $k$ and $\alpha_1, \dots, \alpha_k$ are fixed.

Furthermore, by Lemma~\ref{lem:standard_diagonalization_1}, there is an invertible matrix $W = \diag(W_1, \dots, W_k)$ with $W_i \in \GL(m_i, \F_2)$ for any $1 \leq i \leq k$ such that $\mathbf{D}_{\mathbf{f}'} W = \diag(\mathbf{F}_1,\dots, \mathbf{F}_\ell)$ with $\mathbf{F}_i$ defined as the diagonalization of $n_i$ copies of $\mathbf{B}_i$ if $\mathbf{B}_i^T$ is left-right equivalent to $\mathbf{B}_i$, or $n_i$ copies of $\diag(\mathbf{B}_i, \mathbf{B}_i^T)$, where $n_i$ is the number of blocks in a minimum left-right block diagonalization of $\mathbf{D}_{\mathbf{f}_1}$ that is left-right equivalent to $\mathbf{B}_i$ for any $1 \leq i \leq \ell$.
Since there are at most $2^{\sum_{i = 1}^k m_i^2}$ different possible $W$ once $m_1, \dots, m_k$ are fixed, 
there are at most $2^{\sum_i m_i^2}$ 
different partial derivative matrix tuples for $\mathbf{f}'$ once $k$ and $\alpha_1, \dots, \alpha_k$ are fixed. 
Since $\mathbf{f}$ satisfies the $\rho$-non-zero-square-term condition with respect to any standard congruent block diagonalization of $\mathbf{D}_\mathbf{f}$ with respect to $\mathbf{B}_1, \dots, \mathbf{B}_\ell$, 
there are at most
\[2^{\sum_{i = 1}^k m_i^2}\cdot 2^{n \rho} \leq 2^{n \cdot 2r} \cdot 2^{n \rho} = 2^{n \cdot (2r + \rho)}\]
different quadratic form tuples satisfying the two conditions of the lemma once $k$ and $\alpha_1, \dots, \alpha_k$ are fixed.

On the other hand, there are at most $n \cdot 2^n$ different possible $k$ and $(\alpha_1, \dots, \alpha_k)$. The total number of different possible $\mathbf{f}'$ is upper bounded by 
\[ n \cdot 2^n \cdot 2^{n \cdot (2r + \rho)} = 2^{O(n \cdot (r + \rho))}.\]
\end{proof}

Our algorithm enumerates all possible quadratic form tuples congruent to $\mathbf{f}_1$ having their partial derivative matrix tuples being standard congruent block diagonalizations of $\mathbf{D}_{\mathbf{f}_1}$ with respect to $\mathbf{B}_1, \dots, \mathbf{B}_\ell$. 
And check if a quadratic form congruent to $\mathbf{f}_2$ with a partial derivative matrix tuple satisfying the same condition equals any of the quadratic form tuples for $\mathbf{f}_1$ enumerated.

In our algorithm, we obtain all the standard congruent block diagonalized quadratic form tuples from a given one by performing some local operations iteratively. 
For a standard congruent block diagonalized quadratic form tuple $\mathbf{f}$ over $n$ variables (with respect to $\mathbf{B}_1, \dots, \mathbf{B}_\ell$). Denote $\mathbf{D}_\mathbf{f}$ as $\diag(\mathbf{D}_1, \dots, \mathbf{D}_d)$ with standard congruent blocks $\mathbf{D}_1, \dots, \mathbf{D}_d$. Let $m_i$ denote the number of rows and columns of $\mathbf{D}_i$, and $r_i = \sum_{j = 1}^{i} m_i$ for any $1 \leq i \leq d$.
A \emph{local operation} for $\mathbf{f}$ is an invertible matrix $S$ satisfying one of the following conditions:

\begin{enumerate}
    \item For any $1 \leq i \leq d$, let 
    \[S = \begin{pmatrix}I_{r_{i - 1}} & &  \\ & U & \\ & & I_{n - r_i} \end{pmatrix}\]
    be an invertible matrix such that $\mathbf{f}^S$ is a standard congruent block diagonalized quadratic form tuple with respect to $\mathbf{B}_1, \dots, \mathbf{B}_\ell$. 
    \item For any $1 \leq i < j \leq d$, let 
    \[S = \begin{pmatrix}I_{r_{i - 1}} & &  \\ & U_1 & & U_2\\ 
    & & I_{r_{j - 1} - r_i} \\
    & U_3 & & U_4 \\
    & & & & I_{n - r_j} \end{pmatrix}\]
    with $U_1 \in \M(m_i \times t, \F_2), U_2 \in \M(m_i \times (m_i + m_j - t), \F_2)$, $U_3 \in \M(m_j \times t, \F_2)$,  and $U_4 \in \M(m_j \times (m_i + m_j - t), \F_2)$ for some $0 \leq t \leq m_i + m_j$ such that $S$ is an invertible matrix and $\mathbf{f}^S$ is a standard congruent block diagonalized quadratic form tuple with respect to $\mathbf{B}_1, \dots, \mathbf{B}_\ell$. 

\end{enumerate}

We have the following algorithm for the Problem~\ref{problem:small_block}:

\begin{framed}\vspace{-.3cm}\begin{flushleft}
\textbf{Frequent Block Quadratic Form Tuple Isometry Algorithm}

\textbf{Input:} Two quadratic form tuples $\mathbf{f}_1$, $\mathbf{f}_2$, representative blocks $\mathbf{B}_1, \dots, \mathbf{B}_\ell$ for $\mathbf{D}_{\mathbf{f}_1}$, and a parameter $0 \leq \rho \leq n$.

\textbf{Output:}
Yes or no. 

\begin{enumerate}
    \item Output no if $\mathbf{D}_{\mathbf{f}_1}$ and $\mathbf{D}_{\mathbf{f}_2}$ are not left-right equivalent. 
    \item Compute $\mathbf{f}_1'$ and $\mathbf{f}_2'$, which are arbitrary standard congruent block diagonalized quadratic form tuples for $\mathbf{f}_1$ and $\mathbf{f}_2$, respectively, both with respect to $\mathbf{B}_1, \dots, \mathbf{B}_\ell$. 
    Let $\Gamma = \{\mathbf{f}_1'\}$. 
    \item Iteratively perform the following process until $\Gamma$ does not change: For any $\mathbf{f} \in \Gamma$ and any local operation $S$ for $\mathbf{f}$, add $\mathbf{f}^S$ to $\Gamma$. 
    \item Output yes if $\mathbf{f}_2'$ is in $\Gamma$. Otherwise output no. 
\end{enumerate}
\end{flushleft}
\vspace{-.3cm}
\end{framed}

\begin{proof}[Proof of Theorem~\ref{thm:small_block}]
By Definition~\ref{def:standard_diagonalization}, if $\mathbf{f}_1$ and $\mathbf{f}_2$ are congruent, then any standard congruent block diagonalized quadratic form tuple of $\mathbf{f}_1$ with respect to $\mathbf{B}_1, \dots, \mathbf{B}_\ell$ is also a standard congruent block diagonalized quadratic form tuple of $\mathbf{f}_2$ with respect to $\mathbf{B}_1, \dots, \mathbf{B}_\ell$. 
Then the correctness of the algorithm is obtained if $\Gamma$ contains all the standard congruent block diagonalized quadratic form tuples of $\mathbf{f}_1$ with respect to $\mathbf{B}_1, \dots, \mathbf{B}_\ell$. 
By the algorithm, it is sufficient to show that for two arbitrary standard congruent block diagonalized quadratic form tuples $\mathbf{f}$ and $\mathbf{g}$ for $\mathbf{f}_1$ with respect to $\mathbf{B}_1, \dots, \mathbf{B}_\ell$, there exists a sequence of local operations to turn $\mathbf{f}$ into $\mathbf{g}$.

Let $U_\mathbf{f}$ be an arbitrary invertible matrix such that $\mathbf{g} = \mathbf{f}^{U_\mathbf{f}}$. 
We use $k_\mathbf{f}$ (or $k_\mathbf{g}$) to denote the number of standard congruent blocks in $\mathbf{f}$ (or $\mathbf{g}$), and $m_{\mathbf{f}, i}$ (or $m_{\mathbf{g}, i}$) to denote the number of columns/rows of $i$-th standard congruent block of $\mathbf{D}_\mathbf{f}$ (or $\mathbf{D}_\mathbf{g}$).
Let $\alpha_{U_\mathbf{f}}$ denote the first integer such that the $\alpha_{U_\mathbf{f}}$-th column of $U_\mathbf{f}$ is different from $e_{\alpha_{U_\mathbf{f}}}$, where $e_{\alpha_{U_\mathbf{f}}}$ denotes the column vector in which only the $\alpha_{U_\mathbf{f}}$-th coordinate is one, and all the other coordinates are zero. 
Suppose the $\alpha_{U_\mathbf{f}}$-th column of $U_\mathbf{f}$ is used to construct the $\beta_{U_\mathbf{f}}$-th standard congruent block of $\mathbf{D}_\mathbf{g}$. 

We show that there exists a sequence of local operations to turn $\mathbf{f}$ into $\mathbf{f}'$ such that $\beta_{U_{\mathbf{f}'}} > \beta_{U_{\mathbf{f}}}$. Then the lemma is obtained by iteratively applying the local operations until the quadratic form tuple is identical to $\mathbf{g}$.

To prove this argument, we only need to show that there exists a sequence of local operations to turn $\mathbf{f}$ to $\mathbf{h}$, another standard congruent block diagonalized quadratic form tuple of $\mathbf{f}_1$, such that there exists an invertible matrix $U_\mathbf{h}$ satisfying the following conditions:
\begin{enumerate}
    \item $\mathbf{h}^{U_\mathbf{h}} =\mathbf{g}$.
    \item The $i$-th column of $U_\mathbf{h}$ is $e_i$ for any $1 \leq i \leq \sum_{j = 1}^{\beta_{U_\mathbf{f}}} m_{\mathbf{g}, j}$. 
    \item The first $\sum_{j = 1}^{\beta_{U_\mathbf{f}}}m_{\mathbf{g}, j}$ rows and columns of $\mathbf{h}$ correspond to the first $\beta_{U_\mathbf{f}}$ standard congruent blocks of $\mathbf{h}$.
\end{enumerate}

Let $X$ and $Y$ be arbitrary matrices satisfying the following conditions:
\begin{enumerate}
    \item The $\beta_{U_\mathbf{f}}$-th standard congruent block of $\mathbf{D}_\mathbf{g}$ is obtained from $X^T \mathbf{D}_\mathbf{g} Y = \mathbf{B}_t$ for some $1 \leq t \leq \ell$.
    \item $\langle \mathrm{columnspan}(X) \cup \mathrm{columnspan}(Y)\rangle$ is a subset of \[\left\langle e_{\left(\sum_{j = 1}^{\beta_{U_\mathbf{f}}- 1}m_{\mathbf{g}, j} \right) + 1}, \dots, e_{\sum_{j = 1}^{\beta_{U_\mathbf{f}}} m_{\mathbf{g}, j}}\right\rangle.\]
\end{enumerate}
By Definition~\ref{def:standard_congruent_block}, such $X$ and $Y$ exist. 
Therefore, we have 
$X^T \mathbf{D}_\mathbf{g} Y = X^T U_\mathbf{f}^T \mathbf{D}_\mathbf{f} U_\mathbf{f} Y$ of $\mathbf{D}_\mathbf{g}$. 
Our goal is to use $X^T U_\mathbf{f}^T \mathbf{D}_\mathbf{f} U_\mathbf{f}Y$ to construct $\beta_{U_\mathbf{f}}$-th standard congruent block of $\mathbf{h}$ for any $1 \leq i \leq k_\mathbf{f}$. 

Let 
\[E_{\mathbf{f}, i} = \left\langle e_{\left(\sum_{j = 1}^{i - 1}m_{\mathbf{f}, j} \right) + 1}, \dots, e_{\sum_{j = 1}^{i} m_{\mathbf{f}, j}}\right\rangle,\]
and $X_i$ be the projection of $U_{\mathbf{f}}X$ on $E_{\mathbf{f}, i}$ with respect to $E_{\mathbf{f}, 1} \oplus \dots \oplus E_{\mathbf{f}, k_\mathbf{f}}$. Since $\mathbf{f}^{U_\mathbf{f}} = \mathbf{g}$, 
$X_i$ is a zero matrix for any $i < \beta_{U_\mathbf{f}}$. 

If $X_\beta$ is not in $V_{\mathbf{B}_t} \setminus K_{\mathbf{B}_t}$ with the definitions of $V_{\mathbf{B}_t}$ and $K_{\mathbf{B}_t}$ following Theorem~\ref{thm:left_right_decomposition}, then we can use a local operation to make $X_\beta$ be in $V_{\mathbf{B}_t} \setminus K_{\mathbf{B}_t}$. By Theorem~\ref{thm:left_right_decomposition} and Theorem~\ref{thm:congruence_block_diagonalization}, such a local operation exists. 
So, in the rest of this proof, we suppose $X_\beta$ is in $V_{\mathbf{B}_t} \setminus K_{\mathbf{B}_t}$.  

Let $i > \beta_{U_\mathbf{f}}$ be the smallest integer such that $X_i$ is a non-zero matrix. 
By Theorem~\ref{thm:left_right_decomposition} and Theorem~\ref{thm:congruence_block_diagonalization}, there is a local operation such that after applying the operation, the quadratic form tuple satisfies that $X_j$ is a non-zero matrix for any $j \leq i$ not equal to $\beta$. 
Hence, after at most $n$ local operations, the quadratic form satisfies that $X_i$ is a zero matrix for all $i \neq \beta_{U_\mathbf{f}}$. 

Then one can obtain $\mathbf{h}$ by applying one local operation on the $\beta$-th standard congruent block by Theorem~\ref{thm:left_right_decomposition} and Theorem~\ref{thm:congruence_block_diagonalization}. 




Finally, we bound the running time. 
By Lemma~\ref{lem:small_block_diag_number}, the set $\Gamma$ is of size at most $2^{O(n\cdot (\rho + r))}$.
For each standard block diagonalization, there are at most $n^2 \cdot 2^{O(r^2)}$ possible local operations. Hence, the overall running time is $2^{O(n\cdot (\rho + r))}\cdot \mathrm{poly}(m)$.  
\end{proof}

\section{Infrequent block quadratic form tuple isometry}

In this section, we consider the following problem. 

\begin{problem}\label{problem:large_block}
Decide the equivalence of two quadratic form tuples $\mathbf{f}_1$ and 
$\mathbf{f}_2$ satisfying the following
conditions: 
\begin{enumerate}
    \item For any $i \in \{1, 2\}$ and any minimum block $\mathbf{B}$ of the 
partial derivative matrix tuple $\mathbf{D}_{\mathbf{f}_i}$ of 
$\mathbf{f}_i$, there are at most $D$ 
minimum left-right blocks in an arbitrary minimum left-right block 
diagonalization of $\mathbf{D}_{\mathbf{f}_i}$ that are left-right equivalent to $\mathbf{B}$ for a given parameter $D \leq n$. 
\item For any $i \in \{1, 2\}$, there is no nonzero vector $v$ such that
$\mathbf{D}_{\mathbf{f}_i}v=0$.
\end{enumerate}
\end{problem}

Let us remark that the multiplicity bound $D$ in this section will correspond to $n^{1-\Delta}$ as discussed in Section~\ref{sec:intro:techniques}. 
We show the following.

\begin{theorem}
\label{thm:large_block}
Problem \ref{problem:large_block} can be solved in time $2^{nD}\cdot 
\poly(m, n)$, 
where $m$ is the length of the tuples.     
\end{theorem}

We are actually looking for an invertible matrix $S$ such that 
$\mathbf{f}_2=\mathbf{f}_1^S$. We can limit the search space
based on the following observation.

\begin{lemma}
\label{lem:in-adjoint}
Let $S,S_0,S_0'\in \GL(n,\F_2)$ such that
$\mathbf{f}_2=\mathbf{f}_1^S$ and 
$\mathbf{D_{f_2}}={S_0}^T\mathbf{D_{f_1}}S_0'$.
Then the matrix $S_1:=S_0^{-1}S$ satisfies the
following properties:
\begin{itemize}
\item $\mathbf{f_2}=\mathbf{f_1}^{S_0S_1}$;
\item 
there exists a matrix $S_1'$ 
such that 
$S_1^T\mathbf{D_{f_2}}=\mathbf{D_{f_2}}S_1'$.
\end{itemize}
\end{lemma}

\begin{proof}
The first property of $S_1$ is obvious. To see the second, 
put $S_1'=S^{-1}S_0'$. By Lemma~\ref{lem:partial_derivative_matrix}, 
we have 
$S^T\mathbf{D_{f_1}}=\mathbf{D_{f_2}}S^{-1}$ which can be used to show
\[S_1^T\mathbf{D_{f_2}}=S^T{S_0^T}^{-1}\mathbf{D_{f_2}}=
S^T{S_0^T}^{-1}S_0^T\mathbf{D_{f_1}}S_0'=S^T\mathbf{D_{f_1}}S_0'=
\mathbf{D_{f_2}}S^{-1}S_0'=\mathbf{D_{f_2}}S_1'.\]
\end{proof}

We shall begin our algorithm with the relatively easy task of
finding $S_0$ and $S_0'$ as in the lemma. Then the equivalence of
${\mathbf f_1}$ and $\mathbf f_2$ is equivalent to the existence
of $S_1$ in the set 
$$G=\{S_1\in\GL(n, \F_2):\exists S_1'\in\GL(n,\F_2)\mbox{~such~that~} 
S_1\mathbf{D_{f_2}}=\mathbf{D_{f_2}}S_1'\}$$
such that $\mathbf{f_2}=\mathbf{f_1}^{S_0S_1}$.
It turns out that $G$ is a matrix group having an efficiently computable
 unipotent normal subgroup $R$ of index at most $2^{nD}$, 
and for
any invertible matrix $Y$, the problem of existence (and finding) $X\in R$
such that $\mathbf{f}_2=\mathbf{f}_1^{YX}$ can be solved in time
$\poly(mn)$. We take a system $\Gamma$ of coset representatives of $R$ in $G$
and for all $Z\in \Gamma$ attempt to solve the above search problem for
$Y=S_0Z$. 

Before describing the formal algorithm, we introduce some
notation and terminology. Also, we present some basic facts
and some tools that can be used in the algorithm.
We shall consider the set 
\begin{equation}\label{eq:calA}
{\cal A}=\{S_1\in\\M(n\times n,\F_2):\exists S_1'\mbox{~such~that~} 
S_1\mathbf{D_{f_2}}=\mathbf{D_{f_2}}S_1'\}
\end{equation}
closely related to $G$. 
Since the condition
is linear, ${\cal A}$ is a linear subspace of $\M(n\times n,\F_2)$.
Also, if $S_1,S_2\in{\cal A}$, that is,
$S_1\mathbf{D_{f_2}}=\mathbf{D_{f_2}}S_1'$ and $S_2\mathbf{D_{f_2}}=\mathbf{D_{f_2}}S_2'$ for
some matrices $S_1'$ and $S_2'$, then
\[(S_1S_2)^T\mathbf{D_{f_2}}=S_2^TS_1^T\mathbf{D_{f_2}}=S_2^T\mathbf{D_{f_2}}S_1'=\mathbf{D_{f_2}}S_2'S_1',\]
whence $S_1S_2\in {\cal A}$. 
Thus ${\cal A}$ is in fact a subalgebra of $\M(n\times n,\F_2)$ and
$G$ is the group of invertible elements of ${\cal A}$.
The Jacobson radical $\Rad({\cal A})$ is the largest nilpotent
ideal of ${\cal A}$.  We define $R$
as the set 
$$R=I+\Rad({\cal A})=\{I+S_1:S_1\in \Rad({\cal A})\}.$$ 
Using the fact
that $\Rad({\cal A})$ is an ideal of ${\cal A}$, it is straightforward
to check that $R$ is a normal subgroup of $G$. 

\begin{lemma}
\label{lem:generators-unip}
A system of generators for $R$ can be computed in polynomial time.
\end{lemma}

\begin{proof}
A basis for $\Rad({\cal A})$ can be computed
in polynomial time by a method of R\'onyai~\cite{Ronyai90}.
Using this basis, bases for $\Rad({\cal A})^2,\ldots,
\Rad({\cal A})^k$ can be efficiently computed where
$k$ is the smallest integer such that
$\Rad({\cal A})^{k+1}=0.$ Then take the basis of
$\Rad({\cal A})^{k}$, extend it to a basis of
$\Rad({\cal A})^{k-1}$, further extend this basis
to a basis of $\Rad({\cal A})^{k-2}$, and so on. 
This way we obtain a basis of
$\Rad({\cal A})$ that contains bases of 
$\Rad({\cal A})^i$ ($i=1,\ldots,k$). 
We define the subgroup $R_i=I+\Rad({\cal A})^i$
($i=1,\ldots,k$) $R_i$ is the kernel of the homomorphism
from $R$ obtained from the projection ${\cal A}/\Rad({\cal A})^i$.
We show by induction on $i$ that the matrices
$I+M$ where $M$ is taken from the basis of
$\Rad({\cal A})^{k-i}$ generate the subgroup
$R_{k-i}=I+\Rad({\cal A})^{k-i}$ ($i=0,\ldots,k-1$). 
The initial case $i=0$ as well as the induction step
is based on \[(I+M_1)(I+M_2)=I+M_1+M_2+M_1M_2\in I+M_1+M_2+
\Rad({\cal A})^{2(k-i)}.\]
\end{proof}

A square matrix $M$ is called 
unipotent if $M-I$ is nilpotent. A matrix group is unipotent if all its
elements are unipotent.  
Over $\F_2$ an
invertible matrix is unipotent if and only if its order is a power of $2$.
Thus, over $\F_2$ the matrix groups which are 2-groups and unipotent
matrix groups coincide.

Let $L$ be a $d$-dimensional matrix
group over a finite field given by generators.
The {\em vector transporter problem} for $L$ is the following. Given
$d$-dimensional vectors $v_1,v_2$ compute the set
$L_{v_1\rightarrow v_2}=\{M \in L:v_2=Mv_1\}$. It is either
the empty set or a right coset of the stabilizer $L_{v_2}=L_{v_2\rightarrow
v_2}$. By computing $L_{v_1\rightarrow v_2}$ (if it is non-empty) we mean
giving a list of generators for the stabilizer $L_{v_2}$ and a 
single representative of the coset of $L_{v_1\rightarrow v_2}$.
We shall refer to $v_1$ as the {\em initial vector} and $v_2$ as the {\em target vector}. In fact, looking for
$S$ such that $\mathbf{f}_2=\mathbf{f}_1^S$ can be cast as 
(a weaker version of) an instance of the vector transporter problem for
$\GL(n,\F_2)$. 
However, we shall consider instances of the vector transporter problem
for the unipotent group $R$. The following is implicit
in some works of Luks and Miyazaki \cite{LuksMiyazaki02,LuksMiyazaki11}.

\begin{fact}
\label{fact:unip-transp}
The vector transporter problem for a unipotent matrix group over
$\F_2$ can be solved in polynomial time.
\end{fact}

In fact, the result is stated for solvable groups with certain
parameters that can be potentially too large. However, the method
works in polynomial time for unipotent groups. Here we present
an algorithm based on the method of Luks and Miyazaki for
the special case of unipotent matrix groups over $\F_2$.

\begin{proof}
We consider an extension of the vector transporter problem, which is more 
suitable for iteration or recursion. Recall that an affine transformation 
of a vector  space $V$ is a linear transformation followed by a translation. 
A formal definition could be the following.
The $\Phi:V\rightarrow V$ is called an affine transformation if
the map $lin_\Phi$ defined as $lin_\Phi(v)=\Phi(v)-\Phi(0)$ is a
linear transformation. We shall refer to $lin_\Phi$ as the linear part
of $\Phi$. It can be seen that taking the linear part (i.e., 
the map $\Phi\mapsto lin_\Phi$) is a homomorphism from the group of 
affine transformations to the group of linear transformations. We
call a group of affine transformations unipotent if the linear
parts are a unipotent linear group. We show that the vector 
transporter problem for a unipotent group of affine transformations
is solvable in polynomial time.

We begin by showing that the problem is reducible to the case
when the target vector $v_2$ is the zero vector. Let $G$ be a unipotent 
group of affine transformations. Given $v_1,v_2\in V$, we want to compute
$G_{v_1\rightarrow v_2}$. 
We consider the set 
\[G^{v_2}=\{\Phi^{v_2}: \Phi\in G \mbox{~and~} \Phi^{v_2}(v)=\Phi(v+v_2)-v_2\;(v\in V)\}.\]
The set $G^{v_2}$ is in fact the conjugate of $G$ by the translation by 
$v_2$ and hence it is a group of affine transformation isomorphic to $G$. 
(Furthermore, one can show that it has the same group of linear parts as
$G$.) 
The condition $\Phi(v_1)=v_2$ is equivalent to $\Phi^{v_2}(v_1-v_2)=0$. Thus
computing $G_{v_1\rightarrow v_2}$ is equivalent to
computing $G^{v_2}_{v_1-v_2\rightarrow 0}$. 

Thus it is sufficient to show an efficient algorithm for the vector
transporter problem of computing $G_{v_1\rightarrow 0}$ for
a unipotent group $G$ of affine transformation. The following
is crucial in the method of Luks and Miyazaki. Let $\Gamma$
be the group of the linear parts of the element of $G$. Assume that
$W$ is a $\Gamma$-invariant subspace of $V$. Then we decompose computing
$G_{v_1\rightarrow 0}$ into two smaller instances of the vector transporter
problem. First we compute $G_{v_1+W\rightarrow 0+W}$ 
for the induced action of $G$ on the factor $V/W$. (This action
is again affine.) If the solution set is empty then $G_{v_1\rightarrow 0}$ 
is also empty. Otherwise the solution set is a left coset of the stabilizer 
$H$ of the subspace $W$ in $G$. We can take any representative $\Phi_0$ of 
the coset, compute $v_1'=\Phi_0(v_1)\in W$
and proceed with computing $H_{v_1\rightarrow 0}$.

We triangularize the unipotent matrix group $\Gamma$: compute a basis of the subspace $W$
consisting of the common fixed points of the generators for $\Gamma$, 
and then proceed with the fixed points modulo $W$, and so on. This
way we obtain a basis in which all the elements of $\Gamma$ are
unitriangular (triangular with all 1 in the diagonal). To this basis
corresponds a flag \[V=V_1>V_2>\ldots >V_d>V_{d+1}=0\]
of $\Gamma$-invariant subspaces of $V$ such that $V_{i+1}$ has
codimension one in $V_i$. We follow the recursion described above
to ``descend'' along this flag. We show the details for the first
step (that from $V$ to $W=V_2$), all the subsequent steps are the same.

The factor space $V/W$ is one-dimensional over $\F_2$ 
and every element of $G$ acts as a translation by $0$ or $1$.
To compute the stabilizer of $W$, let $A_0$ be the subset 
of the given set of generators that act as the identity on $V/W$
and let $A_1$ consist of the rest. If $A_1$ is empty then
$A_0$ generates the stabilizer. If $A_1$ is non-empty, pick
an arbitrary element $\Phi\in A_1$. Then, by Schreier's lemma~\cite{Seress03}, the stabilizer of $W$ is generated by 
$A_0\cup \Phi\cdot A_1\cup\Phi\cdot A_0\cdot\Phi^{-1} \cup A_1\Phi^{-1}$.
If $v_1$ is in $W$ then the solution set
is the stabilizer of $W$. If $v_1\not \in W$ then the solution set is
empty if $A_1$ is empty, while it is the coset containing $\Phi$ otherwise.
As a $d$-dimensional affine matrix group has a natural $(d+1)$-dimensional
linear representation and membership in matrix groups over $\F_2$ that 
are 2-groups can be tested in polynomial time by \cite{Luks92}, 
one can reduce the set of generators to size at most $\log|G|$ to
prevent exponential blowup.
\end{proof}

We shall use Fact~\ref{fact:unip-transp} in the following context. Recall that we defined the matrix algebra $\cal A$ in relation to the quadratic form tuple ${\mathbf f_2}$ in Equation~\ref{eq:calA} and $R=I+\Rad(\cal A)$. We consider the restriction of the action 
$({\mathbf f},X)\mapsto {\mathbf f}^X$ of $\GL(n,\F_2)$ 
on $m$-tuples of quadratic forms to $R$ and we shall solve instances of
the vector transporter problem for this action with initial vector $\mathbf{f_1}^Y$
and target vector ${\mathbf f_2}$ for various $Y\in\GL(n,\F_2)$.
The action is an additive, i.e., $\F_2$-linear action of $R$ on an $m\cdot\binom{n+1}{2}$-dimensional space over $\F_2$. 
As $R$ is a 2-group, this matrix
representation of $R$ is also unipotent.

\begin{lemma}
\label{lem:semisimple-size}
Under the conditions of Problem~\ref{problem:large_block}, the dimension of ${\cal A}/\Rad({\cal A})$ is $\leq nD$.
\end{lemma}

\begin{proof}
Note that if $S_1^T\mathbf{D_{f_2}}=\mathbf{D_{f_2}}S_1'
=\mathbf{D_{f_2}}S_1''$ then \[0=(S_1-S_1)^T\mathbf{D_{f_2}}=
\mathbf{D_{f_2}}(S_1'-S_2''),\] which,  by the second condition in 
Problem~\ref{problem:large_block}, implies $S_1'-S_1''=0$. 
Thus for $S_1\in{\cal A}$ there is a unique counterpart $S_1'$.
Then it can also be seen that $(S_1S_2)'=S_2'S_1'$. 
The symmetry of $\mathbf{D_{f_2}}$ implies that if $S_1\in{\cal A}$
then $S_1'\in{\cal A}$ with $S_1''=S_1$. It can also be seen that 
$(S_1S_2)'=S_2'S_1'$, thus the map $S_1\mapsto S_1'$ is an anti-isomorphism
of ${\cal A}$. 

We implicitly use the standard known correspondence between idempotents 
in the endomorphism algebra of a module \cite{DrozdKirichenko}.
Instead of showing the  details of how to put our data $\mathbf{D_{f_2}}$
and ${\cal A}$ in this context, we are going to show the necessary facts
directly. Let $E$ be an idempotent matrix in ${\cal A}$ and consider
its counterpart $E'$. (We do not consider the zero matrix idempotent.)
Then, by the anti-isomorphic property, $E'$ is
also idempotent. Two idempotents $E$ and $F$ are called orthogonal
if $EF=FE=0$. An idempotent $E\in{\cal A}$ is called primitive if
it cannot be decomposed as a sum of two orthogonal idempotents. 
Note that an idempotent matrix is just a projection of the space
$V=\F_2^n$. Its kernel is $(I-E)V$. If $E$ is an idempotent in
$\cal A$ with counterpart $E'$, then it is straightforward to see
that \[E^T\mathbf{D_{f_2}}(I-E')=(I-E)^T\mathbf{D_{f_2}}E'=0,\]
which means that if $E\neq I$ then the two decompositions of $V$
into \[V=EV\oplus (I-E)V \text{ and }
V=E'V\oplus (I-E')V\] give a block decomposition of $\mathbf{D_{f_2}}$.
Conversely, one can show that pairs of projections corresponding to the block 
decomposition of $\mathbf{D_{f_2}}$ are (pairs of) idempotents of
${\cal A}$. Thus, minimal blocks correspond to primitive idempotents
and a decomposition of the identity matrix into a sum of pairwise
orthogonal primitive idempotents. 

We consider a Wedderburn complement ${\cal W}$ in ${\cal A}$,
a subalgebra isomorphic to ${\cal A}/\Rad({\cal A})$. For the existence of Wedderburn complements,
see, e.g., Section 6.2 of \cite{DrozdKirichenko}. Then ${\cal W}$ is
a direct sum of full matrix algebras over extensions of $\F_2$.
Let $E$ and $F$ be two primitive idempotents from the same simple
component of ${\cal W}$. Then, as the primitive idempotents
of a full matrix algebra are all similar to the matrix
whose entries are all zero except that in the upper left corner
which is one,
there exists an invertible matrix $M\in {\cal W}$ such that
$F=M^{-1}EM$. Then $$F^T\mathbf{D_{f_2}}F'=
M^TE^T{M^{-1}}^T\mathbf{D_{f_2}}M'E'{M'}^{-1}=
M^TE^T\mathbf{D_{f_2}}E'{M'}^{-1}.$$
(In the second equality we used that $M^T\mathbf{D_{f_2}}=
\mathbf{D_{f_2}}M'$.)
Thus the restrictions of $M$ and $M'^{-1}$ to the appropriate subspaces
give isomorphisms of the blocks corresponding to $E$ and $F$.
It follows that each simple component of ${\cal W}$
cannot contain more than $D$ pairwise orthogonal
primitive idempotents. This translates in terms of matrix algebra
dimensions as each simple component 
of ${\cal W}$ is isomorphic to $\M(d_i\times d_i,\F_{2^{m_i}})$ for 
some $d_i\leq D$ and for some positive integer $m_i$.
But ${\cal W}$ has a faithful
$n$-dimensional matrix representation over $\F_2$, from which we infer
$\sum_i d_im_i\leq n$. (Here we used the fact that the minimal
nonzero module for $\M(d_i\times d_i,\F_{2^{m_i}})$ has dimension
$d_im_i$ over $\F_2$). The dimension of 
${\cal W}\cong{\cal A}/\Rad{\cal A}$
over $\F_2$ is $\sum_id_i^2m_i\leq \sum Dd_i m_i\leq nD$.
\end{proof}

Here is a formal description of our algorithm.

\begin{framed}\vspace{-.3cm}\begin{flushleft}

\textbf{Infrequent Block Quadratic Form Tuple Isometry Algorithm}

\textbf{Input:} Two quadratic form tuples $\mathbf{f}_1$, $\mathbf{f}_2$.
\noindent 

\textbf{Output:}
Yes if $\mathbf{f}_1$ and $\mathbf{f}_2$ are isometric. No otherwise.

\begin{enumerate}

\item Find a pair $(S_0,S_0')$ of invertible matrices such that 
$\mathbf{D_{f_2}}=S_0^T\mathbf{D_{f_1}}S_0'$.
\begin{itemize}
\item Use e.g. the method described in Proposition 3.2 of
\cite{IQ19} to do this in polynomial time. 
\item Return no if no such pair exists.
\end{itemize}
\item Compute a basis of ${\cal A}$.
\begin{itemize}
\item Compute a basis of the space of pairs of matrices \\
\mbox{~~~~}${\cal A}'=\{(S_1,S_1')\in \M(n \times n,\F_2)^2:
S_1\mathbf{D_{f_2}}=\mathbf{D_{f_2}}S_1'\}$
\item[] By linearity of the condition, ${\cal A}'$ is indeed a vector
space and a basis for it can be computed by solving a system of homogeneous
linear equations. 
\item Take the first coordinate of each element of the computed basis
of ${\cal A}'$.
\end{itemize}
\item Compute a basis of $\Rad(\cal{A})$.
\begin{itemize}
\item[]
This can be efficiently done by. e.g., the method of R\'onyai~\cite{Ronyai90}.
\end{itemize}
\item Compute a system of generators for the group 
$R=I+\Rad(\cal{A})$
\begin{itemize}
\item[] Use e.g., the method of Lemma~\ref{lem:generators-unip}.
\end{itemize}
\item Compute a subspace ${\cal C}$ of ${\cal A}$ complementary to 
$\Rad(\cal{A})$
\begin{itemize}
\item Extend the basis of $\Rad({\cal A})$ to a basis of ${\cal A}$.
\item The elements of this basis not in $B$ are a basis $C$ of 
a complementary subspace ${\cal C}$. 
\end{itemize}
\item For each invertible matrix $Z$ in ${\cal C}$ 
\begin{itemize}
\item Run the
vector transporter algorithm of Fact~\ref{fact:unip-transp} for the action
$(X,{\mathbf{f}})\rightarrow \mathbf{f}^{X}$ of the group 
$R$, with initial vector $\mathbf{f_1}^{S_0Z}$ and
target $\mathbf{f_2}$.
The result is the set 
$\{X\in R:\mathbf{f_2}=\mathbf{f_1}^{S_0ZX}\}$.
\item If the set above is non-empty, return yes.
\end{itemize}
\item Return no.
\end{enumerate}
\end{flushleft}
\vspace{-.3cm}
\end{framed}

\begin{proof}[Proof of Theorem~\ref{thm:large_block}]
The correctness of the algorithm above follows from
Lemma~\ref{lem:in-adjoint}. Its running time is 
$|{\cal A}/\Rad({\cal A})|\cdot \poly(mn)$ and 
$|{\cal A}/\Rad({\cal A})|\leq 2^{nD}$ 
by
Lemma~\ref{lem:semisimple-size}.
\end{proof}

\section{Quadratic form tuple isometry algorithm}
In this section, we first present an algorithm to solve the following quadratic form tuple isometry problem, and then show that the quadratic form tuple constructed by the Quadratic Form Tuple Construct Algorithm in Section~\ref{sec:quadratic_form_tuple_construction} can be solved by this algorithm.

\begin{problem}\label{problem:final_tuple}
Given two quadratic form tuples $\mathbf{f}_1$ and $\mathbf{f}_2$ over $n$ variables, and representative blocks $\mathbf{B}_1, \dots, \mathbf{B}_\ell$ such that both $\mathbf{f}_1$ and $\mathbf{f}_2$ satisfy the $\rho$-non-zero-square-term condition with respect to any standard congruent block diagonalization of $\mathbf{D}_{\mathbf{f}_1}$ with respect to $\mathbf{B}_1, \dots, \mathbf{B}_\ell$, 
decide whether there exists an invertible matrix $S$ such that $\mathbf{f}_1^S = \mathbf{f}_2$. 
\end{problem}

\begin{theorem}\label{thm:quadratic_form_tuple_iso_algo}
There is an algorithm to solve Problem~\ref{problem:final_tuple} in time $2^{O(n^{1.5} + n\cdot \rho)}$.  
\end{theorem}


\begin{lemma}\label{lem:matrix_kernal}
Let $\mathbf{A}$ be a matrix tuple and $U = \begin{pmatrix}U_1 & \cdots  & U_k \end{pmatrix}$, $W = \begin{pmatrix}W_1  & \dots  & W_k \end{pmatrix}$ be matrices such that $U^T \mathbf{A} W = \diag(U_1^T \mathbf{A} W_1, \dots, U_k^T \mathbf{A} W_k)$. 
Let $Q$ be a matrix such that $Q - I =  \begin{pmatrix}Q_1 & \cdots & Q_k \end{pmatrix}$ satisfying the following two conditions:
\begin{enumerate}
\item $Q_i$ is of the same dimension as $U_i$ for any $i \in [k]$.
\item $U Q_i\in K_{U_i^T \mathbf{A} W_i}$ for any $i \in [k]$, where $K_{U_i^T \mathbf{A} W_i}$ is defined as Theorem~\ref{thm:left_right_decomposition}.
\end{enumerate}
Then $Q$ is an invertible matrix such that there exists another invertible matrix $P$ satisfying \[Q^T U^T \mathbf{A} W P = U^T \mathbf{A} W.\]
\end{lemma}
\begin{proof}
We prove the case that $Q_1$ is a non-zero matrix, and $Q_2, \dots, Q_k$ are all zero matrices. 
Then lemma holds for the case that only one $Q_i$ is non-zero among all the $Q_1, \dots, Q_k$. 
The general case can be obtained by iteratively applying the special case to each non-zero $Q_i$. 

$Q_1$ can be uniquely represented as $F + G$ satisfying \[\begin{split}Q_1^T U^T \mathbf{A} W = & (F^T + G^T) U^T \mathbf{A} W \\ = & F^T \begin{pmatrix} U_1^T
\mathbf{A}W_1 & 0 & 0 & 0 \\
0 & 0 \\
0 & & \cdots \\
0 & & & 0\end{pmatrix} + G^T \begin{pmatrix} 
0 & 0 & 0 & 0\\ 
0 & U_2^T  \mathbf{A}W_2 & & \\ 
0 & & \dots & \\
0 & & & U_k^T \mathbf{A}W_k \end{pmatrix}.\end{split}\]
By Theorem~\ref{thm:left_right_decomposition} and the second condition of the lemma, both $F^T U^T \mathbf{A} W$ and $G^T U^T \mathbf{A} W$ are in $K_{U_1^T \mathbf{A} W_1}$.
Let $Q' = I + \begin{pmatrix}F & 0 & \cdots & 0\end{pmatrix}$. 
We have 
\[Q'^T U^T \mathbf{A} W  = \diag(U_1^T \mathbf{A} W_1 + F^T U^T \mathbf{A} W, \dots, U_k^T \mathbf{A} W_k).\]
Since $U_1^T \mathbf{A} W_1$ is in $W_{U_1^T \mathbf{A} W_1}$, by Theorem~\ref{thm:left_right_decomposition}, \[U_1^T \mathbf{A} W_1 + F^T U^T \mathbf{A} W\] is also in  $W_{U_1^T \mathbf{A} W_1}$. 
Hence, there exists an invertible $P'$ such that \[Q'^T U^T \mathbf{A} W P' = U^T \mathbf{A} W.\]

Similarly, let $Q'' = I + \begin{pmatrix}G & 0 & \cdots & 0\end{pmatrix}$.
There exists another invertible matrix $P''$ such that 
\[Q''^T U^T \mathbf{A} W P'' = U^T \mathbf{A} W.\] 
Notice that 
\[Q = I +  \begin{pmatrix}Q_1 &  0 & \cdots & 0  \end{pmatrix} = \begin{pmatrix}I + \begin{pmatrix}F & 0 & \cdots & 0\end{pmatrix} \end{pmatrix} \begin{pmatrix}I + \begin{pmatrix}G & 0 & \cdots & 0\end{pmatrix}\end{pmatrix}.\] The lemma follows.
\end{proof}
\begin{lemma}\label{lem:congruence_canceling}
Let $\mathbf{A} = \diag(\mathbf{A}_1, \mathbf{A}_2)$ be an alternating matrix tuple such that every minimum left-right block of $\mathbf{A}_1$ is not a minimum left-right block of $\mathbf{A}_2$. 
Let $n_1$ and $n_2$ denote the dimensions of $\mathbf{A}_1$ and $\mathbf{A}_2$ respectively. 
Let $F$ and $F'$ be arbitrary $n_1 \times n_2$ and $n_2 \times n_1$ matrices respectively such that $\mathbf{A}_1 F =  (\mathbf{A}_2 F')^T$.
Then the matrix 
\[\begin{pmatrix}I & F \\ F' & I\end{pmatrix}\] is invertible such that
\[\begin{pmatrix}I & F \\ F' & I\end{pmatrix}^T \mathbf{A} \begin{pmatrix}I & F \\ F' & I\end{pmatrix} = \diag(\mathbf{A}_1', \mathbf{A}_2')\]
for some $n_1 \times n_1$ $\mathbf{A}_1' = \mathbf{A}_1 W_1$ and $n_2\times n_2$ $\mathbf{A}_2' = \mathbf{A}_2 W_2$ with some invertible matrices $W_1$ and $W_2$. 
\end{lemma}
\begin{proof}
We first consider the case where $\mathbf{A}_1$ and $\mathbf{A_2}$ are in their standard congruent block diagonalization as defined by Definition~\ref{def:standard_diagonalization}. 
By $\mathbf{A}_1 F =  (\mathbf{A}_2 F')^T$, we have 
\[\begin{split}
& \begin{pmatrix}I & F \\ F' & I\end{pmatrix}^T \mathbf{A} \begin{pmatrix}I & F \\ F' & I\end{pmatrix} \\
= & \begin{pmatrix}I & F \\ F' & I\end{pmatrix}^T \begin{pmatrix} \mathbf{A}_1 & \\ & \mathbf{A}_2 \end{pmatrix} \begin{pmatrix}I & F \\ F' & I\end{pmatrix}\\
= & \begin{pmatrix} \mathbf{A_1} + F'^T \mathbf{A_2} F' & \mathbf{A}_1 F  + F'^T \mathbf{A_2} \\ F^T \mathbf{A_1} + \mathbf{A}_2 F'& F^T \mathbf{A_1} F + \mathbf{A_2}\end{pmatrix} \\
= & \diag( \mathbf{A_1} + F'^T \mathbf{A_2} F', F^T \mathbf{A_1} F + \mathbf{A_2}). 
\end{split}
\]
By Lemma~\ref{lem:standard_diagonalization_1}, there are invertible matrices $Z_1$ and $Z_2$ such that 
$\mathbf{A}_1 Z_1 = \diag(\mathbf{B}_1, \dots, \mathbf{B}_k)$ and $\mathbf{A}_2 Z_2 = \diag(\mathbf{C}_1, \dots, \mathbf{C}_\ell)$ such that 
$\mathbf{G}_1, \dots, \mathbf{G}_g$ and $\mathbf{H}_1, \dots, \mathbf{H}_h$ are minimum left-right blocks for $\mathbf{A}_1$ and $\mathbf{A}_2$ respectively.
By $\mathbf{A}_1 F =  (\mathbf{A}_2 F')^T$ and Lemma~\ref{lem:matrix_kernal} with $U = I$, $W = \diag(Z_1, Z_2)$, 
and \[Q = I + \begin{pmatrix}& F \\ F' & \end{pmatrix},\]  
the matrix \[\begin{pmatrix}I & F \\ F' & I\end{pmatrix}\]is invertible. 

For the general case, by Corollary~\ref{cor:standard_block_diagonalization}, there are invertible matrices $U_1$ and $U_2$ such that $U_1^T \mathbf{A}_1 U_1$ and $U_2^T \mathbf{A}_2 U_2$ are standard congruent block diagonalizations for $\mathbf{A}_1$ and $\mathbf{A}_2$ respectively.
Since $\mathbf{A}_1 F = (\mathbf{A}_2 F')^T$, 
we have 
\[U_1^T \mathbf{A}_1 U_1 U_1^{-1} F = U_1^T (\mathbf{A}_2 F')^T = U_1^T (U_2^{-T} U_2^T \mathbf{A}_2 U_2 U_2^{-1} F')^T = U_1^T F'^T U_2^{-T} (U_2^T \mathbf{A}_2 U_2
)^T U_2^{-1},\] which implies 
$(U_1^T\mathbf{A}_1 U_1) (U_1^{-1}F U_2) = ((U_2^T \mathbf{A}_2 U_2) (U_2^{-1} F' U_1))^T$. By the first case, we have
\[\begin{pmatrix}I & U_1^{-1}F U_2 \\ U_2^{-1} F' U_1 & I\end{pmatrix}^T \begin{pmatrix}U_1 & \\ & U_2 \end{pmatrix}^T\mathbf{A}  \begin{pmatrix}U_1 & \\ & U_2 \end{pmatrix}\begin{pmatrix}I & U_1^{-1}F U_2 \\ U_2^{-1} F' U_1 & I\end{pmatrix} = \diag(\mathbf{A}_1'', \mathbf{A}_2'')\]
Consequently, let 
\[S =  \begin{pmatrix}U_1 & \\ & U_2 \end{pmatrix}\begin{pmatrix}I & U_1^{-1}F U_2 \\ U_2^{-1} F' U_1 & I\end{pmatrix} \begin{pmatrix}U_1^{-1} \\ & U_2^{-1}\end{pmatrix}.\]
We have 
\[ S^T \mathbf{A}  S 
=  \begin{pmatrix} I & F \\ F' & I\end{pmatrix}^T \mathbf{A} \begin{pmatrix} I & F \\ F' & I\end{pmatrix} 
= \diag(U_1^{-T}\mathbf{A}_1''U_1^{-1}, U_2^{-T}\mathbf{A}_2''U_2^{-1}).\]
Then the lemma follows.
\end{proof}

\begin{lemma}\label{lem:block_diagonal_invertible}
Let $\mathbf{A}$ be an alternating matrix tuple, and let $\Lambda$ be a collection of minimum left-right blocks of $\mathbf{A}$ such that no two blocks in $\Lambda$ are left-right equivalent and any minimum left-right block of $\mathbf{A}$ is equivalent to a block in $\Lambda$. 
Let $\Lambda_1$ and $\Lambda_2$ be a partition of $\Lambda$ such that for any $\mathbf{B} \in \Lambda_1$, there is a block left-right equivalent to $\mathbf{B}^T$ in $\Lambda_1$. 
Let  $\mathbf{A}'= \diag(\mathbf{A}_1', \mathbf{A}_2')$ and $\mathbf{A}''= \diag(\mathbf{A}_1'', \mathbf{A}_2'')$  be two alternating matrix tuples satisfying the following conditions:
\begin{enumerate}
    \item Both $\mathbf{A}'$ and $\mathbf{A}''$ are congruent to $\mathbf{A}$. 
    \item For any $i \in \{1, 2\}$, any minimum left-right block of either $\mathbf{A}_i'$ or $\mathbf{A}_i''$ is left-right equivalent to a block in $\Lambda_i$. 
\end{enumerate}
Let $n_1$ and $n_2$ denote the number of rows/columns of $\mathbf{A}_1'$ and $\mathbf{A}_2'$ respectively. 
Then $\mathbf{A}_1''$ and $\mathbf{A}_2''$ have $n_1$ and $n_2$ rows/columns respectively. 
And for any invertible matrix \[U = \begin{pmatrix} U_{1, 1} & U_{1,2} \\ U_{2,1} & U_{2,2} \end{pmatrix}\]
with $U_{1, 1} \in \M(n_1 \times n_1, \F_2)$, $U_{1, 2} \in \M(n_1 \times n_2, \F_2)$, $U_{2, 1} \in \M(n_2 \times n_1, \F_2)$, and $U_{2, 2} \in \M(n_2 \times n_2, \F_2)$ such that $U^T \mathbf{A}' U = \mathbf{A}''$, 
$U_{1, 1}$ and $U_{2,2}$ are invertible. 
\end{lemma}
\begin{proof}
Without loss of generality, we can assume that there exists an invertible matrix $W$ such that $ \mathbf{A}_1'' W = \diag(\mathbf{C}_1, \dots, \mathbf{C}_k)$ is a minimum left-right block diagonalization of $\mathbf{A}_1''$. 
Denote 
\[U_{2, 1} = \begin{pmatrix} U_{2, 1, 1} & \dots & U_{2, 1, k}\end{pmatrix}\] such that $U_{2, 1, i}$ corresponds to $\mathbf{C}_i$ within $U_{2, 1}$. 
By Theorem~\ref{thm:left_right_decomposition}, \[\begin{pmatrix} 0 &  U_{2, 1, i}\end{pmatrix}\] is an element in $K_{\mathbf{C}_i}$ with respect to $\mathbf{A}'$ for any $i \in [k]$, because otherwise $\mathbf{C}_i$ is a left-right block of $\mathbf{A}_2'$, contradicting the condition that $\mathbf{A}_{2}'$ does not contain any minimum left-right block equivalent to $\Lambda_1$. 

By Theorem~\ref{thm:left_right_decomposition} and Lemma~\ref{lem:congruence_canceling}, 
there are invertible matrices $Q$ and $P$ such that 
\[Q U^T \mathbf{A}' =  \begin{pmatrix}U_{1, 1} & U_{1, 2}\\ & U_{2,2}\end{pmatrix}^T \mathbf{A}'\]
and \[Q U^T\mathbf{A}' U P = \mathbf{A}''.\] 
Therefore, $U_{1, 1}$ is an invertible matrix. 
The case for $U_{2, 2}$ is similar. 
\end{proof}

We have the following algorithm for the Problem~\ref{problem:final_tuple}:

\begin{framed}\vspace{-.3cm}\begin{flushleft}
\textbf{Non-Zero-Square-Term Condition Quadratic Form Tuple Isometry Algorithm}

\textbf{Input:} Two quadratic form tuples $\mathbf{f}_1$ and $\mathbf{f}_2$, representative blocks $\mathbf{B}_1, \dots, \mathbf{B}_\ell$ for $\mathbf{D}_{\mathbf{f}_1}$, and two parameters $0 \leq \rho \leq n$ and $0 < \Delta < 1$.

\textbf{Output:}
Yes or no. 
\begin{enumerate}
\item Compute canonical forms for $\mathbf{D}_{\mathbf{f}_1}$ and $\mathbf{D}_{\mathbf{f}_2}$ under the left-right action. Return no if the canonical forms are different. 
For each $i \in \{1, 2\}$, run the Standard Congruent Block Diagonalization Algorithm for $\mathbf{D}_{\mathbf{f}_i}$ with respect to $\mathbf{B}_1, \dots, \mathbf{B}_\ell$, and let $S_i$ be the output matrix. 
\item Let $r \in \{0, 1, \dots, \ell\}$ be the largest integer such that, in any minimum left-right block diagonalization of $\mathbf{D}_{\mathbf{f}_1}$, at most $n^{ 1- \Delta}$ minimum left-right blocks are left-right equivalent to $\mathbf{B}_1, \dots, \mathbf{B}_r$. 
If $r = 0$, then run the Frequent Block Quadratic Form Tuple Isometry Algorithm for $\mathbf{f}_1$ and $\mathbf{f}_2$, and return the output. 
If $r = \ell$, then run the Infrequent Block Quadratic Form Tuple Isometry Algorithm for $\mathbf{f}_1$ and $\mathbf{f}_2$, and return the output. 

\item Let $\mathbf{f} = \mathbf{f}_1^{S_1}$ and $\Gamma_1 = \{\mathbf{f}\}$. Denote  $\mathbf{D}_{\mathbf{f}}$ as $\diag(\mathbf{D}_{1}, \dots, \mathbf{D}_{k})$ with $\mathbf{
D}_{1}$ being the unions of standard congruence blocks obtained from minimum left-right blocks same to one of $\mathbf{B}_1, \dots, \mathbf{B}_r$, and $\mathbf{D}_{2}, \dots, \mathbf{D}_{k}$ being the rest standard congruent blocks in $\mathbf{D}_{\mathbf{f}}$.
Let $d_1, d_2, \dots d_{k}$ be the number of rows/columns of $\mathbf{D}_{1}, \dots \mathbf{D}_{k}$. 
\item For $i = 2, \dots, k$, an arbitrary $\mathbf{g} \in \Gamma_{i-1}$
with $\mathbf{D}_\mathbf{g}  = \diag(\mathbf{D}_1', \dots, \mathbf{D}_{k}')$ such that $\mathbf{
D}_{1}'$ is the unions of standard congruence blocks obtained from minimum left-right blocks same to one of $\mathbf{B}_1, \dots, \mathbf{B}_r$, and $\mathbf{D}_{2}', \dots, \mathbf{D}_{k}'$ are the remaining standard congruent blocks in $\mathbf{D}_{\mathbf{g}}$,
and arbitrary matrices $F$ and $F'$ such that $\mathbf{D}_1' F =  (\mathbf{D}_i' F')^T$, 
add $\mathbf{g}^S$ to $\Gamma_i$ with 
\[S=I+\begin{pmatrix}
\dots & F & \dots \\
\vdots &  & \vdots \\
F' \\
\vdots &  & \vdots 
\end{pmatrix},
\]
in which the submatrix induced by the first $d_1$ rows and the columns from $((\sum_{j=1}^{i-1} d_j) +1)$ to $(\sum_{j=1}^{i} d_j)$ is $F$,
the submatrix induced by the rows from  $((\sum_{j=1}^{i-1} d_j) +1)$ to $(\sum_{j=1}^{i} d_j)$ and the first $d_1$ columns is $F'$, and all the other entries are zero.

\item For every $\mathbf{g} \in \Gamma_{k}$, 
run the Infrequent Block Quadratic Form Tuple Isometry Algorithm for $\mathbf{g}$ and $\mathbf{f}_2^{S_2}$ induced on the first $d_1$ variables,
and the Frequent Block Quadratic Form Tuple Isometry Algorithm for $\mathbf{g}$ and $\mathbf{f}_2^{S_2}$ induced on the last $n - d_1$ variables. If both outputs are yes, then return yes. 
\item Return no.  
\end{enumerate}
\end{flushleft}
\vspace{-.3cm}\end{framed}

\begin{proof}[Proof of Theorem~\ref{thm:quadratic_form_tuple_iso_algo}]
We first show that the Non-Zero-Square-Term Condition Quadratic Form Tuple Isometry Algorithm solves the Problem~\ref{problem:final_tuple} correctly. 

\begin{lemma}\label{lem:quadratic_form_tuple_iso_algo_correctness}
The Non-Zero-Square-Term Condition Quadratic Form Tuple Isometry Algorithm solves the Problem~\ref{problem:final_tuple} correctly. 
\end{lemma}
\begin{proof}
By Corollary~\ref{cor:standard_cong_block_diag_algo} and Lemma~\ref{lem:congruence_canceling}, 
for any $\mathbf{g} \in \{\mathbf{f}_1^{S_1}$, $\mathbf{f}_2^{S_2}\} \cup \Gamma_{2} \cup \dots \cup \Gamma_k$, $\mathbf{g}[i, j] = 0$ for any $1 \leq i \leq d_1$ and $d_1 + 1 \leq j \leq n$. 
Hence, if the algorithm outputs yes, then the two quadratic form tuples are isometric. 

In the rest of this proof, we show that if $\mathbf{f}_1$ and $\mathbf{f}_2$ are isometric, then there exists a quadratic form tuple $\mathbf{g} \in \Gamma_k$ such that $\mathbf{g}^S = \mathbf{f}_2^{S_2}$ for some $S = \diag(S_1', S_2')$ with invertible $S_1'$ and $S_2'$.  

Since $\mathbf{f}_1$ and $\mathbf{f}_2$ are isometric, there exists an invertible matrix 
\[Q = \begin{pmatrix}Q_{1, 1} & Q_{1,2} \\ Q_{2,1} & Q_{2,2}\end{pmatrix}\] such that $\mathbf{f}_1^{S_1 Q} = \mathbf{f}_2^{S_2}$. 
By Lemma~\ref{lem:block_diagonal_invertible}, $Q_{1,1}$ and $Q_{2, 2}$ are invertible. 
We show that there exists a quadratic form tuple $\mathbf{g} \in \Gamma_k$ such that $\mathbf{g}^{R} = \mathbf{f}_2^{S_2}$ for 
\[R = \begin{pmatrix}Q_{1, 1} & \\ & R' \end{pmatrix}\] with some invertible $R'$ by induction.

Let 
\[Q' = \begin{pmatrix}I & Q_{1, 2} \\ Q_{2, 1} Q_{1, 1}^{-1} & Q_{2,2} \end{pmatrix}.\]
We have 
\[\mathbf{f}_1^{S_1 Q' \begin{pmatrix}Q_{1, 1} & \\ & I\end{pmatrix}} = \mathbf{f}_1^{S_1 Q} = \mathbf{f}_2^{S_2}.\]
Denote $Q_{2, 1} Q_{1,1}^{-1}$ as \[\begin{pmatrix}F_2 \\ \dots \\ F_{k}\end{pmatrix}\] 
such that 
\[\diag(\mathbf{D}_{2}, \dots, \mathbf{D}_{k}) Q_{2, 1} Q_{1,1}^{-1} = \begin{pmatrix}\mathbf{D}_{2} F_2 \\ \dots \\ \mathbf{D}_{k} F_{k}\end{pmatrix}.\]
Since $\mathbf{f}_2^{S_2}$ is congruent block diagonal, 
every row of  $\diag(\mathbf{D}_{2}, \dots, \mathbf{D}_{k}) Q_{2, 1} Q_{1,1}^{-1}$ is a linear combination of the rows of $\mathbf{D}_1$. 

By Lemma~\ref{lem:congruence_canceling} and induction,  for each $1 \leq i\leq k$, there is a $\mathbf{g} \in \Gamma_i$ such that $\mathbf{g} = \mathbf{f}_1^{S_1 Q_i}$ for 
\[Q_i = \begin{pmatrix}I & Q_{i,1, 2} \\ F_i & Q_{i, 2, 2}\end{pmatrix}\] with \[F_i = \begin{pmatrix}F_2 \\ \dots \\ F_i \\ 0 \\ \dots \\ 0\end{pmatrix}\]
and some $Q_{i, 1, 2}, Q_{i, 2, 2}$. 

Therefore, there is a $\mathbf{g} \in \Gamma_k$ such that $\mathbf{g} = \mathbf{f}_1^{S_1 Q_k}$ with 
\[Q_k = \begin{pmatrix}I & Q_{k, 1, 2} \\ Q_{2, 1}Q_{1, 1}^{-1} & Q_{k, 2, 2}\end{pmatrix}. \]
Since $\mathbf{g}$ is also block diagonal, there is an $R'$ such that \[Q_k \begin{pmatrix} Q_{1, 1} \\ & R' \end{pmatrix} = Q.\]
\end{proof}

\begin{lemma}\label{lem:quadratic_form_tuple_iso_algo_time}
The Non-Zero-Square-Term Condition Quadratic Form Tuple Isometry Algorithm has a $2^{O(n (n^\Delta + \rho) + n^{2 - \Delta})} \cdot \mathrm{poly}(m)$ running time.
\end{lemma}
\begin{proof}
We show that $|\Gamma_i| \leq  2^{4n \cdot n^\Delta}$ for any $1 \leq i \leq k$. Then the lemma follows from the solutions to Problem~\ref{problem:large_block} and Problem~\ref{problem:small_block}.

By the definition of $r$ in the Non-Zero-Square-Term Condition Quadratic Form Tuple Isometry Algorithm, $\ell - r$ is at most $n^{\Delta}$. For each $2 \leq j \leq k$, we assign $\mathbf{D}_j$ to one of $\mathbf{B}_{r+1}, \dots, \mathbf{B}_\ell$ according to the minimum left-right block from which $\mathbf{D}_j$ is obtained. 


For each $r+1 \leq i \leq \ell$, we further divide all $\mathbf{D}_{j}$ assigned to $\mathbf{B}_i$ into two subsets: $\Phi_{\mathbf{B}_i, 1}$ contains all standard congruent blocks corresponding to a single minimum left-right block, and $\Phi_{\mathbf{B}_i, 2}$ contains all standard congruent blocks corresponding to two minimum left-right blocks. 

In this way, we divide all the $\mathbf{D}_{2}, \dots, \mathbf{D}_{k}$ into at most $2n^{\Delta}$ subsets such that for any $\mathbf{D}_{j}$ and $\mathbf{D}_{j'}$ from the same subset, there exists an invertible matrix $Q$ such that $\mathbf{D}_{j} = \mathbf{D}_{j'} Q$ by Lemma~\ref{lem:two_blocks_right_multi}. 
Hence, for each subset $\Phi_{\mathbf{B}_i, s}$ with $r+1 \leq i \leq \ell$ and $s \in \{1, 2\}$, there exists a sequence of at most $h_{i, s}^2$ matrix tuples such that every matrix tuple in $\Phi_{\mathbf{B}_i, s}$ is a linear combination of this sequence of matrix tuples, where $h_{i, s}$ is the number of rows/columns of the standard congruent blocks in $\Phi_{\mathbf{B}_i, s}$, which is at most $n^{\Delta}$. 

Now we bound the number of different $\mathbf{D}_1'$ for all the quadratic forms in $\Gamma_t$. 
Notice that $\mathbf{D}_1'$ equals $\mathbf{D}_{1} + \sum_{j = 2}^{t} F_j'^T \mathbf{D}_{j} F_j'$ for some $h_{i, s} \times d_1$ matrix $F_j'$, where $\Phi_{\mathbf{B}_i, s}$ is the subset to which $\mathbf{D}_j$ belongs.  
For any $\mathbf{B}_{i}$ and $s \in \{1, 2\}$, 
let $F_{1}, \dots, F_{h_{i, s} \cdot d_1}$ be a linear basis for the linear space of all the $h_{i, s} \times d_1$ matrices over $\F_2$, and $\mathbf{E}_1, \dots, \mathbf{E}_u$ for some $u \leq h_{i, s}^2$ be a sequence of matrix tuples such that every matrix in $\Phi_{\mathbf{B}_i, s}$ is a linear combination of  $\mathbf{E}_1, \dots, \mathbf{E}_u$. 
Hence,
\[\sum_{2 \leq j \leq k : \mathbf{D}_{j} \in \Phi_{\mathbf{B}_i, s}} F_j'^T \mathbf{D}_{j } F_j' \]
is a linear combination of $F_g \mathbf{E}_h F_{g'}$ for all $1 \leq g, g' \leq h_{i, s} \cdot d_1$ and $1 \leq h \leq u$, and has $2^{h_{i, s} \cdot d_1} \cdot 2^u \cdot 2^{h_{i, s} \cdot d_1}$ possibilities. 
Hence,  the number of possible $\mathbf{D}_1'$ is at most 
\[\prod_{r + 1 \leq i \leq \ell, 1 \leq s \leq 2} 2^{2 h_{i, s}\cdot d_1 +h_{i, s}^2 } \leq 2^{2 d_1 \cdot \sum_{r+1 \leq i \leq \ell, 1 \leq s \leq 2} h_{i, s}} \cdot 2^{ \sum_{r + 1 \leq i \leq \ell, 1 \leq s \leq 2}h_{i, s}^2} \leq 2^{2n \cdot 2n^\Delta} \cdot 2^{2n^{2\Delta}} = 2^{O(n^{1 + \Delta})}.\]

For each $2 \leq j \leq k$, by Lemma~\ref{lem:congruence_canceling}, $\mathbf{D}_j'$ has at most $2^{d_j^2}$ possibilities. 
By the $\rho$-non-zero-square-term condition, each of $\Gamma_i$ contains at most 
\[2^{O(n^{1+\Delta})}\cdot 2^{\sum_{2 \leq j \leq k} d_j^2} \cdot 2^{n\rho} = 2^{O(n\cdot (n^\Delta + \rho))}\] different quadratic form tuples. 

The lemma is then obtained by Theorem~\ref{thm:left_right_decomposition_algorithm}, Theorem~\ref{thm:small_block}, and Theorem~\ref{thm:large_block}. 
\end{proof}
The theorem follows by Lemma~\ref{lem:quadratic_form_tuple_iso_algo_correctness} and Lemma~\ref{lem:quadratic_form_tuple_iso_algo_time} with $\Delta = 1/2$. 
\end{proof}

Now we prove Theorem~\ref{thm:matrix_tuple_constructed_main}.
Our algorithm first constructs representative blocks for the partial derivative matrix tuple for the quadratic form tuples constructed by the Quadratic Form Tuple Construction Algorithm, and then runs the Quadratic Form Tuple Isometry Algorithm.

\begin{framed}
\noindent \textbf{Representative Blocks Algorithm}

\noindent \textbf{Input:} Quadratic form $\mathbf{f}$ constructed by the Quadratic Form Tuple Construction for $f_1, \dots, f_m$ over $n$ variables with respect to parameters $\alpha_0, \alpha_1, \beta_0, \beta_1$. 

\noindent \textbf{Output:} Representative Blocks $\mathbf{B}_1, \dots, \mathbf{B}_\ell$ for $\mathbf{D}_\mathbf{f}$.
\begin{enumerate}

\item Run the Left-Right Canonical Form Algorithm for $\mathbf{D}_\mathbf{f}$, and denote the output as $\diag(\mathbf{C}_1, \dots, \mathbf{C}_k)$ with invertible matrices $U = \begin{pmatrix}U_1 & \cdots & U_k\end{pmatrix}$ and $W= \begin{pmatrix}W_1 & \cdots & W_k\end{pmatrix}$ such that $U^T \mathbf{D}_\mathbf{f} W = \diag(\mathbf{C}_1, \dots, \mathbf{C}_k)$ with $\mathbf{C}_i = U_i^T \mathbf{D}_\mathbf{f} W_i$ for any $i \in [k]$. 

\item Compute arbitrary invertible $U' = \diag(U_1', \dots, U_k')$ and $W
' = \diag(W_1', \dots, W_k')$ such that every column of $UU'$ and $W W'$ is in one of $E_1, \dots, E_7$ (defined in Section~\ref{sec:quadratic_form_tuple_construction}), and 
$U'^T U^T \mathbf{D}_\mathbf{f} W W' = \diag(U_1'^T \mathbf{C}_1 W_1', \dots, U_k'^T \mathbf{C}_k W_k')$ satisfying if $(U_i'^T \mathbf{C}_i W_i')^T$ is left-right equivalent to $U_i'^T \mathbf{C}_i W_i'$ for some $1 \leq i\leq k$, then $(U_i'^T \mathbf{C}_i W_i')^T= U_i'^T \mathbf{C}_i W_i' Z_i$ for some invertible $Z_i$. 
\item Output arbitrary representative blocks $\mathbf{B}_1, \dots, \mathbf{B}_\ell$ for $\mathbf{D}_\mathbf{f}$ such that every block is equal to $U_i'^T \mathbf{C}_i W_i'$ for some $1 \leq i \leq k$. 
\end{enumerate}
\end{framed}

\begin{lemma}\label{lem:representative_block_algo}
Given quadratic form tuple $\mathbf{f}$ constructed by the Quadratic Form Tuple Construction process for $f_1, \dots, f_m$ over $n$ variables with respect to parameters $\alpha_0, \alpha_1, \beta_0, \beta_1$,  
    the Representative Blocks Algorithm outputs a sequence of representative blocks $\mathbf{B}_1, \dots, \mathbf{B}_\ell$ for $\mathbf{D}_\mathbf{f}$ in time polynomial in $n$ and $m$ such that $\mathbf{f}$ satisfies the $(m - \alpha_0 - \alpha_1)$-non-zero-square-term condition with respect to any standard congruent block diagonalization of $\mathbf{D}_\mathbf{f}$. 
\end{lemma}
\begin{proof}
    The correctness of the algorithm is obtained by Lemma~\ref{lem:nonzero_suare_term_condition}, Corollary~\ref{cor:standard_block_diagonalization}, and Definition~\ref{def:standard_diagonalization}. The running time of the algorithm is obtained by Theorem~\ref{thm:left_right_decomposition_algorithm}. 
\end{proof}

\begin{proof}[Proof of Theorem~\ref{thm:matrix_tuple_constructed_main}]
The algorithm first runs the Representative Blocks Algorithm for $\mathbf{f}$ to obtain representative blocks $\mathbf{B}_1, \dots, \mathbf{B}_\ell$ for $\mathbf{D}_\mathbf{f}$. Then the algorithm runs the Quadratic Form Tuple Isometry Algorithm for $\mathbf{f}, \mathbf{g}$ and  $\mathbf{B}_1, \dots, \mathbf{B}_\ell$. 

Lemma~\ref{lem:nonzero_suare_term_condition}, Corollary~\ref{cor:standard_block_diagonalization}, Definition~\ref{def:standard_diagonalization}, and Lemma~\ref{lem:representative_block_algo}, both $\mathbf{f}$ and $\mathbf{g}$ satisfy the $(m - \alpha_0 - \alpha_1)$-non-zero-square-term condition for any matrix tuple that is a standard congruent block diagonalization of $\mathbf{D}_\mathbf{f}$ with respect to $\mathbf{B}_1, \dots, \mathbf{B}_\ell$. 
Since $\mathbf{f}$ and $\mathbf{g}$ are quadratic form tuples of length $\mathrm{poly}(n, m)$ over $O(n+m)$ variables by the Quadratic Form Tuple Construction process, 
    the theorem is then obtained by Theorem~\ref{thm:quadratic_form_tuple_iso_algo}. 
\end{proof}

\bibliographystyle{alphaurl}
\bibliography{references}

\appendix
\section{Minimum congruent block diagonalization of alternating matrix tuples}\label{sec:appendix}

In this section, we give an alternating matrix tuple over $\F_2$ whose minimum congruent block diagonalization is not unique. 

We construct a $9 \times 9$ alternating matrix tuple $\mathbf{A}$ of length $6$ over $\F_2$ as follows.  Let $\mathbf{A}=\diag(\mathbf{A}_1, \mathbf{A}_2)$ with 
\[\mathbf{A}_1 = \begin{pmatrix}   & 100000 & 0100000    \\  100000 &  & 001000  \\  010000 & 001000
\end{pmatrix}\]
and 
\[\mathbf{A}_2 = \begin{pmatrix}   & 000100 & 000010 & & 100000 & 010000   \\  000100 &  & 000001 & 100000 &  & 001000\\  000010 & 000001 & & 010000 & 001000\\ 
 & 100000  & 010000  &   \\  100000 &  & 001000  \\  010000 & 001000 
\end{pmatrix},\]
where $a_i$ in the binary string $a_1a_2a_3a_4a_5a_6$ at position $(j, k)$ denotes the $(j, k)$th entry of the $i$-th matrix.
\ynote{Was: where the binary string $a_1a_2a_3a_4a_5a_6$ is the corresponding entry of the matrix tuple such that the $i$-th matrix being $a_i$ for any $1 \leq i \leq 6$.}
\gnote{Seems OK now, but please have a look.}

One can show that both $\mathbf{A}_1$ and $\mathbf{A}_2$ are minimum left-right blocks of $\mathbf{A}$, and thus also   minimum congruent blocks. 

On the other hand, 
with \[U = \begin{pmatrix} 1 & & & &  & & 1 \\ & 1 & & & & & & 1 \\ &  &  1  & & &  & & & 1\\ 1 & & & 1 \\ & 1 & & & 1 \\ & & 1 & & & 1 \\ & & & & & & 1\\
& & & & & & & 1 \\
& & & & & & & & 1\end{pmatrix},\]
$U^T \mathbf{A} U = \diag(\mathbf{
A}_1', \mathbf{A}_2')$ with 
\[\mathbf{A}_1' = \begin{pmatrix}   & 100000 & 0100000    \\  100000 &  & 001000  \\  010000 & 001000
\end{pmatrix}\]
and 
\[\mathbf{A}_2' = \begin{pmatrix}   & 100100 & 010010 & & 100000 & 010000   \\  100100 &  & 001001 & 100000 &  & 001000\\  010010 & 001001 & & 010000 & 001000\\ 
 & 100000  & 010000  &   \\  100000 &  & 001000  \\  010000 & 001000 
\end{pmatrix}.\]
Both $\mathbf{A}_1'$ and $\mathbf{A}_2'$ are minimum left-right  blocks of $\mathbf{A}$, and also minimum congruent blocks. 

Hence, $\mathbf{A}$ and $U^T \mathbf{A} U$ are congruent, and $\mathbf{A}_1$ and $\mathbf{A}_1'$ are the same. But one can verify that $\mathbf{A}_2$ and $\mathbf{A}_2'$ are not congruent. 
This suggests that the alternating matrix tuple $\mathbf{A}$ has multiple minimum congruent block diagonalizations. 
\end{document}